\documentclass[journal]{IEEEtran}

\usepackage{amsthm}
\usepackage{amsmath}
\usepackage{amssymb}
\usepackage{amsmath}
\usepackage{amssymb}
\usepackage{epsfig}
\usepackage{epsf}
\usepackage{subfigure}
\usepackage{graphicx}
\usepackage{url}
\usepackage[]{authblk}
\usepackage{cite}

\def\BibTeX{{\rm B\kern-.05em{\sc i\kern-.025em b}\kern-.08em
    T\kern-.1667em\lower.7ex\hbox{E}\kern-.125emX}}

\newtheorem{theorem}{Theorem}
\newtheorem{definition}{Definition}

\newtheorem{lemma}{Lemma}

\newtheorem{remark}{Remark}

\newcounter{numcount}

\newcommand{\cnt}{\Roman{numcount}\;\stepcounter{numcount}}

\begin{document}

\newcommand{\caseone} { { \nearrow } { \hspace{-3.8mm} \searrow } {\hspace{-3.85mm}  \raisebox{4.7pt}{{$\rightharpoonup$}}} {\hspace{-3.85mm}  \raisebox{-4.7pt}{{$\rightharpoondown$}}} }

\newcommand{\casetwo} { { \searrow } {\hspace{-3.85mm}  \raisebox{4.7pt}{{$\rightarrow$}}} {\hspace{-3.85mm}  \raisebox{-4.7pt}{{$\rightharpoondown$}}} }

\newcommand{\casethree} { { \nearrow } {\hspace{-3.85mm}  \raisebox{4.7pt}{{$\rightarrow$}}} {\hspace{-3.85mm}  \raisebox{-4.7pt}{{$\rightharpoondown$}}} }

\newcommand{\casefour} { { \raisebox{4.7pt}{{$\rightarrow$}}} {\hspace{-3.85mm}  \raisebox{-4.7pt}{{$\rightarrow$}}} }

\newcommand{\casefive} { { \raisebox{4.7pt}{{$\rightarrow$}}} }

\newcommand{\casesix} { { \searrow } {\hspace{-3.85mm}  \raisebox{4.7pt}{{$\rightarrow$}}}  }

\newcommand{\caseseven} { { \nearrow } {\hspace{-3.85mm}  \raisebox{4.7pt}{{$\rightharpoonup$}}} }

\newcommand{\caseeight} { { \nearrow } { \hspace{-3.8mm} \searrow } {\hspace{-3.85mm}  \raisebox{4.7pt}{{$\rightharpoonup$}}} }

\newcommand{\casefifteen} { { \nearrow } {\hspace{-3.85mm} \searrow } }


\title{Capacity Results for Binary Fading Interference\\ Channels with Delayed CSIT}

\author{Alireza~Vahid,
        Mohammad~Ali~Maddah-Ali,
        and~Amir~Salman~Avestimehr
        \thanks{Alireza Vahid is with the School of Electrical and Computer Engineering, Cornell University, Ithaca, NY, USA. Email: {\sffamily av292@cornell.edu}.}
        \thanks{A. Salman Avestimehr is with the Electrical Engineering Department, University of Southern California, Los Angeles, CA, USA. Email: {\sffamily avestimehr@ee.usc.edu}.}
\thanks{Mohammad~Ali~Maddah-Ali is with Bell Labs, Alcatel-Lucent, Holmdel, NJ, USA. Email: {\sffamily mohammadali.maddah-ali@alcatel-lucent.com}.}
\thanks{The work of A. S. Avestimehr and A. Vahid is in part supported by NSF Grants CAREER-0953117, CCF-1161720, NETS-1161904, AFOSR Young Investigator Program Award, and ONR award N000141310094.}
\thanks{Preliminary parts of this work were presented at the 2011 Allerton Conference on Communication, Control, and Computing~\cite{BFICAllerton}, and the 2012 International Symposium on Information Theory (ISIT)~\cite{BFICISIT2012}.}
\thanks{Copyright (c) 2014 IEEE. Personal use of this material is permitted.  However, permission to use this material for any other purposes must be obtained from the IEEE by sending a request to pubs-permissions@ieee.org.}
}

\maketitle


\begin{abstract}
To study the effect of lack of up-to-date channel state information at the transmitters (CSIT), we consider two-user binary fading interference channels with Delayed-CSIT. We characterize the capacity region for such channels under homogeneous assumption where channel gains have identical and independent distributions across time and space, eliminating the possibility of exploiting time/space correlation. We introduce and discuss several novel coding opportunities created by outdated CSIT that can enlarge the achievable rate region. The capacity-achieving scheme relies on accurate combination, concatenation, and merging of these opportunities, depending on the channel statistics. The outer-bounds are based on an extremal inequality we develop for a binary broadcast channel with Delayed-CSIT. We further extend the results and characterize the capacity region when output feedback links are available from the receivers to the transmitters in addition to the delayed knowledge of the channel state information. We also discuss the extension of our results to the non-homogeneous setting.
\end{abstract}


\section{Introduction}
\label{sec:introduction}

The history of studying the effect of feedback channel in communication systems traces back to Shannon~\cite{Shannon}, and ever since, there have been extensive efforts to discover new techniques that exploit feedback channels in order to benefit wireless networks. In today's wireless networks, one of the main objectives in utilizing feedback channels is to provide the transmitters with the knowledge of the channel state information (CSI). In slow-fading networks, this task could have been carried on with negligible overhead. However, as wireless networks started growing in size, as mobility became an inseparable part of networks, and as fast-fading networks started playing a more important role, the availability of up-to-date channel state information at the transmitters (CSIT) has become a challenging task to accomplish. Specifically, in fast-fading scenarios, the coherence time of the channel is smaller than the delay of the feedback channel, and thus, providing the transmitters with up-to-date channel state information is practically infeasible.

As a result, there has been a recent growing interest in studying the effect of lack of up-to-date channel state information at the transmitters in wireless networks. In particular, in the context of multiple-input single-output (MISO)  broadcast channels (BC), it was recently  shown that even completely stale CSIT (a.k.a. Delayed-CSIT) can still be very useful and can change the scale of the capacity, measured by the degrees of freedom (DoF)~\cite{Maddah-Tse-Allerton}. A key idea behind the scheme proposed in~\cite{Maddah-Tse-Allerton} is that instead of predicting future channel state information, transmitters should focus on the side-information provided in the past signaling stages via the feedback channel, and try to create signals that are of \emph{common interest} of multiple receivers. Hence, we can increase spectral efficiency by retransmission of such signals of common interest. These ideas were later extended to derive constant-gap approximation of the capacity region of the MISO BCs with Delayed-CSIT~\cite{MISOBCAllerton,MISOBCISIT}.

There have also been several recent works in the literature on wireless networks with distributed transmitters and Delayed-CSIT. This includes the study of the DoF region of multi-antenna two-user Gaussian IC and X channel~\cite{GhasemiX1,Vaze_DCSIT_MIMO_BC}, $k$-user Gaussian IC and X channel~\cite{Jafar_Retrospective,abdoli2011degrees}, and multi-antenna two-user Gaussian IC with Delayed-CSIT and Shannon feedback~\cite{tandon2012degrees,vaze2011degrees}. In particular, the DoF region of multi-antenna two-user Gaussian IC has been characterized in~\cite{vaze2012degrees}, and it has been shown that the $k$-user Gaussian IC and X channels can still achieve more than one DoF with Delayed-CSIT~\cite{Jafar_Retrospective,Abdoli2011IC-X-Arxiv} (for $k>2$). 

A major challenge that arises in interference channels with Delayed-CSIT is that in such networks the transmitter has no longer access to all transmit signals in the network. In fact, each transmitter has \emph{only} access to its own interference contribution. Therefore, unlike broadcast channels in which the task of creating signals of common interest could be simply done at a single transmitter that has access to \emph{all} messages, exploiting Delayed-CSIT becomes much more challenging. This issue has become a major challenge both in deriving achievablility strategies and tight outer-bounds. In order to shed light on fundamental limits of communications with Delayed-CSIT in interference channels, in this paper we focus on a binary fading model as described below. 

We consider a two-user interference channel as illustrated in Fig.~\ref{fig:ChannelModel}. In this network, the channel gains at each time instant are either $0$ or $1$ according to some Bernoulli distribution, and are independent from each other and \emph{over time}. The input and output signals are also in the binary field and if two signals arrive simultaneously at a receiver, then the receiver obtains the exclusive OR (XOR) of them. We shall refer to this network as the two-user Binary Fading Interference Channel (BFIC).  

\begin{figure}[h]
\centering
\includegraphics[height = 3.5cm]{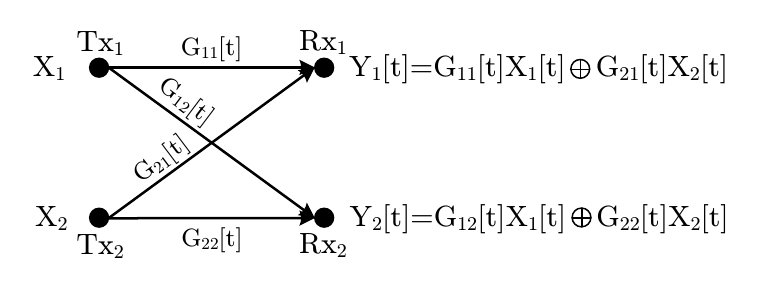}
\caption{Binary fading model for a two-user interference channel. The channel gains, the transmit signals and the received signals are in the binary field. The channel gains are distributed as i.i.d. Bernoulli random variables. The channel gains are independent across time so that the transmitters cannot predict future based on the past channel state information.\label{fig:ChannelModel}}
\end{figure}

As the main motivation, we study the two-user BFIC as a stepping stone towards understanding the capacity of more complicated fading interference channels with Delayed-CSIT. Lately, the linear deterministic model introduced in~\cite{ADT10}, has been utilized to translate the results from deterministic networks into Gaussian networks (\emph{e.g.}, \cite{ADT10,Guy, Guy2,Suh, vahid2010two,AlirezaFB,avestimehr2010capacity,Ave3}). In the linear deterministic model, there is a non-negative integer representing the channel gain from a transmitter to a receiver. Hence, one can view the binary fading model as a fading interpretation of the linear deterministic model where the non-negative integer associated to each link is either $0$ or $1$. Furthermore, as demonstrated in~\cite{vahid2013communication}, the binary fading model provides a simple, yet useful physical layer abstraction for wireless packet networks in which whenever a collision occurs, the receiver can store its received analog signal and utilize it for decoding the packets in future (for example, by successive interference cancellation techniques).


In this work, we fully characterize the capacity region of the two-user BFIC with Delayed-CSIT. We introduce and discuss several novel coding opportunities, created by outdated CSIT, which can enlarge the achievable rate region. In particular, we propose a new transmission strategy, which is carried on over several phases. Each channel realization creates multiple coding opportunities which can be exploited in the next phases, to improve the rate region. However, we observe that \emph{merging} or \emph{concatenating} some of the opportunities can offer even more gain. To achieve the capacity region, we find the most efficient arrangement of combination, concatenation, and merging of the opportunities, depending on the channel statistics. This can take up to five phases of communication for a two-user channel. For converse arguments, we start with a genie-aided interference channel and show that the problem can be reduced to some particular form of broadcast channels with Delayed-CSIT. We establish a new extremal inequality for the underlying BC that leads to a tight outer-bound for the original interference channel. The established inequality provides an outer-bound on how much the transmitter in a BC can favor one receiver to the other using Delayed-CSIT (in terms of the entropy of the received signal at the two receivers). 

We also consider the scenario in which output feedback links are available from the receivers to the transmitters on top of the delayed knowledge of the channel state information. We demonstrate how output feedback can be utilized to further improve the achievable rates in terms of both enlarging the capacity region and improving the achievable sum-rate.  In addition, output feedback can help us simplify the achievability strategy. For converse, again the core idea is to reduce the problem to a broadcast channel with Delayed-CSIT and output feedback, and establishing a new extremal inequality for the resultant broadcast channel. The inequality then helps us prove a tight outer-bound for the original interference channel. 

Our contributions are therefore multi-fold. We develop several new coding opportunities for the BFIC with Delayed-CSIT, as well as when the transmitters have access to perfect instantaneous CSIT. We then develop a framework to demonstrate how to combine and merge these coding opportunities in an optimal fashion based on the channel statistics. For converse, we develop an extremal entropy inequality that captures the effect of Delayed-CSIT at the transmitters. Using this extremal entropy inequality, we derive an outer-bound that matches our achievability.

The rest of the paper is organized as follows. In Section~\ref{sec:problem}, we formulate our problem. In Section~\ref{sec:results}, we present our main results and illustrate them through an example. We then provide an overview of our main achievability and converse techniques in Section~\ref{sec:opportunities}. Sections~\ref{sec:achallhalf}-\ref{sec:ConvInstFB} are dedicated to the proof of our main results. In Section~\ref{sec:extension}, we discuss how our results could be extended to more general settings. Section~\ref{sec:conclusion} concludes the paper and describes several interesting future directions.



\section{Problem Setting}
\label{sec:problem}


We consider the two-user Binary Fading Interference Channel as illustrated in Fig.~\ref{fig:detIC} and defined below. 

\begin{definition}
The two-user Binary Fading Interference Channel includes two transmitter-receiver pairs in which the channel gain from transmitter ${\sf Tx}_i$ to receiver ${\sf Rx}_j$ at time instant $t$ is denoted by $G_{ij}[t]$, $i,j \in \{1,2\}$. The channel gains are either $0$ or $1$ (\emph{i.e.} $G_{ij}[t] \in \{0,1\}$), and they are distributed as independent Bernoulli random variables (independent from each other and \emph{over time}). We consider the homogeneous setting where 
\begin{align}
G_{ij}[t] \overset{d}\sim \mathcal{B}(p), \qquad i,j = 1,2,
\end{align}
for $0 \leq p \leq 1$, and we define $q \overset{\triangle}= 1 - p$. 

At each time instant $t$, the transmit signal at ${\sf Tx}_i$ is denoted by $X_i[t] \in \{ 0, 1 \}$, and the received signal at ${\sf Rx}_i$ is given by
\begin{equation} 
\label{eq:receivedsignal}
Y_i[t] = G_{ii}[t] X_i[t] \oplus G_{\bar{i}i}[t] X_{\bar{i}}[t], \quad i = 1, 2,
\end{equation}
where the summation is in $\mathbb{F}_2$. 
\end{definition}

\begin{definition}
We define the channel state information (CSI) at time instant $t$ to be the quadruple 
\begin{align}
G[t] \overset{\triangle}= (G_{11}[t], G_{12}[t], G_{21}[t], G_{22}[t]).
\end{align}
\end{definition}

\begin{figure}[ht]
\centering
\includegraphics[height = 3cm]{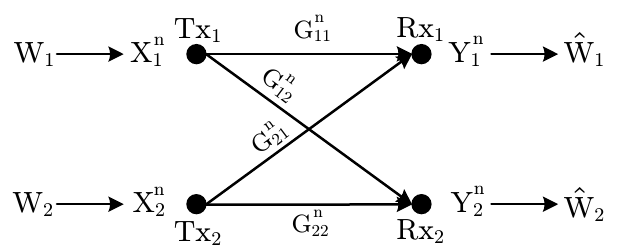}
\caption{Two-user Binary Fading Interference Channel (BFIC). The channel gains, the transmit and the received signals are in the binary field. The channel gains are distributed as i.i.d. Bernoulli random variables.\label{fig:detIC}}
\end{figure}

We use the following notations in this paper. We use capital letters to denote random variables (RVs), \emph{e.g.}, $G_{ij}[t]$ is a random variable at time instant $t$. Furthermore for a natural number $k$, we set 
\begin{align}
G^{k} \overset{\triangle}= \left[ G[1], G[2], \ldots, G[k] \right]^{\top}.
\end{align}

Finally, we set 
{\small \begin{align}
& G_{ii}^t X_i^t \oplus G_{\bar{i}i}^t X_{\bar{i}}^t \\
&~\overset{\triangle}= \left[ G_{ii}[1] X_i[1] \oplus G_{\bar{i}i}[1] X_{\bar{i}}[1], \ldots, G_{ii}[t] X_i[t] \oplus G_{\bar{i}i}[t] X_{\bar{i}}[t] \right]^{\top}. \nonumber 
\end{align}}
In this paper, we consider three models for the available channel state information at the transmitters:
\begin{enumerate}

\item Instantaneous-CSIT: In this model, the channel state information $G^t$ is available at each transmitter at time instant $t$, $t=1,2,\ldots,n$;

\item No-CSIT: In this model, transmitters only know the distribution from which the channel gains are drawn, but not the actual realizations of them;

\item Delayed-CSIT: In this model, at time instant $t$, each transmitter has the knowledge of the channel state information up to the previous time instant (\emph{i.e.} $G^{t-1}$) and the distribution from which the channel gains are drawn, $t=1,2,\ldots,n$.

\end{enumerate}

We assume that the receivers have instantaneous knowledge of the CSI. We consider the scenario in which ${\sf Tx}_i$ wishes to reliably communicate message $\hbox{W}_i \in \{ 1,2,\ldots,2^{n R_i}\}$ to ${\sf Rx}_i$ during $n$ uses of the channel, $i = 1,2$. We assume that the messages and the channel gains are {\it mutually} independent and the messages are chosen uniformly. For each transmitter ${\sf Tx}_i$, let message $\hbox{W}_i$ be encoded as $X_i^n$ using the encoding function $f_i(.)$, which depends on the available CSI at ${\sf Tx}_i$. Receiver ${\sf Rx}_i$ is only interested in decoding $\hbox{W}_i$, and it will decode the message using the decoding function $\widehat{\hbox{W}}_i = g_i(Y_i^n,G^n)$. An error occurs when $\widehat{\hbox{W}}_i \neq \hbox{W}_i$. The average probability of decoding error is given by
\begin{equation}
\label{}
\lambda_{i,n} \overset{\triangle}= \mathbb{E}[P[\widehat{\hbox{W}}_i \neq \hbox{W}_i]], \hspace{5mm} i = 1, 2,
\end{equation}
and the expectation is taken with respect to the random choice of the transmitted messages $\hbox{W}_1$ and $\hbox{W}_2$. A rate tuple $(R_1,R_2)$ is said to be achievable, if there exists encoding and decoding functions at the transmitters and the receivers respectively, such that the decoding error probabilities $\lambda_{1,n},\lambda_{2,n}$ go to zero as $n$ goes to infinity. The capacity region is the closure of all achievable rate tuples.

In addition to the setting described above, we consider a separate scenario in which an output feedback (OFB) link is available from each receiver to its corresponding transmitter\footnote{As we will see later, our result holds for the case in which output feedback links are available from each receiver to \emph{both} transmitters.}. More precisely, we consider a noiseless feedback link of infinite capacity from each receiver to its corresponding transmitter. 

Due to the presence of output feedback links, the encoded signal $X_i[t]$ of transmitter ${\sf Tx}_i$ at time $t$, would be a function of its own message, previous output sequence at its receiver, and the available CSIT. For instance, with Delayed-CSIT and OFB, we have
\begin{equation}
\label{}
X_i[t] =  f_i[t](\hbox{W}_i,Y_{i}^{t-1},G^{t-1}), \hspace{5mm} i = 1, 2.
\end{equation}

As stated in the introduction, our goal is to understand the effect of the channel state information and the output feedback, on the capacity region of the two-user Binary Fading Interference Channel. Towards that goal, we consider several scenarios about the availability of the CSIT and the OFB. For all scenarios, we provide exact characterization of the capacity region. In the next section, we present the main results of the paper. 



\section{Statement of the Main Results}
\label{sec:results}


In this paper, we focus on the following scenarios about the availability of the CSI and the OFB: $(1)$ Delayed-CSIT and no OFB; $(2)$ Delayed-CSIT and OFB; and $(3)$ Instantaneous-CSIT and OFB. In order to illustrate the results, we first establish the capacity region of the two-user BFIC with No-CSIT and Instantaneous-CSIT as our benchmarks.

\subsection{Benchmarks}

Our baseline is the scenario in which there is no output feedback link from the receivers to the transmitters, and we assume the No-CSIT model. In other words, the only available knowledge at the transmitters is the distribution from which the channel gains are drawn. In this case, it is easy to see that for any input distribution, the two received signals are \emph{statistically} the same, hence
\begin{align}
& I\left( X_1^n ; Y_1^n | G^n \right) = I\left( X_1^n ; Y_2^n | G^n \right), \nonumber \\
& I\left( X_2^n ; Y_1^n | G^n \right) = I\left( X_2^n ; Y_2^n | G^n \right).
\end{align}

Therefore, the capacity region in this case, $\mathcal{C}^{\mathrm{No-CSIT}}$, is the same as the intersection of the capacity region of the multiple-access channels (MACs) formed at the receivers:
\begin{equation}
\label{eq:RegionNoCSIT}
\mathcal{C}^{\mathrm{No-CSIT}} =
\left\{ \begin{array}{ll}
\vspace{1mm} 0 \leq R_i \leq p, & i = 1,2, \\
R_1 + R_2 \leq 1 - q^2. &
\end{array} \right.
\end{equation}

The other extreme point on the available CSIT is the Instantaneous-CSIT model. The capacity region in this case is given in the following theorem which is proved in Appendices~\ref{sec:AchInst} and~\ref{sec:ConvInst}.

\begin{theorem}
\label{THM:IC-InstCSIT}
{\bf [Capacity Region with Instantaneous-CSIT]} The capacity region of the two-user Binary Fading IC with Instantaneous-CSIT (and no output feedback), $\mathcal{C}^{\mathrm{ICSIT}}$, is the set of all rate tuples $\left( R_1, R_2 \right)$ satisfying
\begin{equation}
\label{eq:fullNSIregion}
\mathcal{C}^{\mathrm{ICSIT}} =
\left\{ \begin{array}{ll}
\vspace{1mm} 0 \leq R_i \leq p, & i = 1,2, \\
R_1 + R_2 \leq 1 - q^2 + p q.  &  \\
\end{array} \right.
\end{equation}
\end{theorem}

\begin{remark}
Comparing the capacity region of the two-user BFIC with No-CSIT~(\ref{eq:RegionNoCSIT}) and Instantaneous-CSIT~(\ref{eq:fullNSIregion}), we observe that the bounds on individual rates remain unchanged while the sum-rate outer-bound is increased by $pq$. This increase can be intuitively explained as follows. The outer-bound of $1-q^2$ corresponds to the fraction of time in which at least one of the links to each receiver is equal to $1$. Therefore, this outer-bound corresponds to the fraction of time that each receiver gets ``useful'' signal. This is tight with No-CSIT since each receiver should be able to decode both messages. However, once we move to Instantaneous-CSIT, we can send a private message to one of the receivers by using those time instants in which the link from the corresponding transmitter to that receiver is equal to $1$, but that transmitter is not interfering with the other receiver. This corresponds to $pq$ fraction of the time.
\end{remark}

Now that we have covered the benchmarks, we are ready to present our main results. 

\subsection{Main Results}

As the first step, we consider the Delayed-CSIT model. In this case, the following theorem establishes our result.

\begin{theorem}
\label{THM:IC-DelayedCSIT}
{\bf [Capacity Region with Delayed-CSIT]} The capacity region of the two-user Binary Fading IC with Delayed-CSIT (and no output feedback), $\mathcal{C}^{\mathrm{DCSIT}}$, is the set of all rate tuples $\left( R_1, R_2 \right)$ satisfying
\begin{equation}
\label{eq:DelayedNSIregion}
\mathcal{C}^{\mathrm{DCSIT}} =
\left\{ \begin{array}{ll}
\vspace{1mm} 0 \leq R_i \leq p, &  i = 1,2, \\
R_i + \left( 1 + q \right) R_{\bar{i}} \leq p \left( 1 + q \right)^2, & i = 1,2.
\end{array} \right.
\end{equation}
\end{theorem}

\begin{remark}
Comparing the capacity region of the two-user BFIC with Delayed-CSIT~(\ref{eq:DelayedNSIregion}) and Instantaneous-CSIT~(\ref{eq:fullNSIregion}), we can show that for $0 \leq p \leq \left( 3 - \sqrt{5} \right)/2$, the two regions are equal. However, for $\left( 3 - \sqrt{5} \right)/2 < p < 1$, the capacity region of the two-user BFIC with Delayed-CSIT is strictly smaller than that of Instantaneous-CSIT. Moreover, we can show that the capacity region of the two-user BFIC with Delayed-CSIT is strictly larger than that of No-CSIT (except for $p=0$ or $1$).
\end{remark}

Furthermore, since the channel state information is acquired through the feedback channel, it is also important to understand the effect of output feedback on the capacity region of the two-user BFIC with Delayed-CSIT. In the study of feedback in wireless networks, one other direction is to consider the transmitter cooperation created through the output feedback links. In this context, it is well-known that feedback does not increase the capacity of discrete memoryless point-to-point channels~\cite{Shannon}. However, feedback can enlarge the capacity region of multi-user networks, even in the most basic case of the two-user memoryless multiple-access channel~\cite{Gaar, Oza}. In~\cite{Suh,AlirezaFB}, the feedback capacity of the two-user Gaussian IC has been characterized to within a constant number of bits. One consequence of these results is that output feedback can provide an unbounded capacity increase. This is in contrast to point-to-point and multiple-access channels where feedback provides no gain and bounded gain respectively. In this work, we consider the scenario in which an output feedback link is available from each receiver to its corresponding transmitter on top of the delayed knowledge of the channel state information as depicted in Fig.~\ref{fig:detICFB}(a).
\begin{figure}[ht]
\centering
\subfigure[]{\includegraphics[height = 4cm]{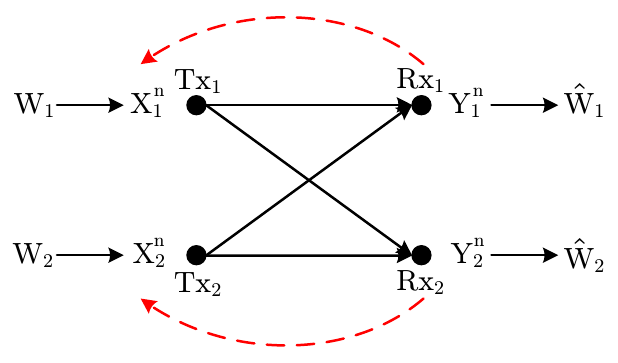}}
\hspace{0.75 in}
\subfigure[]{\includegraphics[height = 4cm]{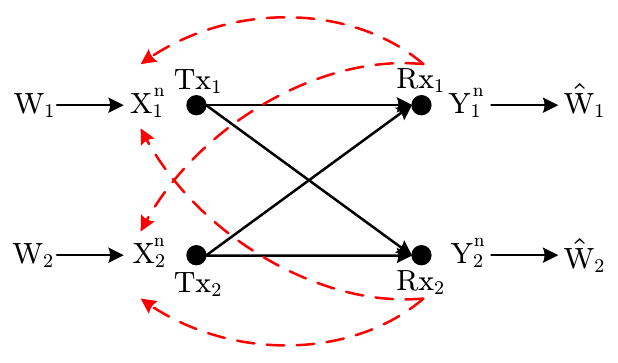}}
\caption{\it Two-user Binary Fading Interference Channel: $(a)$ with output feedback links from each receiver to its corresponding transmitter. In this setting, the transmit signal of ${\sf Tx}_i$ at time instant $t$, would be a function of the message $\hbox{W}_i$, the available CSIT, and the output sequences $Y_i^{t-1}$, $i=1,2$; and $(b)$ with output feedback links from each receiver to both transmitters. In this setting, the transmit signal of ${\sf Tx}_i$ at time instant $t$, would be a function of the message $\hbox{W}_i$, the available CSIT, and the output sequences $Y_1^{t-1}, Y_2^{t-1}$, $i=1,2$.\label{fig:detICFB}}
\end{figure}


In the presence of output feedback and Delayed-CSIT, we have the following result.

\begin{theorem}
\label{THM:IC-DCSITFB}
{\bf [Capacity Region with Delayed-CSIT and OFB]} For the two-user binary IC with Delayed-CSIT and OFB, the capacity region $\mathcal{C}^{\mathrm{DCSIT,OFB}}$, is given by
\begin{align}
\label{eq:capacity-FB}
& \mathcal{C}^{\mathrm{DCSIT,OFB}} = \\
& \left\{ R_1, R_2 \in \mathbb{R}^+~s.t.~R_i + (1+q) R_{\bar{i}} \leq p(1+q)^2,~i=1,2 \right\}. \nonumber
\end{align}
\end{theorem}

\begin{figure*}[ht]
\centering
\subfigure[]{\includegraphics[height = 4.15cm]{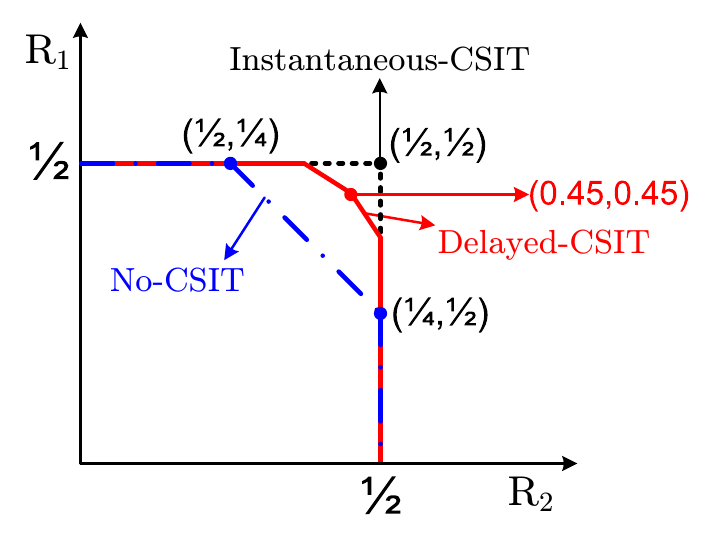}}
\hspace{0.25 in}
\subfigure[]{\includegraphics[height = 4.15cm]{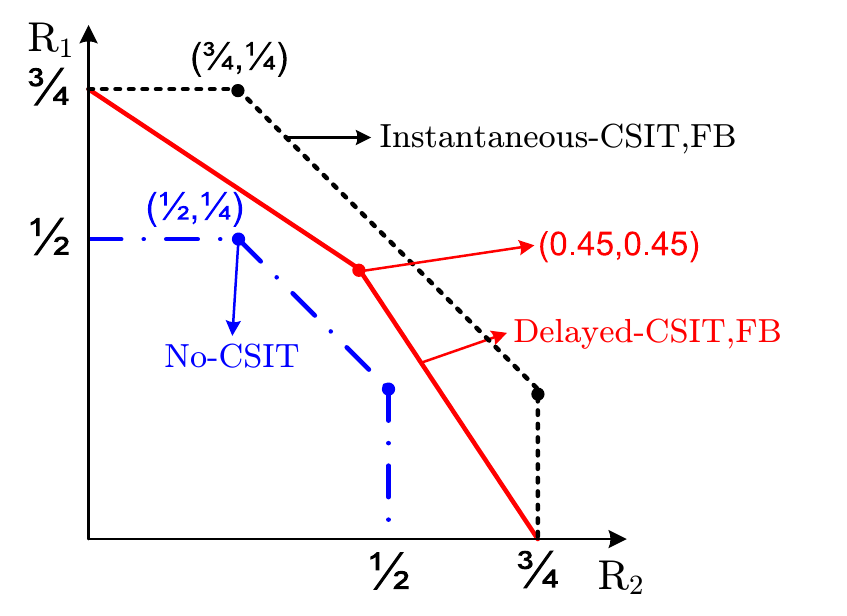}}
\hspace{0.25 in}
\subfigure[]{\includegraphics[height = 4.15 cm]{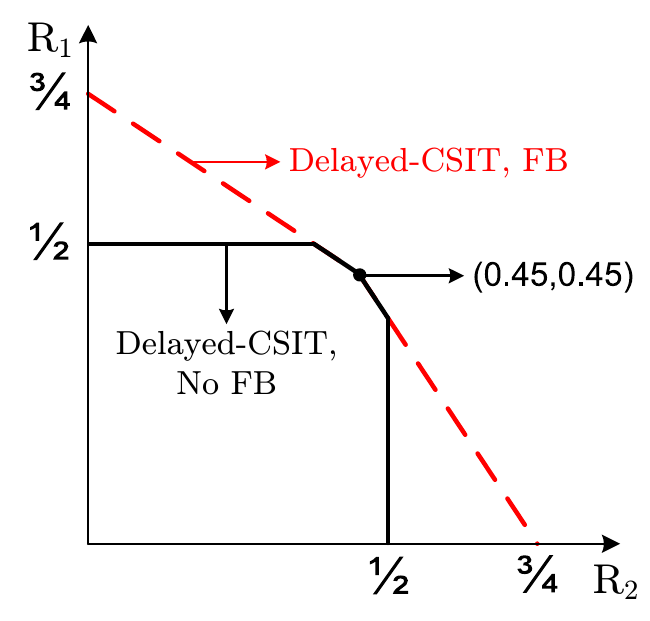}}
\caption{\it Two-user Binary Fading IC: $(a)$ the capacity region with No-CSIT, Delayed-CSIT, and Instantaneous-CSIT, without OFB; $(b)$ the capacity region with No-CSIT, Delayed-CSIT, and Instantaneous-CSIT, with OFB; and $(c)$ the capacity region with Delayed-CSIT, with  and without output feedback.\label{fig:illustration}}
\end{figure*}

\begin{remark}
\label{remark:OFBConverse}
The outer-bound on the capacity region with only Delayed-CSIT~(\ref{eq:DelayedNSIregion}) is in fact the intersection of the outer-bounds on the individual rates (\emph{i.e.} $R_i \leq p$, $i=1,2$) and the capacity region with Delayed-CSIT and OFB~(\ref{eq:capacity-FB}). Therefore, the effect of OFB is to remove the constraints on individual rates. This can be intuitively explained by noting that OFB creates a new path to flow information from each transmitter to its corresponding receiver (\emph{e.g.}, ${\sf Tx}_1 \rightarrow {\sf Rx}_2 \rightarrow {\sf Tx}_2 \rightarrow {\sf Rx}_1$). This opportunity results in elimination of the individual rate constraints in this case.
\end{remark}

\begin{remark}
As we will see in Section~\ref{sec:conversedelayedhalfFB}, same outer-bounds hold in the presence of \emph{global output feedback} where output feedback links are available from each receiver to both transmitters, see Fig.~\ref{fig:detICFB}(b). Therefore, the capacity region of two user binary IC with Delayed-CSIT and global output feedback is the same as the capacity region described in~(\ref{eq:capacity-FB}). This implies that in this case, global output feedback does not provide new coding opportunities, nor does it enhance the existing ones. Similar observation has been made in the context of MIMO Interference Channels~\cite{tandon2012degrees,vaze2011degrees}, even though the coding opportunities in Binary IC and MIMO IC are not the same.
\end{remark}

Finally, we present our result for the case of Instantaneous-CSIT and output feedback. Note that in this scenario, although transmitters have instantaneous knowledge of the channel state information, the output signals are available at the transmitters with unit delay. This scenario corresponds to a slow-fading channel where output feedback links are available from the receivers to the transmitters.

\begin{theorem}
\label{THM:IC-ICSITFB}
{\bf [Capacity Region with Instantaneous-CSIT and OFB]} For the two-user binary IC with Instantaneous-CSIT and OFB, the capacity region $\mathcal{C}^{\mathrm{ICSIT,OFB}}$, is the set of all rate tuples $(R_1,R_2)$ satisfying
\begin{equation}
\label{eq:capacity-FBInst}
\mathcal{C}^{\mathrm{ICSIT,OFB}} =
\left\{ \begin{array}{ll}
\vspace{1mm} 0 \leq R_i \leq 1-q^2, &  i=1,2,\\
R_1 + R_2 \leq 1 - q^2 + p q.  &  \\
\end{array} \right.\end{equation}
\end{theorem}

\begin{remark}
Comparing the capacity region of the two-user BFIC with Instantaneous-CSIT, with OFB~(\ref{eq:capacity-FBInst}) and without OFB~(\ref{eq:fullNSIregion}), we observe that the outer-bound on the sum-rate remains unchanged. However, the bounds on individual rates are further increased to $1-q^2$. Similar to the previous remark, this is again due to the additional communication path provided by OFB from each transmitter to its intended receiver. However, since the outer-bound on sum-rate with Instantaneous-CSIT and OFB~(\ref{eq:capacity-FBInst}) is higher than that of Delayed-CSIT and OFB~(\ref{eq:capacity-FB}), the bounds on individual rates cannot be eliminated.
\end{remark}

The proof of the results is organized as follows. The proof of Theorem~\ref{THM:IC-DelayedCSIT} is presented in Sections~\ref{sec:achallhalf} and~\ref{sec:outerHalf}. The proof of Theorem~\ref{THM:IC-DCSITFB} is presented in Sections~\ref{sec:AchDelayedFB} and~\ref{sec:conversedelayedhalfFB}, and finally, the proof of Theorem~\ref{THM:IC-ICSITFB} is presented in Sections~\ref{sec:AchInstFB} and~\ref{sec:ConvInstFB}. We end this section by illustrating our main results via an example in which $p = 0.5$.

\subsection{Illustration of the Main Results for $p = 0.5$}

For this particular value of the channel parameter, the capacity region with Delayed-CSIT and Instantaneous-CSIT with or without output feedback is given in Table~\ref{table:regions}, and Fig.~\ref{fig:illustration} illustrates the results presented in this table. We notice the following remarks.
\begin{table}[h]
\caption{Illustration of our main results through an example in which $p = 0.5$.}
\centering
\begin{tabular}{| c | c | c |}
\hline
      		 &  Capacity Region     &  Capacity Region \\
					 &  with Delayed-CSIT   &  with Instantaneous-CSIT  \\ [0.5ex]

\hline

\raisebox{15pt}{No-OFB}   & \raisebox{15pt}{$ \left\{ \begin{array}{ll}  \vspace{1mm} R_i \leq \frac{1}{2}  & \\ R_i + \frac{3}{2} R_{\bar{i}} \leq \frac{9}{8}  &  \end{array} \right. $} & \raisebox{15pt}{$ \left\{ \begin{array}{ll}  \vspace{1mm} R_1 \leq \frac{1}{2}  & \\ R_2 \leq \frac{1}{2}  &  \end{array} \right. $} \\

\hline

\raisebox{15pt}{OFB}    & \raisebox{15pt}{$ \left\{ \begin{array}{ll}  \vspace{1mm}  R_1 + \frac{3}{2} R_{2} \leq \frac{9}{8} & \\ \frac{3}{2} R_1 + R_{2} \leq \frac{9}{8}  &  \end{array} \right. $} & \raisebox{15pt}{$ \left\{ \begin{array}{ll}  \vspace{1mm} R_i \leq \frac{3}{4}  & \\ R_1 + R_2 \leq 1  &  \end{array} \right. $} \\

\hline

\end{tabular}
\label{table:regions}
\end{table}

\begin{remark}
Note that for $p=0.5$, we have
$$\mathcal{C}^{\mathrm{No-CSIT}} \subset \mathcal{C}^{\mathrm{DCSIT}} \subset \mathcal{C}^{\mathrm{ICSIT}}.$$

In other words, the capacity region with Instantaneous-CSIT is strictly larger than that of Delayed-CSIT, which is in turn strictly larger than the capacity region with No-CSIT. Moreover, we have
$$\mathcal{C}^{\mathrm{DCSIT,OFB}} \subset \mathcal{C}^{\mathrm{ICSIT,OFB}},$$
meaning that the instantaneous knowledge of the CSIT enlarges the capacity region of the two-user BFIC with OFB compared to the case of Delayed-CSIT.
\end{remark}

\begin{remark}
\label{remark:comparison}
In Fig.~\ref{fig:illustration}(c), we have illustrated the capacity region with Delayed-CSIT, with and without output feedback. First, we observe that OFB enlarges the capacity region. Second, we observe that the optimal sum-rate point is the same for $p =0.5$. However, this is not always the case. In fact, for some values of $p$, output feedback can even increase the optimal sum-rate. Using the results of Theorem~\ref{THM:IC-DelayedCSIT} and Theorem~\ref{THM:IC-DCSITFB}, we have plotted the sum-rate capacity of the two-user Binary Fading IC with and without OFB for the Delayed-CSIT model in Fig.~\ref{fig:woFB}. Note that for $0 < p < \left( 3 - \sqrt{5} \right)/2$, the sum-rate capacity with OFB is strictly larger than the no OFB scenario. 
\end{remark}

\begin{figure}[h]
\centering
\includegraphics[height = 5 cm]{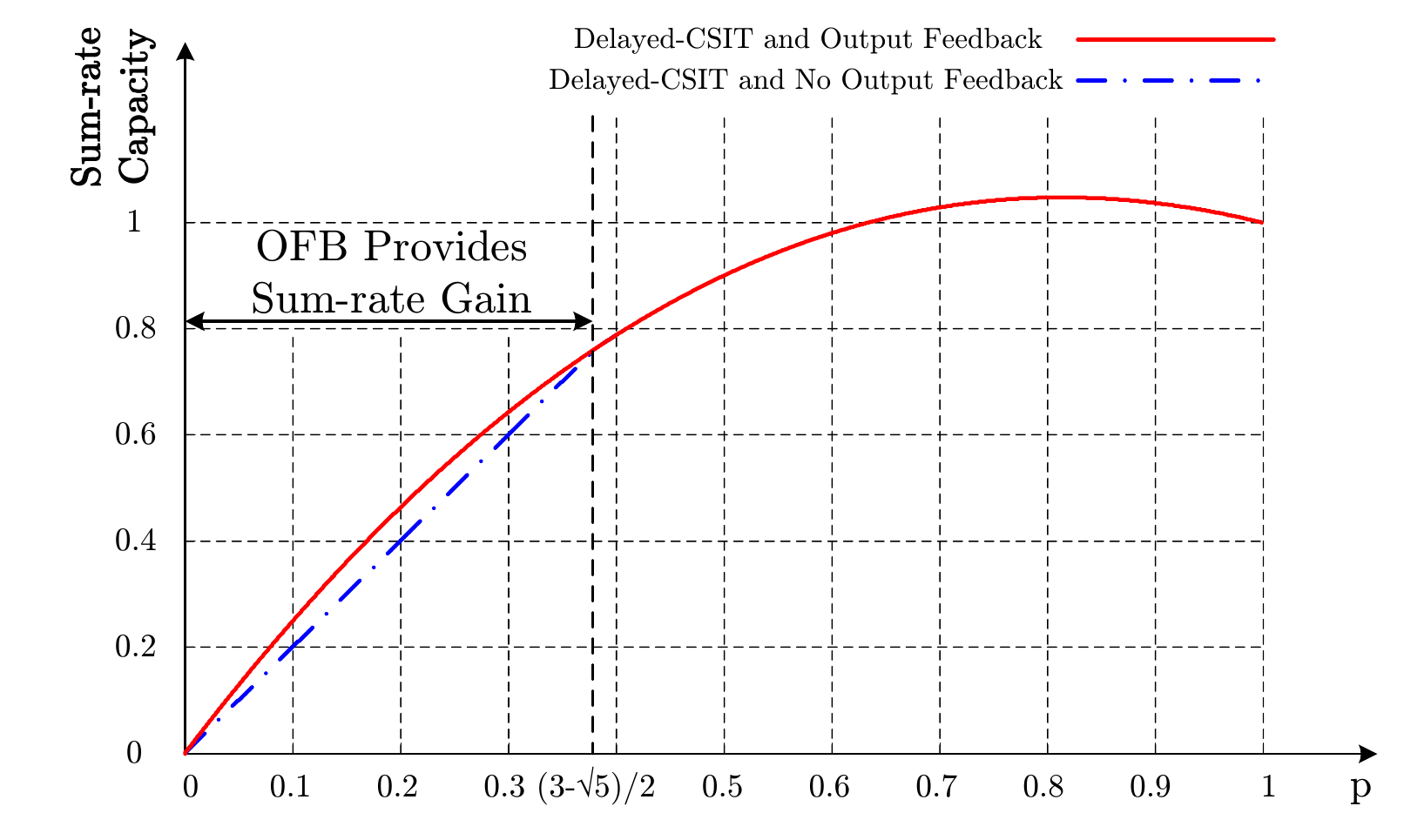}
\caption{The sum-rate capacity of the two-user BFIC with Delayed-CSIT, with and without output feedback. For $0 < p < \left( 3 - \sqrt{5} \right)/2$, the sum-rate capacity with OFB is strictly larger than the scenario where no OFB is available.\label{fig:woFB}}
\end{figure}

\begin{remark}
Comparing the capacity region of the two-user BFIC with Instantaneous-CSIT, with OFB~(\ref{eq:capacity-FBInst}) and without OFB~(\ref{eq:fullNSIregion}), we observe that OFB enlarges the capacity region. Moreover, similar to the Delayed-CSIT scenario, the optimal sum-rate point is the same for $p =0.5$. Again, this is not always the case. In fact, for $0 < p < 0.5$, output feedback can even increase the optimal sum-rate. Using the results of Theorem~\ref{THM:IC-InstCSIT} and Theorem~\ref{THM:IC-ICSITFB}, we have plotted the sum-rate capacity of the two-user Binary Fading IC with and without OFB in Fig.~\ref{fig:woFBInst}.
\end{remark}

\begin{figure}[h]
\centering
\includegraphics[height = 5 cm]{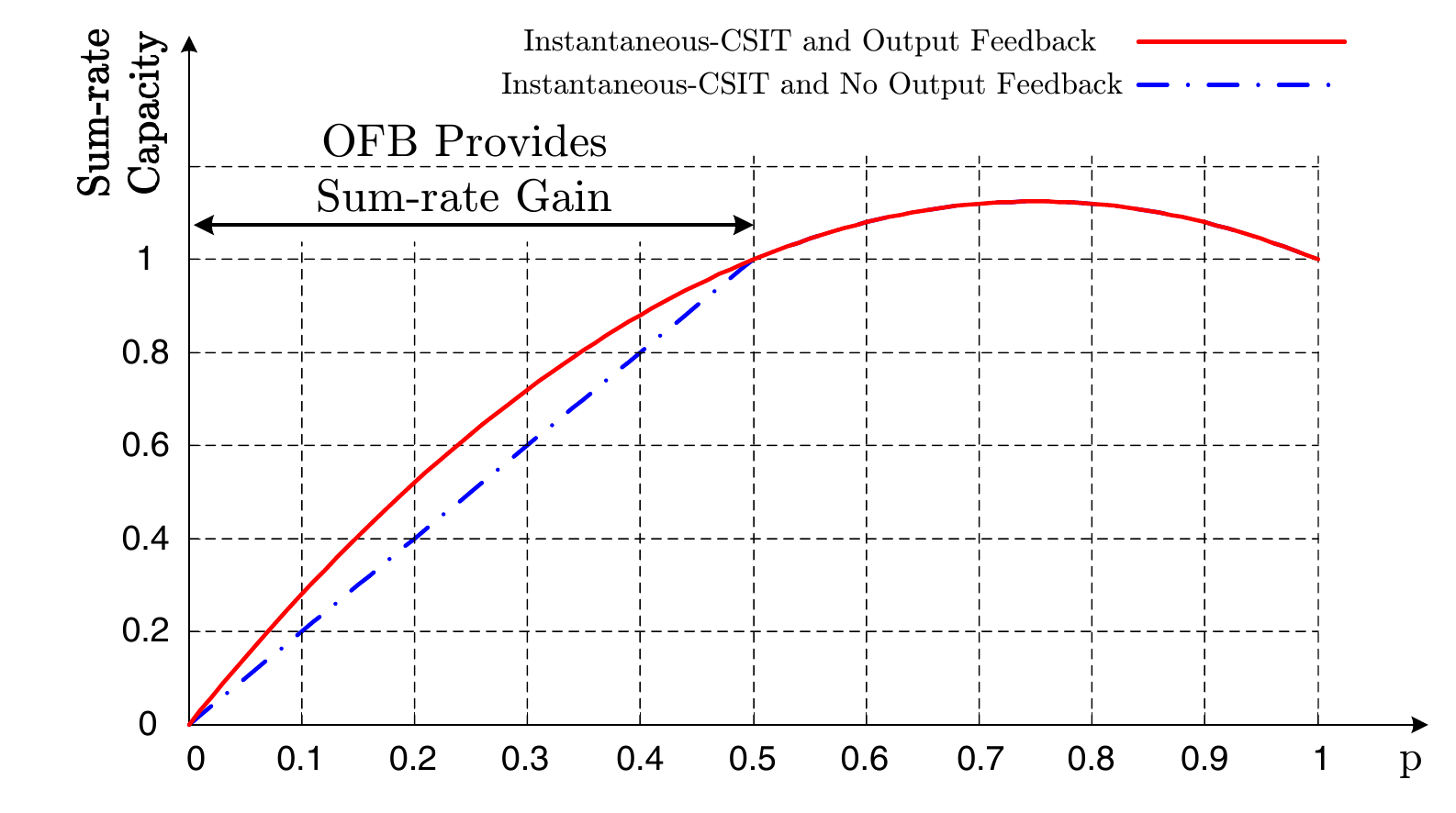}
\caption{The sum-rate capacity of the two-user BFIC with Instantaneous-CSIT, with and without output feedback. For $0 < p < 0.5$, the sum-rate capacity with OFB is strictly larger than the scenario where no OFB is available.\label{fig:woFBInst}}
\end{figure}

\begin{remark}
In Fig.~\ref{fig:woFB} and Fig.~\ref{fig:woFBInst}, we have identified the range of $p$ for which output feedback provides sum-rate gain. Basically, when the sum-rate capacity without OFB is dominated by the capacity of the direct links (\emph{i.e.} $2p$), the additional communication paths created by the means of output feedback links help increase the optimal sum-rate. 
\end{remark}

\begin{remark}
While our capacity results in Theorem~\ref{THM:IC-DelayedCSIT} and Theorem~\ref{THM:IC-DCSITFB} are for binary fading interference channels, in~\cite{vahid2013communication}, we have shown how they can also be utilized to obtain capacity results for a class of wireless packet networks.
\end{remark}

In the following section, we present the main ideas that we incorporate in this paper.



\section{Overview of the Key Ideas}
\label{sec:opportunities}


Our goal in this section is to present the key techniques we use in this paper both for achievability and converse purposes. Although we will provide detailed explanation of the achievability strategy and converse proofs for all different scenarios, we found it instructive to elaborate the main ideas through several clarifying examples. Furthermore, the coding opportunities introduced in this section can be applicable to DoF analysis of wireless networks with linear schemes (Section~III.A of~\cite{issa2013two}) or interference management in packet collision networks (Section~IV of~\cite{vahid2013communication}).

\begin{figure*}[ht]
\centering
\subfigure[]{\includegraphics[height = 3 cm]{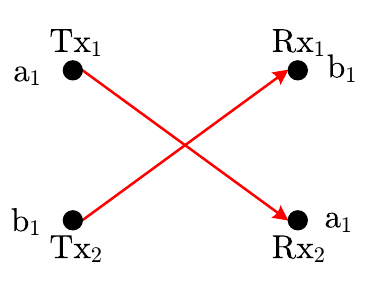}}
\hspace{0.5 in}
\subfigure[]{\includegraphics[height = 3 cm]{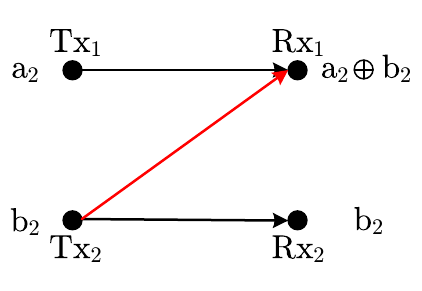}}
\hspace{0.5 in}
\subfigure[]{\includegraphics[height = 3 cm]{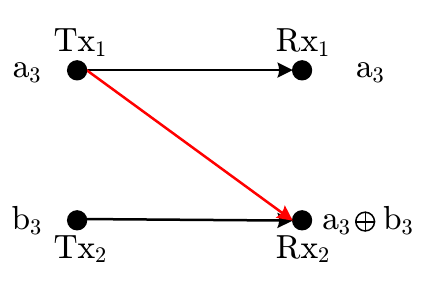}}
\caption{\it Achievability ideas with Delayed-CSIT: $(a)$ via Delayed-CSIT transmitters can figure out which bits are already known at the unintended receivers. Transmission of these bits will no longer create interference at the unintended receivers; $(b)$ to decode the bits, it is sufficient that ${\sf Tx}_2$ provides ${\sf Rx}_1$ with $b_2$ while this bit is available at ${\sf Rx}_2$; and $(c)$ is similar to $(b)$. Note that in $(b)$ and $(c)$ the intended receivers are swapped.\label{fig:opp7-8-9}}
\end{figure*}

\subsection{Achievability Ideas with Delayed-CSIT}
\label{sec:oppideas1}

As we have described in Section~\ref{sec:problem}, the channel gains are independent from each other and over time. This way, transmitters cannot use the delayed knowledge of the channel state information to predict future. However, this information can still be very useful. In particular, Delayed-CSIT allows us to evaluate the contributions of the desired signal and the interference at each receiver in the past signaling stages and exploit it as available side information for future communication. 

\subsubsection{Interference-free Bits}

Using Delayed-CSIT transmitters can identify previously transmitted bits such that if retransmitted, they do not create any further interference. The following examples clarify this idea. 

\emph{ Example 1} [Creating interference channels with side information]: Suppose at a time instant, each one of the transmitters simultaneously sends one data bit. The bits of ${\sf Tx}_1$ and ${\sf Tx}_2$ are denoted by $a_1$ and $b_1$ respectively. Later, using Delayed-CSIT, transmitters figure out that only the cross links were equal to $1$ at this time instant as shown in Fig.~\ref{fig:opp7-8-9}(a). This means that in future, transmission of these bits will no longer create interference at the unintended receivers. 

\emph{ Example 2} [Creating interference channels with swapped receivers and side information]: Assume that at a time instant, transmitters one and two simultaneously send data bits $a_2$ and $b_2$ respectively. Again through Delayed-CSIT, transmitters realize that all links except the link between ${\sf Tx}_1$ and ${\sf Rx}_2$ were equal to $1$, see Fig.~\ref{fig:opp7-8-9}(b). In a similar case, assume that at another time instant, transmitters one and two send data bits $a_3$ and $b_3$ at the same time. Through Delayed-CSIT, transmitters realize that all links except the link between ${\sf Tx}_2$ and ${\sf Rx}_1$ were connected, see Fig.~\ref{fig:opp7-8-9}(c). Then it is easy to see that to successfully finish delivering these bits, it is enough that ${\sf Tx}_1$ sends $a_3$ to ${\sf Rx}_2$, while this bit is already available at ${\sf Rx}_1$; and ${\sf Tx}_2$ sends $b_2$ to ${\sf Rx}_1$, while this bit is already available at ${\sf Rx}_2$. Note that here the intended receivers are \emph{swapped}. 

\begin{remark}
As described in Examples~1 and~2, an interference free bit can be retransmitted without worrying about creating interference at the unintended receiver. These bits can be transferred to a sub-problem, where in a two-user interference channel, ${\sf Rx}_i$ has apriori access to $\hbox{W}_{\bar{i}}$ as depicted in Fig.~\ref{fig:detIC-intfree}, $i=1,2$. Since there will be no interference in this sub-problem, such bits can be communicated at higher rates.
\end{remark}

\begin{figure}[h]
\centering
\includegraphics[height = 5cm]{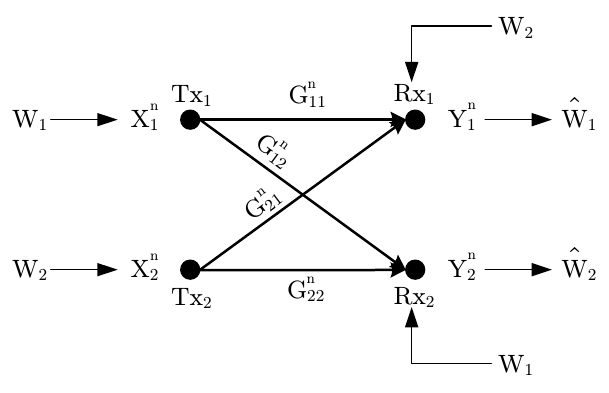}
\caption{Interference channel with side information: the capacity region with no, delayed, or instantaneous CSIT is the same.\label{fig:detIC-intfree}}
\end{figure}

\subsubsection{Bits of Common Interest}

Transmitters can use the delayed knowledge of the channel state information to identify bits that are of interest of both receivers. Below, we clarify this idea through several examples.

\emph{ Example 3} [Opportunistic creation of bits of common interest]: Suppose at a time instant, each one of the transmitters sends one data bit, say $a_4$ and $b_4$ respectively. Later, using Delayed-CSIT, transmitters figure out that all links were equal to $1$. In this case, both receivers have an equation of the transmitted bits, see Fig.~\ref{fig:opp-common}(a). Now, we notice that it is sufficient to provide either of the transmitted bits, $a_4$ or $b_4$, to both receivers rather than retransmitting both bits. We refer to such bits as bits of common interest. Since such bits are useful for both receivers, they can be transmitted more efficiently. 

\begin{figure}[h]
\centering
\subfigure[]{\includegraphics[height = 2.5 cm]{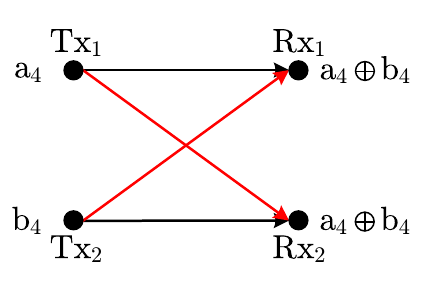}}
\hspace{0.1 in}
\subfigure[]{\includegraphics[height = 2.5 cm]{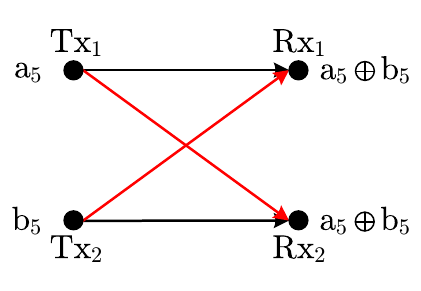}}
\caption{\it In each case, it is sufficient to provide only one of the transmitted bits to both receivers. We refer to such bits as bits of common interest.\label{fig:opp-common}}
\end{figure}



\begin{remark} [Pairing bits of common interest to create a two-multicast problem] We note that in Example 3, one of the transmitters takes the responsibility of delivering one bit of common interest to both receivers. To improve the performance, we can pair this problem with another similar problem as follows. Assume that in another time instant, each one of the transmitters sends one data bit, say $a_5$ and $b_5$ respectively.  Later,  transmitters figure out that all links were equal to $1$, see Fig.~\ref{fig:opp-common}(b). In this case, similar to Example 3, one of the bits $a_5$ and $b_5$, say $b_5$, can be chosen as the bit of common interest. Now we can pair cases depicted in Fig.~\ref{fig:opp-common}(a) and Fig.~\ref{fig:opp-common}(b). Then transmitters can simultaneously send bits $a_4$ and $b_5$ to both receivers. With this pairing, we take advantage of all four links to transmit information. 
\end{remark}

\begin{remark} [Pairing ICs with side information to create a two-multicast problem (pairing Type-I)]
The advantage of interference channels with side information, explained in Examples~1 and~2, is that due to the side information, there is no interference involved in the problem. The downside is that half of the links in the channel become irrelevant and unexploited. More precisely, the cross links in  Example~1 and the direct links in Example~2 are not utilized to increase the rate. Here, we show that these two problems can be paired together to form an efficient two-multicast problem via creating bits of common interest. Referring to Fig.~\ref{fig:opp7-8-9}, one can easily verify that it is enough to deliver $a_1 \oplus a_3$ and $b_1 \oplus b_2$ to both receivers. For instance, if $a_1 \oplus a_3$ and $b_1 \oplus b_2$ are available at ${\sf Rx}_1$, it can remove $b_1$ from $b_1 \oplus b_2$ to decode $b_2$, then using $b_2$ and $a_2 \oplus b_2$ it can decode $a_2$; finally, using $a_3$ and $a_1 \oplus a_3$ it can decode $a_1$. Indeed, bit $a_1 \oplus a_3$ available at ${\sf Tx}_1$, and bit $b_1 \oplus b_2$ available at ${\sf Tx}_2$, are bits of common interest and can be transmitted to both receivers simultaneously in the efficient two-multicast problem as depicted in Fig.~\ref{fig:two-multicastopp}. We note that for the two-multicast problem, the capacity region with no, delayed, or instantaneous CSIT is the same. We shall refer to this pairing as pairing {\bf Type-I} throughout the paper.
\end{remark}

\begin{figure}[h]
\centering
\includegraphics[height = 3.5 cm]{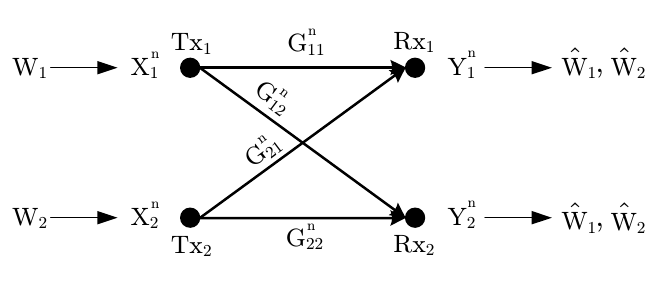}
\caption{Two-multicast network. Transmitter ${\sf Tx}_i$ wishes to reliably communicate message $\hbox{W}_i$ to both receivers, $i=1,2$. The capacity region with no, delayed, or instantaneous CSIT is the same.\label{fig:two-multicastopp}}
\end{figure}

\emph{ Example 4} [Pairing interference-free bits with bits of common interest to create a two-multicast problem (pairing Type-II)]: Suppose at a time instant, each one of the transmitters sends one data bit, say $a_6$ and $b_6$ respectively. Later, using Delayed-CSIT, transmitters figure out that all links were equal to $1$, see Fig.~\ref{fig:opp-commoncoding}(a). In another time instant, each one of the transmitters sends one data bit, say $a_7$ and $b_7$ respectively.  Later,  transmitters figure out that only the cross links were equal to $1$, see Fig.~\ref{fig:opp-commoncoding}(b). Now, we observe that providing $a_6 \oplus a_7$ and $b_6 \oplus b_7$ to both receivers is sufficient to decode the bits. For instance if ${\sf Rx}_1$ is provided with $a_6 \oplus a_7$ and $b_6 \oplus b_7$, then it will use $b_7$ to decode $b_6$, from which it can obtain $a_6$, and finally using $a_6$ and $a_6 \oplus a_7$, it can decode $a_7$. Thus, bit $a_6 \oplus a_7$ available at ${\sf Tx}_1$, and bit $b_6 \oplus b_7$ available at ${\sf Tx}_2$, are bits of common interest and can be transmitted to both receivers simultaneously in the efficient two-multicast problem. We shall refer to this pairing as pairing {\bf Type-II} throughout the paper.

\begin{figure}[h]
\centering
\subfigure[]{\includegraphics[height = 2.5 cm]{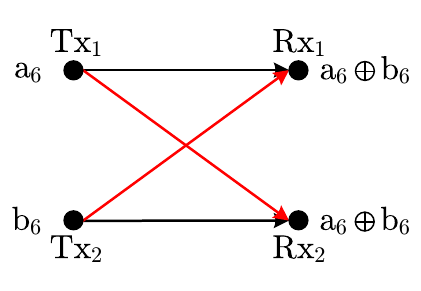}}
\hspace{0.1 in}
\subfigure[]{\includegraphics[height = 2.5 cm]{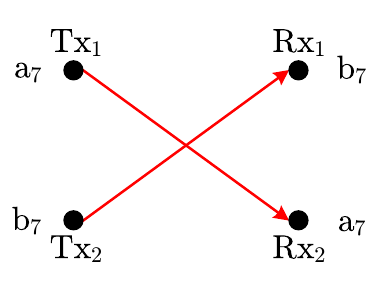}}
\caption{\it Pairing Type-II: providing $a_6 \oplus a_7$ and $b_6 \oplus b_7$ to both receivers is sufficient to decode the bits. In other words, bit $a_6 \oplus a_7$ available at ${\sf Tx}_1$, and bit $b_6 \oplus b_7$ available at ${\sf Tx}_2$, are bits of common interest and can be transmitted to both receivers simultaneously in the efficient two-multicast problem. Note that in $(b)$, the cross links would have been irrelevant for future communications, however, using this pairing, we exploit all links.\label{fig:opp-commoncoding}}
\end{figure}

\emph{ Example 5} [Pairing bits of common interest with interference-free bits with swapped receivers to create a two-multicast problem (pairing Type-III)]: Suppose at a time instant, each one of the transmitters sends one data bit, say $a_8$ and $b_8$ respectively. Later, using Delayed-CSIT, transmitters figure out that all links were equal to $1$ as in Fig.~\ref{fig:opp-commoncodingswapped}(a). In another time instant, each one of transmitters sends one data bit, say $a_9$ and $b_9$ respectively.  Later, transmitters figure out that all links were equal to $1$ except the link from ${\sf Tx}_2$ to ${\sf Rx}_1$, see Fig.~\ref{fig:opp-commoncodingswapped}(b). In a similar case, assume that at another time instant, transmitters one and two send data bits $a_{10}$ and $b_{10}$ at the same time. Through Delayed-CSIT, transmitters realize that all links except the link between ${\sf Tx}_1$ and ${\sf Rx}_2$ were connected, see Fig.~\ref{fig:opp-commoncodingswapped}(c). We observe that providing $a_8 \oplus a_9$ and $b_8 \oplus b_{10}$ to both receivers is sufficient to decode the bits. For instance, if ${\sf Rx}_1$ is provided with $a_8 \oplus a_9$ and $b_8 \oplus b_{10}$, then it will use $a_9$ to decode $a_8$, from which it can obtain $b_8$, then using $b_8$ and $b_8 \oplus b_{10}$, it gains access to $b_{10}$, finally using $b_{10}$, it can decode $a_{10}$ from $a_{10} \oplus b_{10}$. Thus, bit $a_8 \oplus a_9$ available at ${\sf Tx}_1$, and bit $b_8 \oplus b_{10}$ available at ${\sf Tx}_2$, are bits of common interest and can be transmitted to both receivers simultaneously in the efficient two-multicast problem. We shall refer to this pairing as pairing {\bf Type-III} throughout the paper.

\begin{figure}[h]
\centering
\subfigure[]{\includegraphics[height = 2.5 cm]{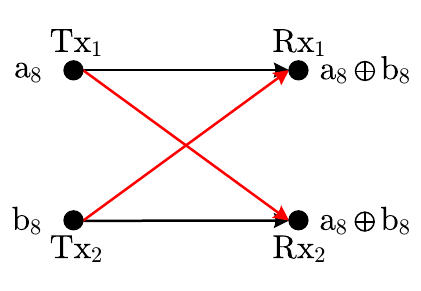}}
\hspace{0.1 in}
\subfigure[]{\includegraphics[height = 2.5 cm]{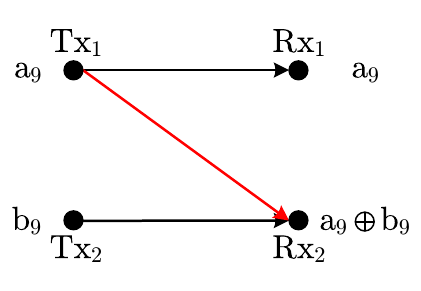}}
\subfigure[]{\includegraphics[height = 2.5 cm]{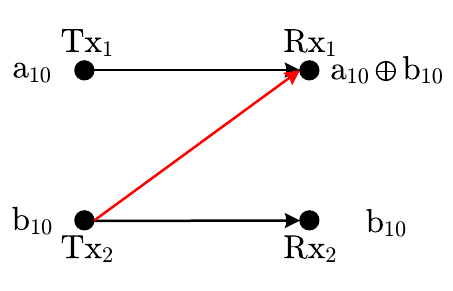}}
\caption{\it Pairing Type-III: providing $a_8 \oplus a_9$ and $b_8 \oplus b_{10}$ to both receivers is sufficient to decode the bits. In other words, bit $a_8 \oplus a_9$ available at ${\sf Tx}_1$, and bit $b_8 \oplus b_{10}$ available at ${\sf Tx}_2$, are bits of common interest and can be transmitted to both receivers simultaneously in the efficient two-multicast problem.\label{fig:opp-commoncodingswapped}}
\end{figure}

As explained in the above examples, there are several ways to exploit the available side information at each transmitter. To achieve the capacity region,  the first challenge is to evaluate various options and choose the most efficient one.  The second challenge is that different opportunities may occur with different probabilities. This makes the process of matching, combining, and upgrading the status of the bits difficult. Unfortunately, there is no simple guideline to decide when to search for the most efficient combination of the opportunities and when to hold on to other schemes. It is also important to note that most of the opportunities we observe here, do not appear in
achieving the DoF of the Gaussian multi-antenna interference channels (see \emph{e.g.},~\cite{GhasemiX1,Vaze_DCSIT_MIMO_BC}).




\subsection{Achievability Ideas with Output Feedback}
\label{sec:oppideas2}

In this subsection, we focus on the effect of the output feedback in the presence of Delayed-CSIT. The first observation is that through output feedback, each transmitter can evaluate the interference of the other transmitter, and therefore, has access to the previously transmitted signal of the other user. Thus, output feedback can create new path for information flow between  each transmitter and the corresponding receiver, \emph{e.g.}, $${\sf Tx}_1 \rightarrow {\sf Rx}_2 \rightarrow {\sf Tx}_2 \rightarrow {\sf Rx}_1.$$   

Although this additional path can improve the rate region, the advantage of output feedback is not limited to that. We explain the new opportunities through two examples

\emph{ Example 6} [Creating two-multicast problem from ICs with side information]:
In the previous subsection, we showed that interference-free transmissions can be upgraded  to two-multicast problems through pairing. However, it is important to note that the different channel realizations used for pairing do not occur at the same probability. Therefore, it is not always possible to fully implement pairing in all cases. In particular, in some cases, some interference-free transmissions are left alone without possibility of pairing. In this example, we show that output feedback allows us to create bits of common interest out of these cases, which in turn allows us to create two-multicast problems.  Referring to Fig.~\ref{fig:opp13}, one can see that through the output feedback links, transmitters one and two can learn $b_{11}$ and $a_{11}$ respectively. Therefore, either of the transmitters is able to create $a_{11} \oplus b_{11}$. It is easy to see that  $a_{11} \oplus b_{11}$ is of interest of both receivers. Indeed, output feedback allows us to form a bit of common interest which can be delivered through the efficient two-multicast problem. 

\begin{figure}[h]
\centering
\includegraphics[height = 2.5 cm]{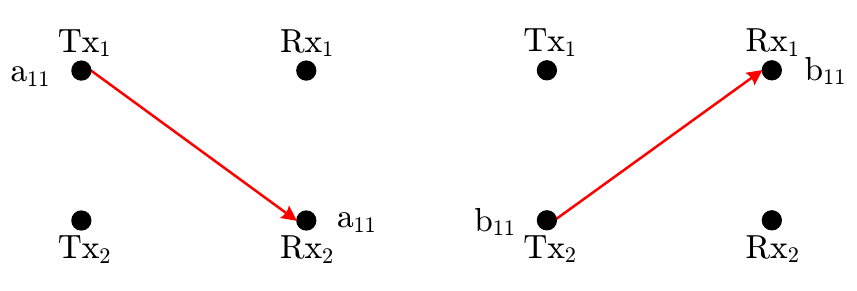}
\caption{\it Opportunistic creation of bits of common interest using output feedback: bit $b_{11}$ is available at ${\sf Tx}_1$ via the feedback link from ${\sf Rx}_1$; it is sufficient that ${\sf Tx}_1$ provides $a_{11} \oplus b_{11}$ to both receivers.\label{fig:opp13}}
\end{figure}

\emph{ Example 7} [Creating two-multicast problem from ICs with swapped receivers and side information]:
As another example, consider the two channel gain realizations depicted in Fig.~\ref{fig:opp14}. In these cases, using output feedback ${\sf Tx}_1$ can learn the transmitted bit of ${\sf Tx}_2$ (\emph{i.e.} $b_{12}$), and then form  $a_{13} \oplus b_{12}$. It is easy to see that $a_{13} \oplus b_{12}$ is useful for both receivers and thus is a bit of common interest. Similar argument is valid for the second receiver. This means that output feedback allows us to upgrade interference-free transmissions with swapped receivers to bits of common interest that can be used to form efficient two-multicast problems.

\begin{figure}[h]
\centering
\includegraphics[height = 2.5 cm]{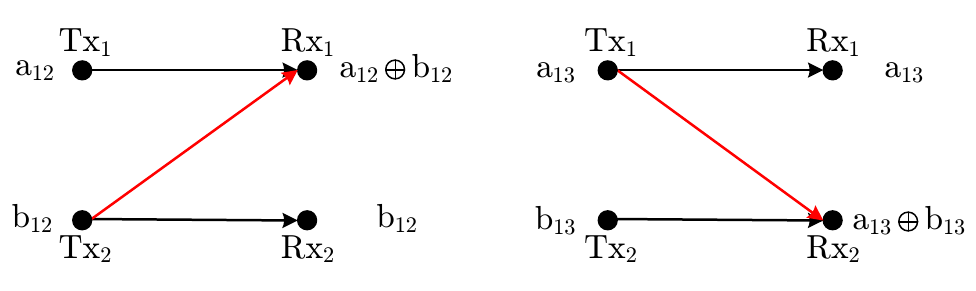}
\caption{\it Opportunistic creation of bits of common interest using output feedback: using output feedback ${\sf Tx}_1$ can learn the transmitted bit of ${\sf Tx}_2$ (\emph{i.e.} $b_{12}$); now, we observe that providing $a_{13} \oplus b_{12}$ to both receivers is sufficient to decode the intended bits.\label{fig:opp14}}
\end{figure}

\subsection{Key Idea for Converse Proofs with Delayed-CSIT}

While we provide detailed proofs in Sections~\ref{sec:outerHalf} and~\ref{sec:conversedelayedhalfFB}, we try to describe the main challenge in deriving the outer-bounds in this subsection. Consider the Delayed-CSIT scenario and suppose rate tuple $\left( R_1, R_2 \right)$ is achievable. Then for $\beta > 0$, we have
\begin{align}
n &\left( R_1 + \beta R_2 \right) = H(W_1|W_2, G^n) + \beta H(W_2|G^n) \nonumber \\
& \overset{(\mathrm{Fano})}\leq I(W_1;Y_1^n|W_2,G^n) + \beta I(W_2;Y_2^n|G^n) + n \epsilon_n \nonumber \\
& = \beta H(Y_2^n|G^n) + \underbrace{H(G_{11}^nX_1^n|G^n) - \beta H(G_{12}^nX_1^n|G^n)} + n \epsilon_n.
\end{align}

We refer the reader to Section~\ref{sec:outerHalf} for the detailed derivation of each step. Here, we would like to find a value of $\beta$ such that  
\begin{align}
\label{eq:betanegative}
H(G_{11}^nX_1^n|G^n) - \beta H(G_{12}^nX_1^n|G^n) \leq 0,
\end{align}
for \emph{any} input distribution. Note that since the terms involved are only a function of $X_1^n$ and the channel gains, this term resembles a broadcast channel formed by ${\sf Tx}_1$ and the two receivers. Therefore, the main challenge boils down to understanding the ratio of the entropies of the received signals in a broadcast channel, and this would be the main focus of this subsection.

\begin{figure}[ht]
\centering
\includegraphics[height = 3cm]{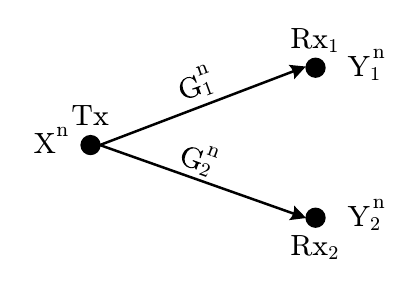}
\caption{A transmitter connected to two receivers through binary fading channels.\label{fig:portionHalfopp}}
\end{figure}

Consider a transmitter that is connected to two receivers through binary fading channels as depicted in Fig.~\ref{fig:portionHalfopp}. We would like to understand how much this transmitter can privilege receiver one to receiver two, given outdated knowledge of the channel state information. Our metric would be the ratio of the entropies of the received signals\footnote{We point out that if $H(G_{11}^nX_1^n|G^n) = 0$, then ratio is not defined. But we keep in mind that what we really care about is (\ref{eq:betanegative}).}. In other words, we would like to understand what is the lower-bound on the ratio of the entropy of the received signal at ${\sf Rx}_2$ to that of ${\sf Rx}_1$. We first point out the result for the No-CSIT and Instantaneous-CSIT cases. With No-CSIT, from transmitter's point of view the two receivers are identical and it cannot favor one over the other and as a result, the two entropies would be equal. However with Instantaneous-CSIT, transmitter can choose to transmit at time $t$ only if $G_1[t] = 1$ and $G_2[t] = 0$. Thus, with Instantaneous-CSIT the ratio of interest could be as low as $0$. For the Delayed-CSIT case, we have the following lemma which we will formally prove in Section~\ref{sec:outerHalf}. Here, we try to provide some intuition about the problem by describing an input distribution that utilizes delayed knowledge of the channel state information in order to favor receiver one. It is important to keep in mind that this should not be considered as a proof but rather just a helpful intuition. Also, we point out that for the two-user BFIC with Delayed-CSIT and OFB, we will derive a variation of this lemma in Section~\ref{sec:conversedelayedhalfFB}.

\begin{lemma}
\label{lemma:portionopp}
{\bf [Entropy Leakage]} For the channel described above with Delayed-CSIT, and for \emph{any} input distribution, we have
\begin{align}
\label{eq:portionopp}
H\left( Y_2^n | G^n \right) \geq \frac{p}{1-q^2} H\left( Y_1^n | G^n \right).
\end{align}
\end{lemma}

\begin{figure}[ht]
\centering
\subfigure[]{\includegraphics[height = 2.5 cm]{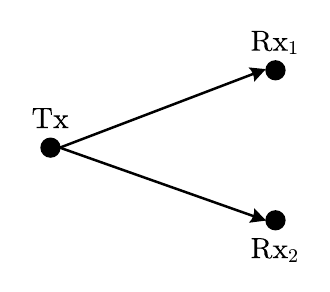}}
\hspace{0.5 in}
\subfigure[]{\includegraphics[height = 2.5 cm]{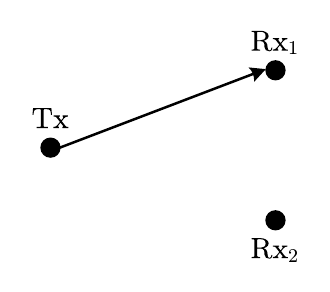}}
\subfigure[]{\includegraphics[height = 2.5 cm]{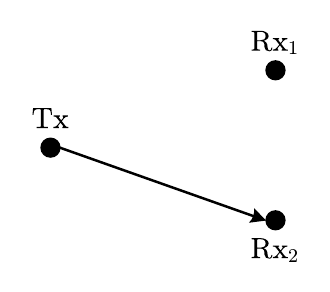}}
\hspace{0.5 in}
\subfigure[]{\includegraphics[height = 2.5 cm]{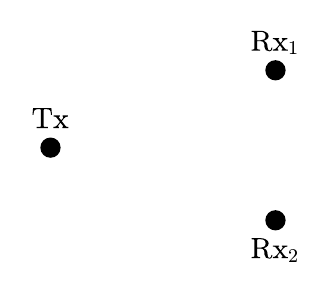}}
\caption{\it Four possible channel realizations for the network in Fig.~\ref{fig:portionHalfopp}. The transmitter sends out a data bit at time instant $t$, and at the next time instant, using Delayed-CSIT, he knows which channel realization has occurred.  If either of the realizations $(a)$ or $(b)$ occurred at time $t$, then we remove the transmitted bit from the initial queue. However, if either of the realizations $(c)$ or $(d)$ occurred at time $t$, we leave this bit in the initial queue. This way the transmitter favors receiver one over receiver two.\label{fig:realizationsopp}}
\end{figure}

As mentioned before, we do not intend to prove this lemma here. We only provide an input distribution for which this lower-bound is tight. Consider $m$ bits drawn from i.i.d. Bernoulli $0.5$ random variables and assume these bits are in some initial queue. At any time instant $t$, the transmitter sends one of the bits in this initial queue (if the queue is empty, then the scheme is terminated). At time instant $t+1$, using Delayed-CSIT, the transmitter knows which one of the four possible channel realizations depicted in Fig.~\ref{fig:realizationsopp} has occurred at time $t$. If either of the realizations $(a)$ or $(b)$ occurred at time $t$, then we remove the transmitted bit from the initial queue. However, if either of the realizations $(c)$ or $(d)$ occurred at time $t$, we leave this bit in the initial queue (\emph{i.e.} among the bits that can be transmitted at any future time instant). Note that this way, any bit that is available at ${\sf Rx}_2$ would be available at ${\sf Rx}_1$. However, there will be bits that are only available at ${\sf Rx}_1$. Hence, transmitter has favored receiver one over receiver two. If we analyze this scheme, we get
\begin{align}
H\left( Y_2^n | G^n \right) = \frac{p}{1-q^2} H\left( Y_1^n | G^n \right),
\end{align}
meaning that the bound given in (\ref{eq:portionopp}) is achievable and thus, it is tight.

Now that we have described the key ideas we incorporate in this paper, starting next section, we provide the proof of our main results.

\section{Achievability Proof of Theorem~\ref{THM:IC-DelayedCSIT}~[Delayed-CSIT]}
\label{sec:achallhalf}


For $0 \leq p \leq \left( 3 - \sqrt{5} \right)/2$, the capacity of the two-user BFIC with Delayed-CSIT is depicted in Fig.~\ref{fig:regionHalf}(a) and as a result, it is sufficient to describe the achievability for point $A = \left( p, p \right)$. However, for $\left( 3 - \sqrt{5} \right)/2 < p \leq 1$, all bounds are active and the region, as depicted in Fig.~\ref{fig:regionHalf}(b), is the convex hull of points $A,B,$ and $C$. By symmetry, it is sufficient to describe the achievability for points $A$ and $C$ in this regime. 

We first provide the achievability proof of point $A$ for $0.5 \leq p \leq 1$ in this section. Then, we provide an overview of the achievability proof of corner point $C$ and we postpone the detailed proof to Appendix~\ref{Appendix:cornerdelayed}. Finally in Appendix~\ref{Appendix:lessthanHalf}, we present the achievability proof of point $A$ for $0 \leq p < 0.5$. 





\begin{figure}[ht]
\centering
\subfigure[]{\includegraphics[height = 4.5cm]{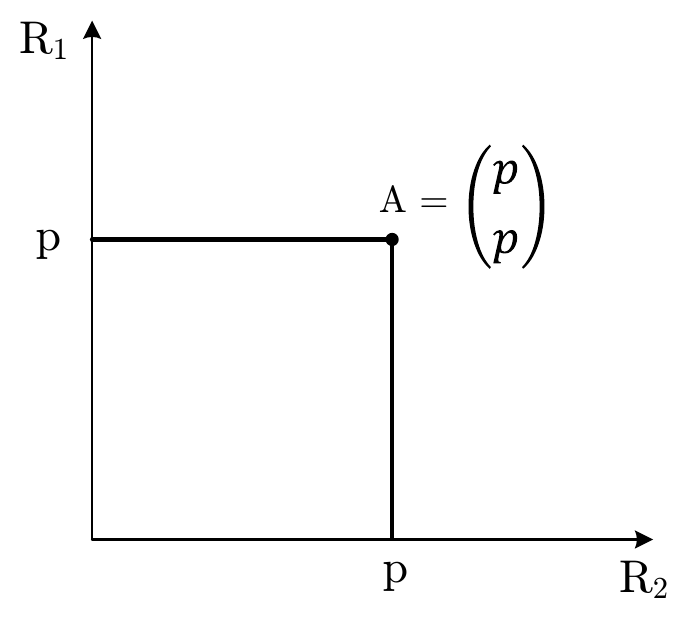}}
\subfigure[]{\includegraphics[height = 4.5cm]{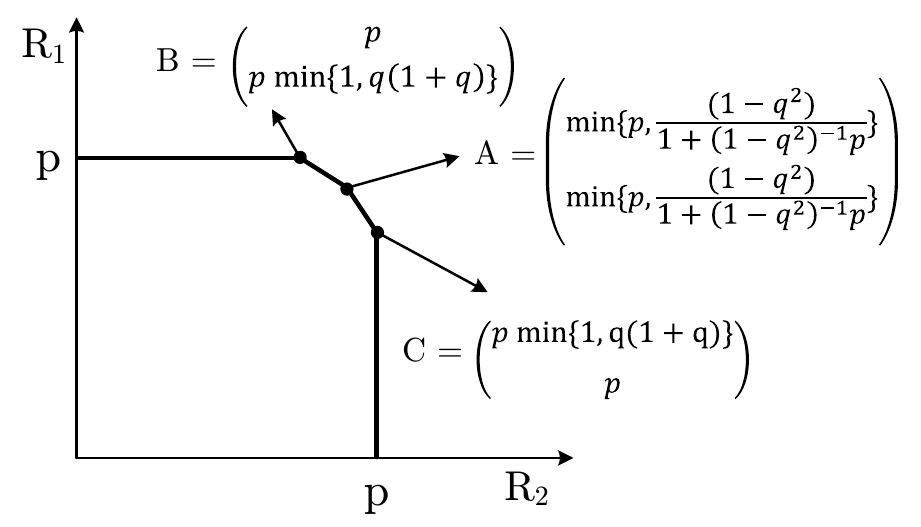}}
\caption{\it Capacity Region of the two-user Binary Fading IC with Delayed-CSIT for: $(a)$ $0 \leq p \leq \left( 3 - \sqrt{5} \right)/2$; and $(b)$ $\left( 3 - \sqrt{5} \right)/2 < p \leq 1$.\label{fig:regionHalf}}
\end{figure}

  

\subsection{Achievability Strategy for Corner Point $A$}

In this subsection, we describe a transmission strategy that achieves a rate tuple arbitrary close to corner point $A$ for $0.5 \leq p \leq 1$ as depicted in Fig.~\ref{fig:regionHalf}(b), \emph{i.e.} 
\begin{align}
\label{eq:cornerA}
R_1 = R_2 = \frac{(1-q^2)}{1+(1-q^2)^{-1}p} \raisebox{2pt}{.}
\end{align}

Let the messages of transmitters one and two be denoted by $\hbox{W}_1 = a_1,a_2,\ldots,a_m$, and $\hbox{W}_2 = b_1,b_2,\ldots,b_m$, respectively, where data bits $a_i$'s and $b_i$'s are picked uniformly and independently from $\{ 0,1 \}$, $i=1,\ldots,m$. We show that it is possible to communicate these bits in 
\begin{align}
n = \left( 1 - q^2 \right)^{-1}m + \left( 1 - q^2 \right)^{-2} pm + O\left( m^{2/3}\right)
\end{align}
time instants\footnote{Throughout the paper whenever we state the number of bits or time instants, say $n$, if the expression for a given value of $p$ is not an integer, then we use the ceiling of that number $\lceil n \rceil$, where $\lceil . \rceil$ is the smallest integer greater than or equal to $n$. Note that since we will take the limit as $m \rightarrow \infty$, this does not change the end results.} with vanishing error probability (as $m \rightarrow \infty$). Therefore achieving the rates given in (\ref{eq:cornerA}) as $m \rightarrow \infty$. Our transmission strategy consists of two phases as described below. 


\begin{table*}[t]
\caption{All possible channel realizations and transitions from the initial queue to other queues; solid arrow from tranmsitter ${\sf Tx}_i$ to receiver ${\sf Rx}_j$ indicates that $G_{ij}[t] = 1$, $i,j \in \{ 1,2\}$, $t=1,2,\ldots,n$. Bit ``$a$'' represents a bit in $Q_{1 \rightarrow 1}$ while bit ``$b$'' represents a bit in $Q_{2 \rightarrow 2}$.}
\centering
\begin{tabular}{| c | c | c | c | c | c |}
\hline
case ID		 & channel realization    & state transition  & case ID		 & channel realization    & state transition \\
					 & at time instant $n$    &                   & 					 & at time instant $n$    &                  \\ [0.5ex]

\hline

\raisebox{18pt}{$1$}    &    \includegraphics[height = 1.4cm]{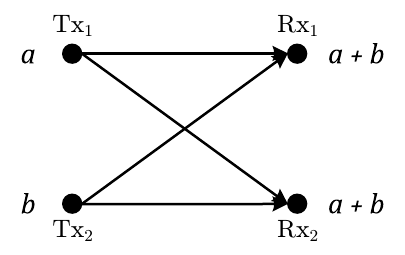}     &  \raisebox{18pt}{$ \left\{ \begin{array}{ll}  \vspace{1mm} a \rightarrow Q_{1,{\sf C}_1}  & \\ b \rightarrow Q_{2,{\sf C}_1}  &  \end{array} \right. $}  &  \raisebox{18pt}{$9$}    &    \includegraphics[height = 1.4cm]{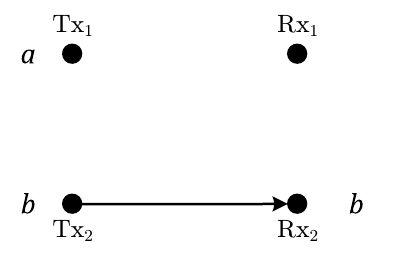}     &   \raisebox{18pt}{$ \left\{ \begin{array}{ll}  \vspace{1mm} a \rightarrow Q_{1 \rightarrow 1}  & \\ b \rightarrow Q_{2 \rightarrow F}  &  \end{array} \right. $} \\

\hline

\raisebox{18pt}{$2$}    &    \includegraphics[height = 1.4cm]{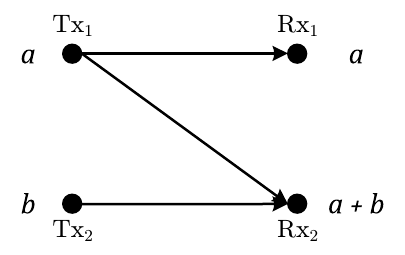}     &  \raisebox{18pt}{$ \left\{ \begin{array}{ll}  \vspace{1mm} a \rightarrow Q_{1 \rightarrow 2|1}  & \\ b \rightarrow Q_{2 \rightarrow F}  &  \end{array} \right. $} &  \raisebox{18pt}{$10$}    &    \includegraphics[height = 1.4cm]{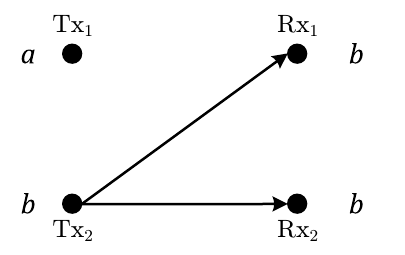}     &   \raisebox{18pt}{$ \left\{ \begin{array}{ll} \vspace{1mm} a \rightarrow Q_{1 \rightarrow 1}  & \\ b \rightarrow Q_{2 \rightarrow F}  &  \end{array} \right. $} \\

\hline

\raisebox{18pt}{$3$}    &    \includegraphics[height = 1.4cm]{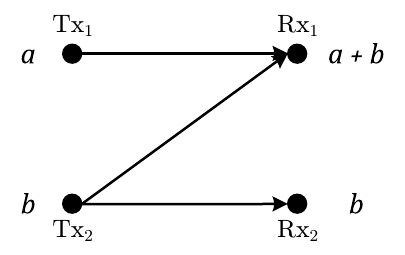}     &  \raisebox{18pt}{$ \left\{ \begin{array}{ll}  \vspace{1mm} a \rightarrow Q_{1 \rightarrow F}  & \\ b \rightarrow Q_{2 \rightarrow 1|2}   &  \end{array} \right. $} &  \raisebox{18pt}{$11$}    &    \includegraphics[height = 1.4cm]{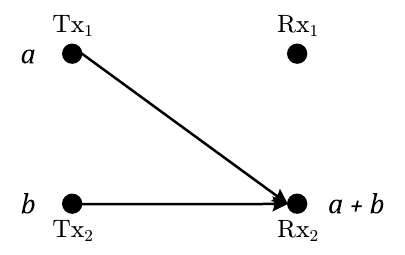}     &   \raisebox{18pt}{$ \left\{ \begin{array}{ll}  \vspace{1mm} a \rightarrow Q_{1 \rightarrow \{ 1,2\}}  & \\ b \rightarrow Q_{2 \rightarrow F} &  \end{array} \right. $} \\

\hline

\raisebox{18pt}{$4$}    &    \includegraphics[height = 1.4cm]{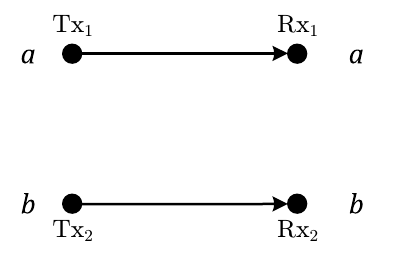}     &   \raisebox{18pt}{$ \left\{ \begin{array}{ll} \vspace{1mm} a \rightarrow Q_{1 \rightarrow F}  & \\ b \rightarrow Q_{2 \rightarrow F}  &  \end{array} \right. $} &  \raisebox{18pt}{$12$}    &    \includegraphics[height = 1.4cm]{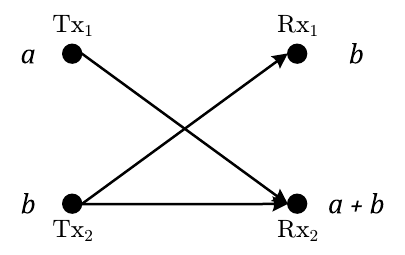}     &   \raisebox{18pt}{$ \left\{ \begin{array}{ll}  \vspace{1mm} a \rightarrow Q_{1 \rightarrow \{ 1,2\}}  & \\ b \rightarrow Q_{2 \rightarrow F} &  \end{array} \right. $} \\

\hline

\raisebox{18pt}{$5$}    &    \includegraphics[height = 1.4cm]{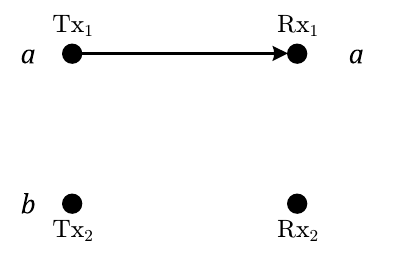}     &   \raisebox{18pt}{$ \left\{ \begin{array}{ll} \vspace{1mm} a \rightarrow Q_{1 \rightarrow F}  & \\ b \rightarrow Q_{2 \rightarrow 2}  &  \end{array} \right. $} &  \raisebox{18pt}{$13$}    &    \includegraphics[height = 1.4cm]{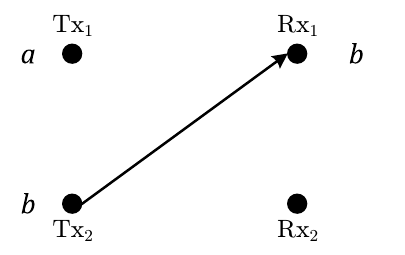}     &   \raisebox{18pt}{$ \left\{ \begin{array}{ll}  \vspace{1mm} a \rightarrow Q_{1 \rightarrow 1}  & \\ b \rightarrow Q_{2 \rightarrow 2|1}  &  \end{array} \right. $} \\

\hline

\raisebox{18pt}{$6$}    &    \includegraphics[height = 1.4cm]{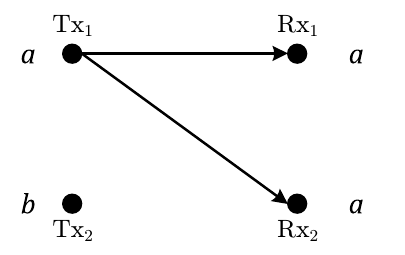}     &   \raisebox{18pt}{$ \left\{ \begin{array}{ll} \vspace{1mm} a \rightarrow Q_{1 \rightarrow F}  & \\ b \rightarrow Q_{2 \rightarrow 2}  &  \end{array} \right. $} &  \raisebox{18pt}{$14$}    &    \includegraphics[height = 1.4cm]{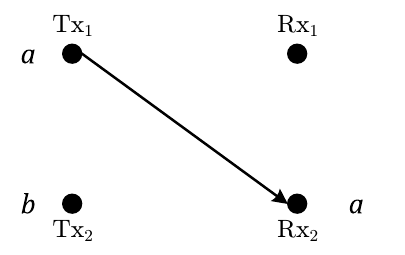}     &   \raisebox{18pt}{$ \left\{ \begin{array}{ll}  \vspace{1mm} a \rightarrow Q_{1 \rightarrow 1|2}  & \\ b \rightarrow Q_{2 \rightarrow 2}  &  \end{array} \right. $} \\

\hline

\raisebox{18pt}{$7$}    &    \includegraphics[height = 1.4cm]{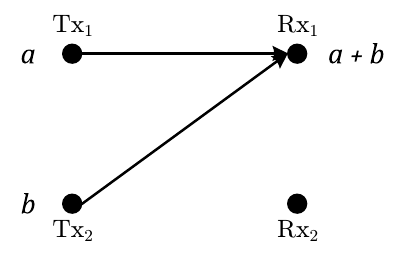}     &   \raisebox{18pt}{$ \left\{ \begin{array}{ll}  \vspace{1mm} a \rightarrow Q_{1 \rightarrow F}  & \\ b \rightarrow Q_{2 \rightarrow \{ 1,2\}} &  \end{array} \right. $} &  \raisebox{18pt}{$15$}    &    \includegraphics[height = 1.4cm]{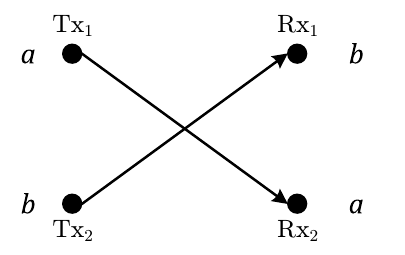}     &   \raisebox{18pt}{$ \left\{ \begin{array}{ll} \vspace{1mm} a \rightarrow Q_{1 \rightarrow 1|2}  & \\ b \rightarrow Q_{2 \rightarrow 2|1}  &  \end{array} \right. $} \\

\hline

\raisebox{18pt}{$8$}    &    \includegraphics[height = 1.4cm]{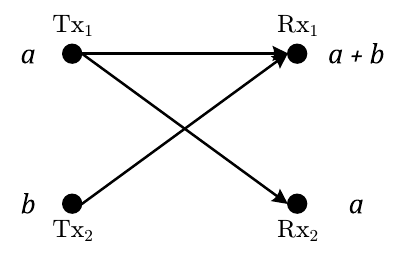}     &   \raisebox{18pt}{$ \left\{ \begin{array}{ll}  \vspace{1mm} a \rightarrow Q_{1 \rightarrow F}  & \\ b \rightarrow Q_{2 \rightarrow \{ 1,2\}} &  \end{array} \right. $} &  \raisebox{18pt}{$16$}    &    \includegraphics[height = 1.4cm]{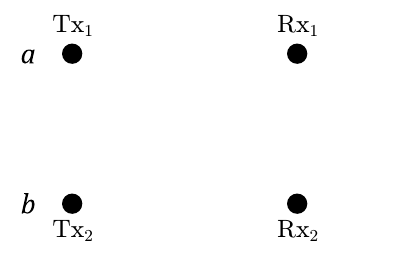}     &   \raisebox{18pt}{$ \left\{ \begin{array}{ll} \vspace{1mm} a \rightarrow Q_{1 \rightarrow 1}  & \\ b \rightarrow Q_{2 \rightarrow 2}  &  \end{array} \right. $} \\

\hline

\end{tabular}
\label{table:16cases}
\end{table*}


\noindent {\bf Phase 1} [uncategorized transmission]: At the beginning of the communication block, we assume that the bits at ${\sf Tx}_i$ are in queue $Q_{i \rightarrow i}$ (the initial state of the bits), $i=1,2$. At each time instant $t$, ${\sf Tx}_i$ sends out a bit from $Q_{i \rightarrow i}$, and this bit will either stay in the initial queue or transition to one of the following possible queues will take place according to the description in Table~\ref{table:16cases}. If at time instant $t$, $Q_{i \rightarrow i}$ is empty, then ${\sf Tx}_i$, $i=1,2$, remains silent until the end of Phase 1.
\begin{enumerate}
\item[(A)] $Q_{i , {\sf C}_1 }$: The bits that at the time of communication, all channel gains were equal to $1$.
\item[(B)] $Q_{i \rightarrow \{ 1,2 \} }$: The bits that are of common interest of both receivers and do not fall in category (A).  
\item[(C)] $Q_{i \rightarrow i|\bar{i}}$: The bits that are required by ${\sf Rx}_i$ but are available at the unintended receiver ${\sf Rx}_{\bar{i}}$. A bit is in  $Q_{i \rightarrow i|\bar{i}}$ if ${\sf Rx}_{\bar{i}}$ gets it without interference and ${\sf Rx}_i$ does not get it with or without interference.
\item[(D)] $Q_{i \rightarrow \bar{i}|i}$: The bits that are required by ${\sf Rx}_{\bar{i}}$ but are available at the intended receiver ${\sf Rx}_i$. More precisely, a bit is in  $Q_{i \rightarrow \bar{i}|i}$ if ${\sf Rx}_{i}$ gets the bit without interference and ${\sf Rx}_{\bar{i}}$ gets it with interference.
\item[(E)] $Q_{i \rightarrow F}$: The bits that we consider delivered and no retransmission is required for them. 
\end{enumerate}

More precisely, based on the channel realizations, a total of $16$ possible configurations may occur at any time instant as summarized in Table~\ref{table:16cases}. The transition for each one of the channel realizations is as follows. 

\begin{itemize}

\item Case~1~$\left( \caseone \right)$: If at time instant $t$, Case 1 occurs, then each receiver gets a linear combination of the bits that were transmitted. Then as illustrated in Fig.~\ref{fig:case1}, if either of such bits is provided to both receivers then the receivers can decode both bits. The transmitted bit of ${\sf Tx}_i$ leaves $Q_{i \rightarrow i}$ and joins $Q_{i,{\sf C}_1}$\footnote{In this paper, we assume that the queues are ordered. Meaning that the first bit that joins the queue is placed at the head of the queue and any new bit occupies the next empty position. For instance, suppose there are $\ell$ bits in $Q_{1,C_1}$ and $\ell$ bits in $Q_{2,C_1}$, then the next time Case~1 occurs, the transmitted bit of ${\sf Tx}_i$ is placed at position $\ell + 1$ in $Q_{i,C_1}$, $i=1,2$.}, $i=1,2$. Although we can consider such bits as bits of common interest, we keep them in an intermediate queue for now and as we describe later, we combine them with other bits to create bits of common interest. 
\begin{figure}[h]
\centering
\includegraphics[height = 3 cm]{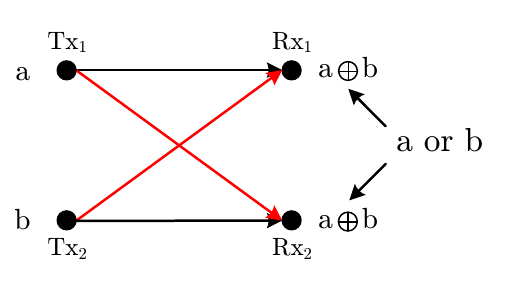}
\caption{Suppose transmitters one and two send out data bits $a$ and $b$ respectively and Case $1$ occurs. Now, if either of the transmitted bits is provided to both receivers, then each receiver can decode its corresponding bit.\label{fig:case1}}
\end{figure}

\item Case~2~$\left( \casetwo \right)$: In this case, ${\sf Rx}_1$ has already received its corresponding bit while ${\sf Rx}_2$ has a linear combination of the transmitted bits, see Table~\ref{table:16cases}. As a result, if the transmitted bit of ${\sf Tx}_1$ is provided to ${\sf Rx}_2$, it will be able to decode both bits. In other words, the transmitted bit from ${\sf Tx}_1$ is available at ${\sf Rx}_1$ and is required by ${\sf Rx}_2$. Therefore, transmitted bit of ${\sf Tx}_1$ leaves $Q_{1 \rightarrow 1}$ and joins $Q_{1 \rightarrow 2|1}$. Note that the bit of ${\sf Tx}_2$ will not be retransmitted since upon delivery of the bit of ${\sf Tx}_1$, ${\sf Rx}_2$ can decode its corresponding bit. Since no retransmission is required, the bit of ${\sf Tx}_2$ leaves $Q_{2 \rightarrow 2}$ and joins $Q_{2,F}$ (the final state of the bits).

\vspace{1mm}

\item Case~3~$\left( \casethree \right)$: This is similar to Case~2 with swapping user IDs.

\vspace{1mm}

\item Case~4~$\left( \casefour \right)$: In this case, each receiver gets its corresponding bit without any interference. We consider such bits to be delivered and no retransmission is required. Therefore, the transmitted bit of ${\sf Tx}_i$ leaves $Q_{i \rightarrow i}$ and joins $Q_{i,F}$, $i=1,2$.

\vspace{1mm}

\item Case~5~$\left( \casefive \right)$ and Case~6~$\left( \casesix \right)$: In these cases, ${\sf Rx}_1$ gets its corresponding bit interference free. We consider this bit to be delivered and no retransmission is required. Therefore, the transmitted bit of ${\sf Tx}_1$ leaves $Q_{1 \rightarrow 1}$ and joins $Q_{1,F}$, while the transmitted bit of ${\sf Tx}_2$ remains in $Q_{2 \rightarrow 2}$.

\vspace{1mm}

\item Case~7~$\left( \caseseven \right)$: In this case, ${\sf Rx}_1$ has a linear combination of the transmitted bits, while ${\sf Rx}_2$ has not received anything, see Table~\ref{table:16cases}. It is sufficient to provide the transmitted bit of ${\sf Tx}_2$ to both receivers. Therefore, the transmitted bit of ${\sf Tx}_2$ leaves $Q_{2 \rightarrow 2}$ and joins $Q_{2 \rightarrow \{ 1,2 \}}$. Note that the bit of ${\sf Tx}_1$ will not be retransmitted since upon delivery of the bit of ${\sf Tx}_2$, ${\sf Rx}_1$ can decode its corresponding bit. This bit leaves $Q_{1 \rightarrow 1}$ and joins $Q_{1,F}$. Similar argument holds for Case~8~$\left( \caseeight \right)$.

\vspace{1mm}

\item Cases~9,10,11, and 12: Similar to Cases~5,6,7, and 8 with swapping user IDs respectively.

\vspace{1mm}

\item Case~13~$\left( \nearrow \right)$: In this case, ${\sf Rx}_1$ has received the transmitted bit of ${\sf Tx}_2$ while ${\sf Rx}_2$ has not received anything, see Table~\ref{table:16cases}. Therefore, the transmitted bit of ${\sf Tx}_1$ remains in $Q_{1 \rightarrow 1}$, while the transmitted bit of ${\sf Tx}_2$ is required by ${\sf Rx}_2$ and it is available at ${\sf Rx}_1$. Hence, the transmitted bit of ${\sf Tx}_2$ leaves $Q_{2 \rightarrow 2}$ and joins $Q_{2 \rightarrow 2|1}$. Queue $Q_{2 \rightarrow 2|1}$ represents the bits at ${\sf Tx}_2$ that are available at ${\sf Rx}_1$, but ${\sf Rx}_2$ needs them.

\vspace{1mm}

\item Case~14~$\left( \searrow \right)$: This is similar to Case~13 with swapping user IDs.

\vspace{1mm}

\item Case~15~$\left( \casefifteen \right)$: In this case, ${\sf Rx}_1$ has received the transmitted bit of ${\sf Tx}_2$ while ${\sf Rx}_2$ has received the transmitted bit of ${\sf Tx}_1$, see Table~\ref{table:16cases}. In other words, the transmitted bit of ${\sf Tx}_2$ is available at ${\sf Rx}_1$ and is required by ${\sf Rx}_2$; while the transmitted bit of ${\sf Tx}_1$ is available at ${\sf Rx}_2$ and is required by ${\sf Rx}_1$. Therefore, we have transition from $Q_{i \rightarrow i}$ to $Q_{i \rightarrow i|\bar{i}}$, $i=1,2$.

\vspace{1mm}

\item Case~16: The transmitted bit of ${\sf Tx}_i$ remains in $Q_{i \rightarrow i}$, $i=1,2$.

\end{itemize}

Phase~$1$ goes on for 
\begin{align}
\left( 1 - q^2 \right)^{-1} m + m^{\frac{2}{3}}
\end{align}
time instants, and if at the end of this phase, either of the queues $Q_{i \rightarrow i}$ is not empty, we declare error type-\cnt and halt the transmission (we assume $m$ is chosen such that $m^{\frac{2}{3}} \in \mathbb{Z}$).

Assuming that the transmission is not halted, let $N_{i,{\sf C}_1}$, $N_{i \rightarrow j|\bar{j}}$, and $N_{i \rightarrow \{ 1,2\}}$ denote the number of bits in queues $Q_{i,{\sf C}_1}$, $Q_{i \rightarrow j|\bar{j}}$, and $Q_{i \rightarrow \{ 1,2\}}$ respectively at the end of the transitions, $i=1,2$, and $j = i,\bar{i}$. The transmission strategy will be halted and an error type-\cnt will occur, if any of the following events happens.
\begin{align}
\label{eq:errortypeI}
& N_{i,{\sf C}_1} > \mathbb{E}[N_{i,{\sf C}_1}] + m^{\frac{2}{3}} \overset{\triangle}= n_{i,{\sf C}_1}, \quad i=1,2; \nonumber \\
& N_{i \rightarrow j|\bar{j}} > \mathbb{E}[N_{i \rightarrow j|\bar{j}}] + m^{\frac{2}{3}} \overset{\triangle}= n_{i \rightarrow j|\bar{j}}, i=1,2,~j = i,\bar{i}; \nonumber \\
& N_{i \rightarrow \{ 1,2 \}} > \mathbb{E}[N_{i \rightarrow \{ 1,2 \}}] + m^{\frac{2}{3}} \overset{\triangle}= n_{i \rightarrow \{ 1,2 \}}, i=1,2.
\end{align}


%
%
%
%

From basic probability, we know that
{\small \begin{align}
\label{eq:expectedvalues}
& \mathbb{E}[N_{i,{\sf C}_1}] = \frac{\Pr\left( \mathrm{Case~1} \right)m}{1 - \sum_{i=9,10,13,16}{\Pr\left( \mathrm{Case~i} \right)}}  = (1-q^2)^{-1} p^4 m, \nonumber \\
& \mathbb{E}[N_{i \rightarrow i|\bar{i}}] = \frac{\sum_{j=14,15}{\Pr\left( \mathrm{Case~j} \right)}m}{1 - \sum_{i=9,10,13,16}{\Pr\left( \mathrm{Case~i} \right)}} = (1-q^2)^{-1} p q^2 m, \nonumber \\
& \mathbb{E}[N_{i \rightarrow \bar{i}|i}] = \frac{\Pr\left( \mathrm{Case~2} \right)m}{1 - \sum_{i=9,10,13,16}{\Pr\left( \mathrm{Case~i} \right)}} = (1-q^2)^{-1} p^3 q m, \nonumber \\
& \mathbb{E}[N_{i \rightarrow \{ 1,2 \}}] = \frac{\sum_{j=11,12}{\Pr\left( \mathrm{Case~j} \right)}m}{1 - \sum_{i=9,10,13,16}{\Pr\left( \mathrm{Case~i} \right)}} = (1-q^2)^{-1} p^2 q m. 
\end{align}}

Furthermore, we can show that the probability of errors of types I and II decreases exponentially with $m$. More precisely, we use Chernoff-Hoeffding bound\footnote{We consider a specific form of the Chernoff-Hoeffding bound~\cite{Hoeffding} described in~\cite{Chernoff}, which is simpler to use and is as follows. If $X_1,\ldots,X_r$ are $r$ independent random variables, and $M=\sum_{i=1}^r{X_i}$, then $Pr\left[ |M-\mathbb{E}\left[ M \right] | > \alpha \right] \leq 2 \exp \left( \frac{-\alpha^2}{4 \sum_{i=1}^r \mathrm{Var} \left( X_i \right)} \right)$.}, to bound the error probabilities of types I and II. For instance, to bound the probability of error type-I, we have
\begin{align}
& \Pr \left[ \mathrm{error~type \noindent - \noindent I} \right] \leq  \sum_{i=1}^2{\Pr \left[ Q_{i \rightarrow i} \mathrm{~is~not~empty} \right]} \nonumber \\
&~\leq 4 \exp \left( \frac{-m^{4/3}}{4 n (1-q^2) q^2} \right) \nonumber \\
&~= 4 \exp \left( \frac{-m^{4/3}}{4 (1-q^2) q^2 \left[ \left( 1 - q^2 \right)^{-1} m + m^{\frac{2}{3}} \right]} \right),
\end{align}
which decreases exponentially to zero as $m \rightarrow \infty$.

At the end of Phase $1$, we add $0$'s (if necessary) in order to make queues $Q_{i,{\sf C}_1}$, $Q_{i \rightarrow j|\bar{j}}$, and $Q_{i \rightarrow \{ 1,2\}}$ of size equal to $n_{i,{\sf C}_1}$, $n_{i \rightarrow j|\bar{j}}$, and $n_{i \rightarrow \{ 1,2\}}$ respectively as defined in (\ref{eq:errortypeI}), $i=1,2$, and $j = i,\bar{i}$. For the rest of this subsection, we assume that Phase 1 is completed and no error has occurred.

We now use the ideas described in Section~\ref{sec:oppideas1}, to further create bits of common interest. Depending on the value of $p$, we use different ideas. We break the rest of this subsection into two parts: $(1)$ $0.5 \leq p \leq \left( \sqrt{5} - 1 \right)/2$; and $(2)$ $\left( \sqrt{5} - 1 \right)/2 < p \leq 1$. In what follows, we first describe the rest of the achievability strategy for $0.5 \leq p \leq \left( \sqrt{5} - 1 \right)/2$. In particular, we demonstrate how to incorporate the ideas of Section~\ref{sec:oppideas1} to create bits of common interest in an optimal way.
\begin{itemize}

\item {\bf Type I} Combining bits in $Q_{i \rightarrow \bar{i}|i}$ and $Q_{i \rightarrow i|\bar{i}}$: Consider the bits that were transmitted in Cases $2$ and $14$, see Fig.~\ref{fig:2and14}. Observe that if we provide $a_1 \oplus a_2$ to \emph{both} receivers then ${\sf Rx}_1$ can decode bits $a_1$ and $a_2$, whereas ${\sf Rx}_2$ can decode bit $b_1$. Therefore, $a_1 \oplus a_2$ is a bit of common interest and can join $Q_{1 \rightarrow \{ 1,2 \}}$. Hence, as illustrated in Fig.~\ref{fig:XOR}, we can remove two bits in $Q_{1 \rightarrow 2|1}$ and $Q_{1 \rightarrow 1|2}$, by inserting their XOR in $Q_{1 \rightarrow \{ 1,2 \}}$, and we deliver this bit of common interest to both receivers during the second phase. Note that due to the symmetry of the channel, similar argument holds for $Q_{2 \rightarrow 1|2}$ and $Q_{2 \rightarrow 2|1}$.


\begin{figure}[ht]
\centering
\includegraphics[height = 2.5 cm]{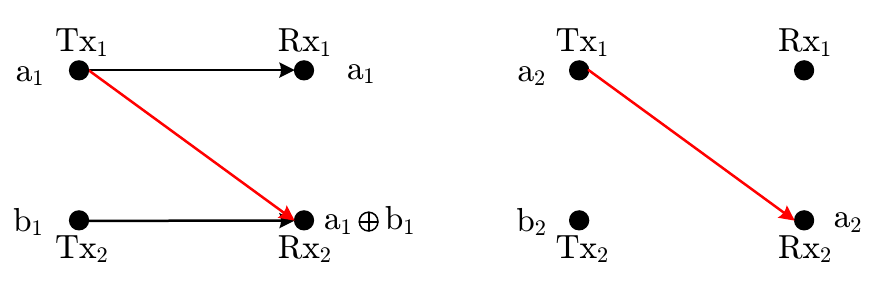}
\caption{Suppose at a time instant, transmitters one and two send out data bits $a_1$ and $b_1$ respectively, and later using Delayed-CSIT, transmitters figure out Case 2 occurred at that time. At another time instant, suppose transmitters one and two send out data bits $a_2$ and $b_2$ respectively, and later using Delayed-CSIT, transmitters figure out Case 14 occurred at that time. Now, bit $a_1 \oplus a_2$ available at ${\sf Tx}_1$ is useful for both receivers and it is a bit of common interest. Hence, $a_1 \oplus a_2$ can join $Q_{1 \rightarrow \{ 1,2 \}}$.\label{fig:2and14}}
\end{figure}

\begin{figure}[h]
\centering
\includegraphics[height = 2.75 cm]{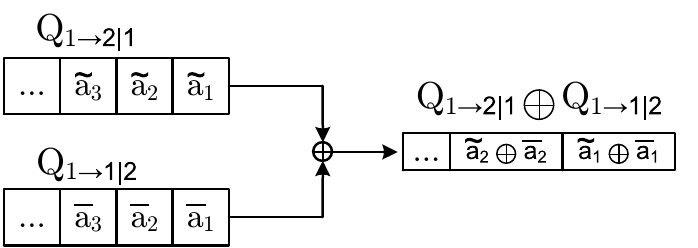}
\caption{Creating XOR of the bits in two different queues. We pick one bit from each queue and create the XOR of the two bits.\label{fig:XOR}}
\end{figure}

 
For $0.5 \leq p \leq \left( \sqrt{5} - 1 \right)/2$, we have $\mathbb{E}[N_{i \rightarrow \bar{i}|i}] \leq \mathbb{E}[N_{i \rightarrow i|\bar{i}}]$. Therefore, after this combination, queue $Q_{i \rightarrow \bar{i}|i}$ becomes empty and we have
\begin{align}
\mathbb{E}[N_{i \rightarrow i|\bar{i}}] - \mathbb{E}[N_{i \rightarrow \bar{i}|i}] =  \left( 1 - q^2 \right)^{-1} pq\left( q - p^2 \right) m
\end{align}
bits left in $Q_{i \rightarrow i|\bar{i}}$.

\item {\bf Type II} Combining the bits in $Q_{i,{\sf C}_1}$ and $Q_{i \rightarrow i|\bar{i}}$: Consider the bits that were transmitted in Cases $1$ and $15$, see Fig.~\ref{fig:1and15}. It is easy to see that providing $a_1 \oplus a_2$ and $b_1 \oplus b_2$ to both receivers is sufficient to decode their corresponding bits. For instance, ${\sf Rx}_1$ removes $b_2$ from $b_1 \oplus b_2$ to decode $b_1$, then uses $b_1$ to decode $a_1$ from $a_1 \oplus b_1$. Therefore, $a_1 \oplus a_2$ and $b_1 \oplus b_2$ are bits of common interest and can join $Q_{1 \rightarrow \{ 1,2 \}}$ and $Q_{2 \rightarrow \{ 1,2 \}}$ respectively. Hence, we can remove two bits in $Q_{i,{\sf C}_1}$ and $Q_{i \rightarrow i|\bar{i}}$, by inserting their XOR in $Q_{i \rightarrow \{ 1,2 \}}$, $i=1,2$, and then deliver this bit of common interest to both receivers during the second phase.


\begin{figure}[ht]
\centering
\includegraphics[height = 2.5 cm]{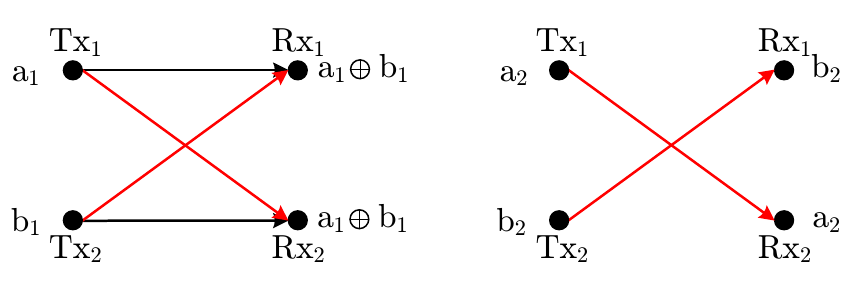}
\caption{Suppose at a time instant, transmitters one and two send out data bits $a_1$ and $b_1$ respectively, and later using Delayed-CSIT, transmitters figure out Case 1 occurred at that time. At another time instant, suppose transmitters one and two send out data bits $a_2$ and $b_2$ respectively, and later using Delayed-CSIT, transmitters figure out Case 15 occurred at that time. Now, bit $a_1 \oplus a_2$ available at ${\sf Tx}_1$ and bit $b_1 \oplus b_2$ available at ${\sf Tx}_2$ are useful for both receivers and they are bits of common interest. Therefore, bits $a_1 \oplus a_2$ and $b_1 \oplus b_2$ can join $Q_{1 \rightarrow \{ 1,2 \}}$ and $Q_{2 \rightarrow \{ 1,2 \}}$ respectively.\label{fig:1and15}}
\end{figure}


For $0.5 \leq p \leq \left( \sqrt{5} - 1 \right)/2$, we have $\left( 1 - q^2 \right)^{-1} pq\left( q - p^2 \right) m \leq \mathbb{E}[N_{i,{\sf C}_1}]$. Therefore after combining the bits, queue $Q_{i \rightarrow i|\bar{i}}$ becomes empty and we have 
\begin{align}
& \mathbb{E}[N_{i,{\sf C}_1}] + m^{\frac{2}{3}} - \left( \mathbb{E}[N_{i \rightarrow i|\bar{i}}] - \mathbb{E}[N_{i \rightarrow \bar{i}|i}] \right) \nonumber\\
&~= \left( 1 - q^2 \right)^{-1} p\left( p - q \right) m + m^{\frac{2}{3}}
\end{align}
bits left in $Q_{i,{\sf C}_1}$, $i=1,2$.

\end{itemize}

Finally, we need to describe what happens to the remaining $\left( 1 - q^2 \right)^{-1} p\left( p - q \right) m + m^{\frac{2}{3}}$ bits in $Q_{i,{\sf C}_1}$. As mentioned before, a bit in $Q_{i,{\sf C}_1}$ can be viewed as a bit of common interest by itself. For the remaining bits in $Q_{1,{\sf C}_1}$, we put the first half in $Q_{1 \rightarrow \{ 1,2 \}}$ (suppose $m$ is picked such that the remaining number of bits is even). Note that if these bits are delivered to ${\sf Rx}_2$, then ${\sf Rx}_2$ can decode the first half of the remaining bits in $Q_{2,{\sf C}_1}$ as well. Therefore, the first half of the bits in $Q_{2,{\sf C}_1}$ can join $Q_{2,F}$. 

\begin{figure*}[t]
\centering
\includegraphics[height = 3 cm]{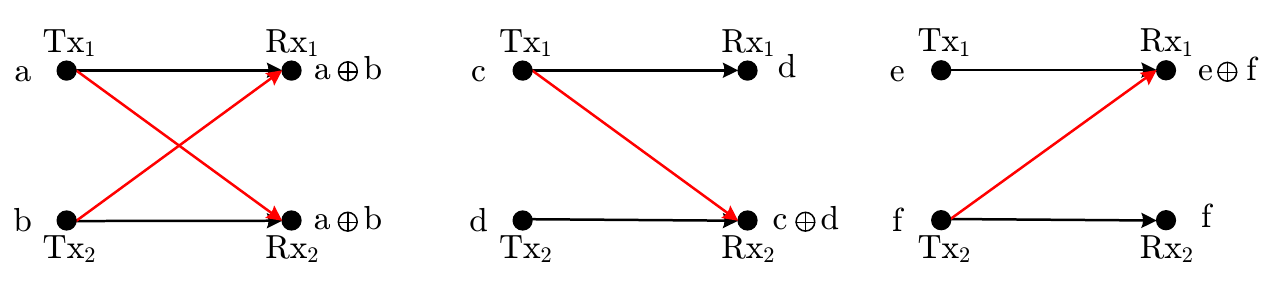}
\caption{Consider the bits transmitted in Cases 1,2, and 3. Now, bit $a \oplus c$ available at ${\sf Tx}_1$ and bit $b \oplus f$ available at ${\sf Tx}_2$ are useful for both receivers and they are bits of common interest. Therefore, bits $a \oplus c$ and $b \oplus f$ can join $Q_{1 \rightarrow \{ 1,2 \}}$ and $Q_{2 \rightarrow \{ 1,2 \}}$ respectively.\label{fig:123}}
\end{figure*}

Then, we put the second half of the remaining bits in $Q_{2,{\sf C}_1}$ in $Q_{2 \rightarrow \{ 1,2 \}}$. Similar to the argument presented above, the second half of the bits in $Q_{1,{\sf C}_1}$ join $Q_{1,F}$.

Hence at the end of Phase $1$, if the transmission is not halted, we have a total of 
\begin{align}
& (1-q^2)^{-1} \left[ \underbrace{p^2q}_{\mathrm{Cases~11~and~12}} + \underbrace{pq^2}_{\mathrm{XOR~opportunities}} \right. \nonumber \\
& \left. + 0.5 \underbrace{(p^4 - pq^2 + p^3q)}_{\mathrm{remaining~Case~1}}  \right] m + 2.5 m^{2/3} \nonumber \\
& \quad = \left( 1 - q^2 \right)^{-1} 0.5 p m + 2.5 m^{2/3}
\end{align}
number of bits in $Q_{1 \rightarrow \{1,2\}}$. Same result holds for $Q_{2 \rightarrow \{1,2\}}$.

This completes the description of Phase $1$ for $0.5 \leq p \leq \left( \sqrt{5} - 1 \right)/2$. For $\left( \sqrt{5} - 1 \right)/2 < p \leq 1$, we combine the bits as follows. Fot this range of $p$, after Phase~1, the number of bits in each queue is such that the mergings described above are not optimal and we have to rearrange them as decsribed below.
\begin{itemize}

\item {\bf Type I} Combining $Q_{i \rightarrow \bar{i}|i}$ and $Q_{i \rightarrow i|\bar{i}}$: We have already described this opportunity for $0.5 \leq p \leq \left( \sqrt{5} - 1 \right)/2$. We create the XOR of the bits in $Q_{1 \rightarrow 2|1}$ and $Q_{1 \rightarrow 1|2}$ and put the XOR of them in $Q_{1 \rightarrow \{ 1,2\}}$. Note that due to the symmetry of the channel, similar argument holds for $Q_{2 \rightarrow 1|2}$ and $Q_{2 \rightarrow 2|1}$.


For $\left( \sqrt{5} - 1 \right)/2 < p \leq 1$, $\mathbb{E}[N_{i \rightarrow i|\bar{i}}] \leq \mathbb{E}[N_{i \rightarrow \bar{i}|i}]$, $i=1,2$. Therefore after combining the bits, queue $Q_{i \rightarrow i|\bar{i}}$ becomes empty, and we have 
\begin{align}
\mathbb{E}[N_{i \rightarrow \bar{i}|i}] - \mathbb{E}[N_{i \rightarrow i|\bar{i}}] =  \left( 1 - q^2 \right)^{-1} pq\left( p^2 - q \right) m
\end{align}
bits left in $Q_{i \rightarrow \bar{i}|i}$, $i=1,2$.

\item {\bf Type III} Combining the bits in $Q_{i,{\sf C}_1}$ and $Q_{i \rightarrow \bar{i}|i}$: Consider the bits that were transmitted in Cases $1,2,$ and $3$, see Fig.~\ref{fig:123}. Now, we observe that providing $a \oplus c$ and $b \oplus f$ to both receivers is sufficient to decode their corresponding bits. For instance, ${\sf Rx}_1$ will have $a \oplus b$, $c$, $e \oplus f$, $a \oplus c$, and $b \oplus f$, from which it can recover $a,c,$ and $e$. Similar argument holds for ${\sf Rx}_2$. Therefore, $a \oplus c$ and $b \oplus f$ are bits of common interest and can join $Q_{1 \rightarrow \{ 1,2 \}}$ and $Q_{2 \rightarrow \{ 1,2 \}}$ respectively. Hence, we can remove two bits in $Q_{i,{\sf C}_1}$ and $Q_{i \rightarrow \bar{i}|i}$, by inserting their XORs in $Q_{i \rightarrow \{ 1,2 \}}$, $i=1,2$, and then deliver this bit of common interest to both receivers during the second phase.


For $\left( \sqrt{5} - 1 \right)/2 < p \leq 1$, we have $\left( 1 - q^2 \right)^{-1} pq\left( p^2 - q \right) m \leq \mathbb{E}[N_{i,{\sf C}_1}]$. Therefore after combining the bits, queue $Q_{i \rightarrow \bar{i}|i}$ becomes empty and we have 
\begin{align}
&\mathbb{E}[N_{i,{\sf C}_1}] + m^{\frac{2}{3}} - \left( \mathbb{E}[N_{i \rightarrow \bar{i}|i}] - \mathbb{E}[N_{i \rightarrow i|\bar{i}}] \right) \nonumber \\
&~= \left( 1 - q^2 \right)^{-1} \left( p^4 - p^3q + pq^2 \right) m + m^{\frac{2}{3}}
\end{align}
bits left in $Q_{i,{\sf C}_1}$.

\end{itemize}

We treat the remaining bits in $Q_{i,{\sf C}_1}$ as described before. Hence at the end of Phase $1$, if the transmission is not halted, we have a total of 
\begin{align}
& (1-q^2)^{-1} \left[ \underbrace{p^2q}_{\mathrm{Cases~11~and~12}} + \underbrace{p^3q}_{\mathrm{XOR~opportunities}} \right. \nonumber \\
& \left. + 0.5 \underbrace{(p^4 - p^3q + pq^2)}_{\mathrm{remaining~Case~1}}  \right] m + 2.5 m^{2/3} \nonumber \\
& \quad = \left( 1 - q^2 \right)^{-1} 0.5 p m + 2.5 m^{2/3}
\end{align}
number of bits in $Q_{1 \rightarrow \{1,2\}}$. Same result holds for $Q_{2 \rightarrow \{1,2\}}$. 

To summarize, at the end of Phase~1 assuming that the transmission is not halted, by using coding opportunities of types I, II, and III, we are only left with $\left( 1 - q^2 \right)^{-1} 0.5 p m + 2.5 m^{2/3}$ bits in queue $Q_{i \rightarrow \{1,2\}}$, $i=1,2$.

We now describe how to deliver the bits of common interest in Phase~2 of the transmission strategy. The problem resembles a network with two transmitters and two receivers where each transmitter ${\sf Tx}_i$ wishes to communicate an independent message $\hbox{W}_i$ to {\it both} receivers, $i=1,2$. The channel gain model is the same as described in Section~\ref{sec:problem}. We refer to this network as the two-multicast network as depicted in Fig.~\ref{fig:two-multicast}. We have the following result for this network.

\begin{figure}[ht]
\centering
\includegraphics[height = 3.5 cm]{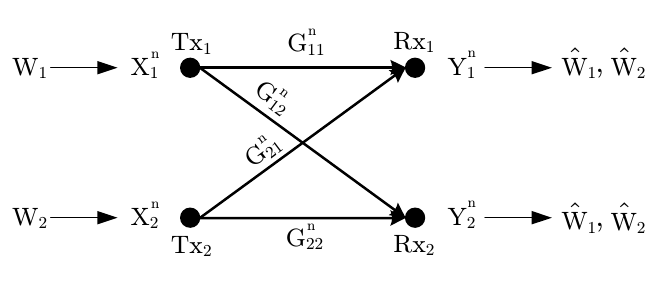}
\caption{Two-multicast network. Transmitter ${\sf Tx}_i$ wishes to reliably communicate message $\hbox{W}_i$ to both receivers, $i=1,2$. The capacity region with no, delayed, or instantaneous CSIT is the same.\label{fig:two-multicast}}
\end{figure}

\begin{lemma}
\label{lemma:multicast}
For the two-multicast network as described above, we have
\begin{equation} 
\mathcal{C}_{\mathrm{multicast}}^{\mathrm{No-CSIT}} = \mathcal{C}_{\mathrm{multicast}}^{\mathrm{DCSIT}} = \mathcal{C}_{\mathrm{multicast}}^{\mathrm{ICSIT}},
\end{equation}
and, we have
\begin{equation}
\label{}
\mathcal{C}_{\mathrm{multicast}}^{\mathrm{ICSIT}} =
\left\{ \begin{array}{ll}
\vspace{1mm} R_i \leq p, & i = 1,2,\\
R_1 + R_2 \leq 1 - q^2. &
\end{array} \right.
\end{equation}
\end{lemma}

This result basically shows that the capacity region of the two-multicast network described above is equal to the capacity region of the multiple-access channel formed at either of the receivers. The proof of Lemma~\ref{lemma:multicast} is presented in Appendix~\ref{Appendix:multicast}.

\noindent {\bf Phase 2} [transmitting bits of common interest]: In this phase, we deliver the bits in $Q_{1 \rightarrow \{ 1,2 \}}$ and $Q_{2 \rightarrow \{ 1,2 \}}$ using the transmission strategy for the two-multicast problem. More precisely, the bits in $Q_{i \rightarrow \{ 1,2 \}}$ will be considered as the message of ${\sf Tx}_i$ and they will be encoded as in the achievability scheme of Lemma~\ref{lemma:multicast}, $i=1,2$. Fix $\epsilon, \delta > 0$, from Lemma~\ref{lemma:multicast}, we know that rate tuple $$\left( R_1, R_2 \right) = \frac{1}{2}\left( (1-q^2) - \delta, (1-q^2) - \delta \right)$$ is achievable with decoding error probability less than or equal to $\epsilon$. Therefore, transmission of the bits in $Q_{1 \rightarrow \{ 1,2 \}}$ and $Q_{2 \rightarrow \{ 1,2 \}}$, will take
\begin{align} 
t_{\mathrm{total}} = \frac{\left( 1 - q^2 \right)^{-1} p m + 5 m^{2/3}}{(1-q^2) - \delta} \raisebox{2pt}{.}
\end{align}

Therefore, the total transmssion time of our two-phase achievability strategy is equal to
\begin{align}
(1-q^2)^{-1} m + m^{\frac{2}{3}} + \frac{\left( 1 - q^2 \right)^{-1} p m + 5 m^{2/3}}{(1-q^2) - \delta} \raisebox{2pt}{,}
\end{align}
hence, if we let $\epsilon,\delta \rightarrow 0$ and $m \rightarrow \infty$, the decoding error probability of delivering bits of common interest goes to zero, and we achieve a symmetric sum-rate of 
\begin{align}
R_1 = R_2 = \lim_{\substack{\epsilon,\delta \rightarrow 0 \\ m \rightarrow \infty}}{\frac{m}{t_{\mathrm{total}}}} = \frac{(1-q^2)}{1+(1-q^2)^{-1}p} \raisebox{2pt}{.}
\end{align}

This completes the achievability proof of point $A$ for $0.5 \leq p \leq 1$.

%
%
%
%
%

\subsection{Overview of the Achievability Strategy for Corner Point $C$}

We now provide an overview of the achievability strategy for corner point $C$ depicted in Fig.~\ref{fig:regionHalf} for $\left( 3 - \sqrt{5} \right)/2 < p \leq 1$, \emph{i.e.}
\begin{align}
\left( R_1, R_2 \right) = \left( pq(1+q), p \right) \raisebox{2pt}{,}
\end{align}
and we postpone the detailed proof to Appendix~\ref{Appendix:cornerdelayed}.

Compared to the achievability strategy of the sum-rate point, the challenges in achieving the other corner points arise from the asymmetricity of the rates. At this corner point, while ${\sf Tx}_2$ (the primary user) communicates at full rate of $p$, ${\sf Tx}_1$ (the secondary user) communicates at a lower rate and tries to coexist with the primary user. The achievability strategy is based on the following two principles.
\begin{enumerate}

\item If the secondary user creates interference at the primary receiver, it is the secondary user's responsibility to resolve this interference;

\item For the achievability of the optimal sum-rate point $A$ (see Fig.~\ref{fig:regionHalf}(b)), the bits of common interest were transmitted such that both receivers could decode them. However, for corner point $C$, when the primary receiver obtains a bit of common interest, we revise the coding scheme in a way that favors the primary receiver.

\end{enumerate}

Our transmission strategy consists of five phases as summarized below.
\begin{itemize}

\item Phase 1 [uncategorized transmission]: This phase is similar to Phase 1 of the achievability of the optimal sum-rate point $A$. The main difference is due to the fact that the transmitters have unequal number of bits at the beginning. In Phase $1$, ${\sf Tx}_1$ (the secondary user) transmits all its initial bits while ${\sf Tx}_2$ (the primary user) only transmits part of its initial bits. Transmitter two postpones the transmission of its remaining bits to Phase 3.

\item Phase 2 [updating status of the bits transmitted when either of the Cases $7$ or $8$ ($11$ or $12$) occurred]: For the achievability of optimal sum-rate point $A$, we transferred the transmitted bits of ${\sf Tx}_2$ (${\sf Tx}_1$) to the two-multicast sub-problem by viewing them as bits of common interest. However, this scheme turns out to be suboptimal for corner point $C$. In this case, we retransmit these bits during Phase~2 and update their status based on the channel realization at the time of transmission. Phase 2 provides coding opportunities that we exploit in Phases 4 and 5.

\item Phase 3 [uncategorized transmission vs interference management]: In this phase, the primary user transmits the remaining initial bits while the secondary user tries to resolve as much interference as it can at the primary receiver. To do so, the secondary user sends the bits that caused interference at the primary receiver during Phase 1, at a rate low enough such that both receivers can decode and remove them regardless of what the primary transmitter does. Note that $pq$ of the time, each receiver gets interference-free signal from the secondary transmitter, hence, the secondary transmitter can take advantage of these time instants to deliver its bits during Phase~3. 

\item Phases 4 and 5 [delivering interference-free bits and interference management]: In the final phases, each transmitter has two main objectives: $(1)$ communicating the bits required by its own receiver but available at the unintended receiver; and $(2)$ mitigating interference at the unintended receiver. This task can be accomplished by creating the XOR of the bits similar to coding type-I described in Section~\ref{sec:opportunities} with a modification: we first encode these bits and then create the XOR of the encoded bits. Moreover, the balance of the two objectives is different between the primary user and the secondary user. 

\end{itemize}

As mentioned before, the detailed proof of the achievability for corner point $C$ is provided in Appendix~\ref{Appendix:cornerdelayed}. In the following section, we describe the converse proof for the two-user BFIC with Delayed-CSIT.




\section{Converse Proof of Theorem~\ref{THM:IC-DelayedCSIT}~[Delayed-CSIT]}
\label{sec:outerHalf}


In this section, we provide the converse proof for Theorem~\ref{THM:IC-DelayedCSIT}. As mentioned in Remark~\ref{remark:OFBConverse}, The outer-bound on the capacity region with only Delayed-CSIT~(\ref{eq:DelayedNSIregion}) is in fact the intersection of the outer-bounds on the individual rates (\emph{i.e.} $R_i \leq p$, $i=1,2$) and the capacity region with Delayed-CSIT and OFB~(\ref{eq:capacity-FB}). Therefore, the converse proof that we will later present in Section~\ref{sec:conversedelayedhalfFB} for the case of Delayed-CSIT and OFB suffices. However, specific challenges arise when OFB is present and careful considerations must be taken into account. Here, we independently present the converse proof of Theorem~\ref{THM:IC-DelayedCSIT} to highlight the key techniques without worrying about the details needed for the case of OFB.

We first present the Entropy Leakage Lemma that plays a key role in deriving the converse. Consider the scenario where a transmitter is connected to two receivers through binary fading channels as in Fig.~\ref{fig:portionHalf}. Suppose $G_1[t]$ and $G_2[t]$ are distributed as i.i.d. Bernoulli RVs (\emph{i.e.} $G_i[t] \overset{d}\sim \mathcal{B}(p)$), $i=1,2$. In this channel the received signals are given as
\begin{align}
Y_i[t] = G_i[t] X[t], \qquad i = 1,2,
\end{align}
where $X[t]$ is the transmit signal at time $t$. We have the following lemma.

\begin{figure}[ht]
\centering
\includegraphics[height = 3.5cm]{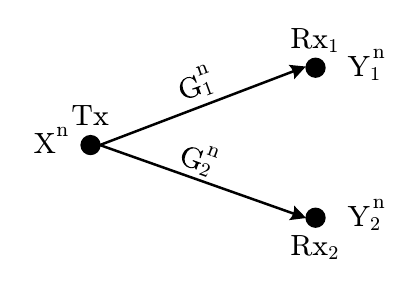}
\caption{A transmitter connected to two receivers through binary fading channels. Using Delayed-CSIT, the transmitter can privilege reveiver one to receiver two. Lemma~\ref{lemma:portion} formalizes this privilege.\label{fig:portionHalf}}
\end{figure}

\begin{lemma}
\label{lemma:portion}
{\bf [Entropy Leakage]} For the channel described above with Delayed-CSIT and for \emph{any} input distribution, we have
\begin{align}
H\left( Y_2^n | G^n \right) \geq \frac{1}{2-p} H\left( Y_1^n | G^n \right).
\end{align}
\end{lemma}

\begin{remark}
Note that with No-CSIT, from the transmitter's point of view, the two receivers are identical and it cannot favor one over the other and as a result, we have $H\left( Y_2^n | G^n \right) = H\left( Y_1^n | G^n \right)$. With Instantaneous-CSIT this ratio can become zero\footnote{This can be done by simply remaining silent whenever $G_2[t] = 1$.}. Therefore, this lemma captures the effect of Delayed-CSIT on the entropy of the received signals at the two receivers.
\end{remark}

\begin{proof} 
For time instant $t$ where $1 \leq t \leq n$, we have
\begin{align}
& H\left( Y_2[t] | Y_2^{t-1}, G^t \right) \nonumber \\
& \quad = p H\left( X[t] | Y_2^{t-1}, G_2[t] = 1, G^{t-1} \right) \nonumber \\
& \quad \overset{(a)}= p H\left( X[t] | Y_2^{t-1}, G^t \right) \nonumber \\
& \quad \overset{(b)}\geq p H\left( X[t] | Y_1^{t-1},Y_2^{t-1}, G^t \right) \nonumber \\
& \quad \overset{(c)}= \frac{p}{1-q^2} H\left( Y_1[t], Y_2[t] | Y_1^{t-1},Y_2^{t-1}, G^t \right),
\end{align}
where $(a)$ holds since $X[t]$ is independent of the channel realization at time instant $t$; $(b)$ follows from the fact that conditioning reduces entropy; and $(c)$ follows from the fact that $\Pr\left[ G_1[t] = G_2[t] = 0 \right] = q^2$. Therefore, we have 
\begin{align}
& \sum_{t=1}^n{H\left( Y_2[t] | Y_2^{t-1}, G^t \right)} \nonumber \\
&~\geq \frac{1}{2-p} \sum_{t=1}^n{H\left( Y_1[t], Y_2[t] | Y_1^{t-1}, Y_2^{t-1}, G^t \right)},
\end{align}
and since the transmit signals at time instant $t$ are independent from the channel realizations in future time instants, we have
\begin{align}
& \sum_{t=1}^n{H\left( Y_2[t] | Y_2^{t-1}, G^n \right)} \nonumber \\
&~\geq \frac{1}{2-p} \sum_{t=1}^n{H\left( Y_1[t], Y_2[t] | Y_1^{t-1}, Y_2^{t-1}, G^n \right)},
\end{align}
hence, we get
\begin{align}
H\left( Y_2^n | G^n \right) \geq \frac{1}{2-p} H\left( Y_1^n, Y_2^n | G^n \right) \geq \frac{1}{2-p} H\left( Y_1^n | G^n \right).
\end{align}

This completes the proof of the lemma.
\end{proof}

We now derive the converse for Theorem~\ref{THM:IC-DelayedCSIT}. The outer-bound on $R_i$ is the same under no, delayed, and instantaneous CSIT, and we present it in Appendix~\ref{sec:ConvInst}. In this section, we provide the proof of
\begin{align}
R_i + (1+q) R_{\bar{i}} \leq p(1+q)^2, \qquad i = 1,2.
\end{align}

By symmetry, it is sufficient to prove it for $i=1$. Let $\beta = (1+q)$, and suppose rate tuple $\left( R_1, R_2 \right)$ is achievable. Then we have
\begin{align}
n &\left( R_1 + \beta R_2 \right) = H(W_1) + \beta H(W_2) \nonumber \\
& \overset{(a)}= H(W_1|W_2, G^n) + \beta H(W_2|G^n) \nonumber \\
& \overset{(\mathrm{Fano})}\leq I(W_1;Y_1^n|W_2,G^n) + \beta I(W_2;Y_2^n|G^n) + n \epsilon_n \nonumber \\
& = H(Y_1^n|W_2,G^n) - \underbrace{H(Y_1^n|W_1,W_2,G^n)}_{=~0} \nonumber 
\end{align}
\begin{align}
& \quad + \beta H(Y_2^n|G^n) - \beta H(Y_2^n|W_2,G^n) + n \epsilon_n \nonumber \\
& \overset{(b)}= \beta H(Y_2^n|G^n) + H(Y_1^n|W_2,X_2^n,G^n) \nonumber \\
& \quad - \beta H(Y_2^n|W_2,X_2^n,G^n) + n \epsilon_n \nonumber \\
& = \beta H(Y_2^n|G^n) + H(G_{11}^nX_1^n|W_2,X_2^n,G^n) \nonumber \\
& \quad - \beta H(G_{12}^nX_1^n|W_2,X_2^n,G^n) + n \epsilon_n \nonumber \\
& \overset{(c)}= \beta H(Y_2^n|G^n) + H(G_{11}^nX_1^n|W_2,G^n) \nonumber \\
& \quad - \beta H(G_{12}^nX_1^n|W_2,G^n) + n \epsilon_n \nonumber \\
& \overset{(d)}= \beta H(Y_2^n|G^n) + H(G_{11}^nX_1^n|G^n) \nonumber \\
& \quad - \beta H(G_{12}^nX_1^n|G^n) + n \epsilon_n \nonumber \\
& \overset{\textrm{Lemma}~\ref{lemma:portion}}\leq \beta H(Y_2^n|G^n) + n \epsilon_n \nonumber \\
& = \beta \sum_{t=1}^n{H(Y_2[t]|Y_2^{t-1},G^n)} + n \epsilon_n \nonumber \\
& \overset{(e)}\leq \beta \sum_{t=1}^n{H(Y_2[t]|G^n)} + n \epsilon_n \nonumber \\
& \overset{(f)}\leq n \beta (1-q^2) + \epsilon_n = n p(1+q)^2 + n \epsilon_n.
\end{align}
where $(a)$ holds since $\hbox{W}_1$, $\hbox{W}_2$ and $G^n$ are mutually independent; $(b)$ and $(c)$ hold since $X_2^n$ is a deterministic function of $W_2$ and $G^n$; $(d)$ follows from
\begin{align}
\label{eq:addX2}
& 0 \leq H(G_{11}^nX_1^n|G^n) - H(G_{11}^nX_1^n|W_2,G^n) \nonumber \\
& = I\left( G_{11}^nX_1^n; W_2|G^n\right) \leq I\left( W_1,G_{11}^nX_1^n; W_2|G^n\right) \nonumber \\
& = \underbrace{I\left( W_1 ; W_2|G^n\right)}_{=~0 \mathrm{~since~W_1 \perp W_2 \perp G^n}} + \underbrace{I\left( G_{11}^nX_1^n; W_2| W_1, G^n\right)}_{=~0 \mathrm{~since~X_1^n=f_1(W_1,~G^n)}} = 0,
\end{align}
which implies $H(G_{11}^nX_1^n|G^n) = H(G_{11}^nX_1^n|W_2,G^n)$, and similarly $H(G_{12}^nX_1^n|G^n) = H(G_{12}^nX_1^n|W_2,G^n)$; $(e)$ is true since conditioning reduces entropy; and $(f)$ holds since the probability that at least one of the links connected to ${\sf Rx}_2$ is equal to $1$ at each time instant is $(1-q^2)$. Dividing both sides by $n$ and let $n \rightarrow \infty$, we get
\begin{align}
R_1 + (1+q) R_2 \leq p(1+q)^2.
\end{align}

This completes the converse proof for Theorem~\ref{THM:IC-DelayedCSIT}.



\section{Achievability Proof of Theorem~\ref{THM:IC-DCSITFB}~[Delayed-CSIT and OFB]}
\label{sec:AchDelayedFB}


We now focus on the effect of the output feedback in the presence of Delayed-CSIT. In particular, we demonstrate how output feedback can be utilized to further improve the achievable rates. The capacity region of the two-user BFIC with Delayed-CSIT and OFB is given by
\begin{align}
& \mathcal{C}^{\mathrm{DCSIT,OFB}} = \\
& \left\{ R_1, R_2 \in \mathbb{R}^+~s.t.~R_i + (1+q) R_{\bar{i}} \leq p(1+q)^2,~i=1,2 \right\}, \nonumber 
\end{align}
and is depicted in Fig.~\ref{fig:region-FB}.

\begin{figure}[ht]
\centering
\includegraphics[height = 4 cm]{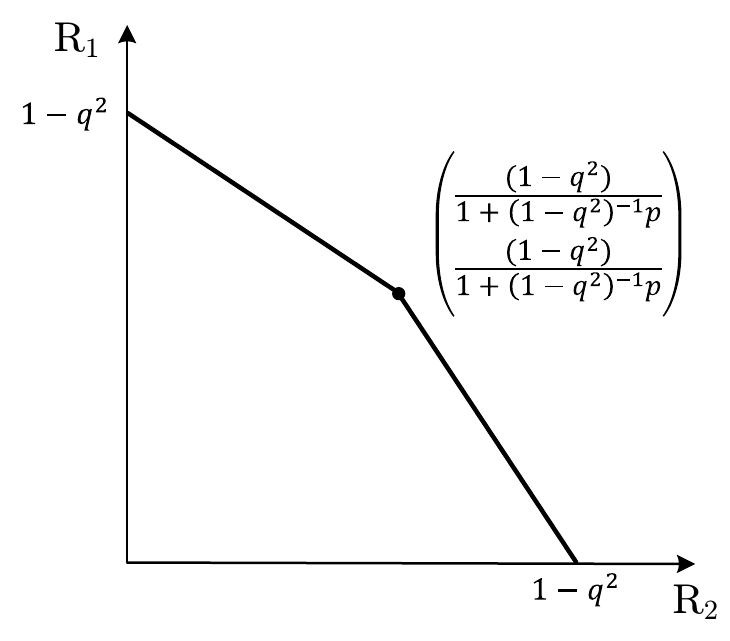}
\caption{Capacity region of the two-user BFIC with Delayed-CSIT and output feedback.\label{fig:region-FB}}
\end{figure}

The achievability strategy of the corner points $\left( 1- q^2, 0 \right)$ and $\left( 0, 1- q^2 \right)$, is based on utilizing the additional communication paths created by the means of the output feedback links, \emph{e.g.}, $${\sf Tx}_1 \rightarrow {\sf Rx}_2 \rightarrow {\sf Tx}_2 \rightarrow {\sf Rx}_1,$$
and is presented in Appendix~\ref{Appendix:cornerDelayedOFB}. Here, we only describe the transmission strategy for the sum-rate point, \emph{i.e.}
\begin{align}
\label{eq:cornerFB}
R_1 = R_2 = \frac{(1-q^2)}{1+(1-q^2)^{-1}p} \raisebox{2pt}{.}
\end{align}

Let the messages of transmitters one and two be denoted by $\hbox{W}_1 = a_1,a_2,\ldots,a_m$, and $\hbox{W}_2 = b_1,b_2,\ldots,b_m$, respectively, where data bits $a_i$'s and $b_i$'s are picked uniformly and independently from $\{ 0,1 \}$, $i=1,\ldots,m$. We show that it is possible to communicate these bits in 
\begin{align}
n = \left( 1 - q^2 \right)^{-1}m + \left( 1 - q^2 \right)^{-2} pm + O\left( m^{2/3}\right)
\end{align}
time instants with vanishing error probability (as $m \rightarrow \infty$). Therefore achieving the rates given in (\ref{eq:cornerFB}) as $m \rightarrow \infty$. Our transmission strategy consists of two phases as described below.

\noindent {\bf Phase 1} [uncategorized transmission]: This phase is identical to Phase 1 of Section~\ref{sec:achallhalf}. At the beginning of the communication block, we assume that the bits at ${\sf Tx}_i$ are in queue $Q_{i \rightarrow i}$, $i=1,2$. At each time instant $t$, ${\sf Tx}_i$ sends out a bit from $Q_{i \rightarrow i}$, and this bit will either stay in the initial queue or transition to a new queue will take place. The transitions are identical to what we have already described in Table~\ref{table:16cases}, therefore, we are not going to repeat them here. Phase~$1$ goes on for 
\begin{align}
\left( 1 - q^2 \right)^{-1} m + m^{\frac{2}{3}}
\end{align}
time instants and if at the end of this phase, either of the queues $Q_{1 \rightarrow 1}$ or $Q_{2 \rightarrow 2}$ is not empty, we declare error type-I and halt the transmission.

The transmission strategy will be halted and an error type-II will occur, if any of the following events happens.
\begin{align}
\label{eq:errortypeIOFB}
& N_{i,{\sf C}_1} > \mathbb{E}[N_{i,{\sf C}_1}] + m^{\frac{2}{3}} \overset{\triangle}= n_{i,{\sf C}_1}, \quad i=1,2; \nonumber \\
& N_{i \rightarrow j|\bar{j}} > \mathbb{E}[N_{i \rightarrow j|\bar{j}}] + m^{\frac{2}{3}} \overset{\triangle}= n_{i \rightarrow j|\bar{j}}, i=1,2,~j = i,\bar{i}; \nonumber \\
& N_{i \rightarrow \{ 1,2 \}} > \mathbb{E}[N_{i \rightarrow \{ 1,2 \}}] + m^{\frac{2}{3}} \overset{\triangle}= n_{i \rightarrow \{ 1,2 \}}, i=1,2.
\end{align}
%
%
%
%
%

From basic probability, we know that
{\small \begin{align}
& \mathbb{E}[N_{i,{\sf C}_1}] = \frac{\Pr\left( \mathrm{Case~1} \right)m}{1 - \sum_{i=9,10,13,16}{\Pr\left( \mathrm{Case~i} \right)}} = (1-q^2)^{-1} p^4 m, \nonumber
\end{align}}
{\small \begin{align} 
\label{eq:expectedvaluesOFB}
& \mathbb{E}[N_{i \rightarrow i|\bar{i}}] = \frac{\sum_{j=14,15}{\Pr\left( \mathrm{Case~j} \right)}m}{1 - \sum_{i=9,10,13,16}{\Pr\left( \mathrm{Case~i} \right)}} = (1-q^2)^{-1} p q^2 m, \nonumber \\
& \mathbb{E}[N_{i \rightarrow \bar{i}|i}] = \frac{\Pr\left( \mathrm{Case~2} \right)m}{1 - \sum_{i=9,10,13,16}{\Pr\left( \mathrm{Case~i} \right)}} = (1-q^2)^{-1} p^3 q m, \nonumber \\
& \mathbb{E}[N_{i \rightarrow \{ 1,2 \}}] = \frac{\sum_{j=11,12}{\Pr\left( \mathrm{Case~j} \right)}m}{1 - \sum_{i=9,10,13,16}{\Pr\left( \mathrm{Case~i} \right)}} = (1-q^2)^{-1} p^2 q m.
\end{align}}

Using Chernoff-Hoeffding bound, we can show that the probability of errors of types I and II decreases exponentially with $m$.

At the end of Phase $1$, we add $0$'s (if necessary) in order to make queues $Q_{i,{\sf C}_1}$, $Q_{i \rightarrow j|\bar{j}}$, and $Q_{i \rightarrow \{ 1,2\}}$ of size equal to $n_{i,{\sf C}_1}$, $n_{i \rightarrow j|\bar{j}}$, and $n_{i \rightarrow \{ 1,2\}}$ respectively as defined in (\ref{eq:errortypeIOFB}), $i=1,2$, and $j = i,\bar{i}$. For the rest of this subsection, we assume that Phase 1 is completed and no error has occurred. We now use the ideas described in Section~\ref{sec:opportunities} for output feedback, to further create bits of common interest.

$\bullet$ Updating the status of bits in $Q_{i,{\sf C}_1}$ to bits of common interest: A bit in $Q_{i,{\sf C}_1}$ can be considered as a bit of common interest. Also note that it is sufficient to deliver only one of the two bits transmitted simultaneously during Case~1. Therefore, ${\sf Tx}_1$ updates the status of the first half of the bits in $Q_{1,{\sf C}_1}$ to $Q_{1 \rightarrow \{ 1,2 \}}$, whereas ${\sf Tx}_2$ updates the status of the second half of the bits in $Q_{2,{\sf C}_1}$ to $Q_{2 \rightarrow \{ 1,2 \}}$. Hence, after updating the status of bits in $Q_{i,{\sf C}_1}$, we have
\begin{align}
(1-q^2)^{-1} \left[ p^2 q + \frac{1}{2} p^4 \right] m + \frac{3}{2} m^{\frac{2}{3}}
\end{align}
bits in $Q_{i \rightarrow \{ 1,2 \}}$, $i=1,2$.

$\bullet$ Upgrading ICs with side information to a two-multicast problem using OFB: Note that through the output feedback links, each transmitter has access to the transmitted bits of the other user during Phase 1. As described above, there are $\mathbb{E}[N_{i \rightarrow i|\bar{i}}] + m^{\frac{2}{3}}$ bits in $Q_{i \rightarrow i|\bar{i}}$ at the end of Phase 1. Now, ${\sf Tx}_1$ creates the XOR of the first half of the bits in $Q_{1 \rightarrow 1|2}$ and $Q_{2 \rightarrow 2|1}$ and updates the status of the resulting bits to $Q_{1 \rightarrow \{ 1,2 \}}$. Note that as described in Example~6 of Section~\ref{sec:opportunities}, the XOR of these bits is a bit of common interest. On the other hand, ${\sf Tx}_2$ creates the XOR of the second half of the bits in $Q_{1 \rightarrow 1|2}$ and $Q_{2 \rightarrow 2|1}$ and updates the status of the resulting bits to $Q_{2 \rightarrow \{ 1,2 \}}$. Thus, we have 
\begin{align}
(1-q^2)^{-1} \left[ p^2 q + \frac{1}{2} p^4 + \frac{1}{2} p q^2  \right] m + 2 m^{\frac{2}{3}}
\end{align}
bits in $Q_{i \rightarrow \{ 1,2 \}}$, $i=1,2$.

$\bullet$ Upgrading ICs with side information and swapped receivers to a two-multicast problem using OFB: As described above, there are $\mathbb{E}[N_{i \rightarrow \bar{i}|i}] + m^{\frac{2}{3}}$ bits in $Q_{i \rightarrow \bar{i}|i}$, ${\sf Tx}_1$ creates the XOR of the first half of the bits in $Q_{1 \rightarrow 2|1}$ and $Q_{2 \rightarrow 1|2}$ and updates the status of the resulting bits to $Q_{1 \rightarrow \{ 1,2 \}}$. Note that as described in Example~7 of Section~\ref{sec:opportunities}, the XOR of these bits is a bit of common interest. On the other hand, ${\sf Tx}_2$ creates the XOR of the second half of the bits in $Q_{1 \rightarrow 2|1}$ and $Q_{2 \rightarrow 1|2}$ and updates the status of the resulting bits to $Q_{2 \rightarrow \{ 1,2 \}}$. Hence, we have 
\begin{align}
& (1-q^2)^{-1} \left[ p^2 q + \frac{1}{2} p^4 + \frac{1}{2} p q^2 + \frac{1}{2} p^3 q \right] m + \frac{5}{2} m^{\frac{2}{3}} \nonumber \\
&~= (1-q^2)^{-1} \frac{p}{2} m + \frac{5}{2} m^{\frac{2}{3}}
\end{align}
bits in $Q_{i \rightarrow \{ 1,2 \}}$, $i=1,2$. This completes the description of Phase $1$.

\noindent {\bf Phase 2} [transmitting bits of common interest]: In this phase, we deliver the bits in $Q_{1 \rightarrow \{ 1,2 \}}$ and $Q_{2 \rightarrow \{ 1,2 \}}$ using the transmission strategy for the two-multicast problem. More precisely, the bits in $Q_{i \rightarrow \{ 1,2 \}}$ will be considered as the message of ${\sf Tx}_i$ and they will be encoded as in the achievability scheme of Lemma~\ref{lemma:multicast}, $i=1,2$. Fix $\epsilon, \delta > 0$, from Lemma~\ref{lemma:multicast} we know that the rate tuple $$\left( R_1, R_2 \right) = \frac{1}{2}\left( (1-q^2) - \delta/2, (1-q^2) - \delta/2 \right)$$ is achievable with decoding error probability less than or equal to $\epsilon$. Therefore, transmission of the bits in $Q_{1 \rightarrow \{ 1,2 \}}$ and $Q_{2 \rightarrow \{ 1,2 \}}$ requires
\begin{align} 
t_{\mathrm{total}} = \frac{\left( 1 - q^2 \right)^{-1} p m + 5 m^{2/3}}{(1-q^2) - \delta}
\end{align}
time instants. Therefore, the total transmission time of our two-phase achievability strategy is equal to
\begin{align}
(1-q^2)^{-1} m + m^{\frac{2}{3}} + \frac{\left( 1 - q^2 \right)^{-1} p m + 5 m^{2/3}}{(1-q^2) - \delta} \raisebox{2pt}{.}
\end{align}

The probability that the transmission strategy halts at any point can be bounded by the summation of error probabilities of types I and II, and the probability that an error occurs in decoding the encoded bits. This probability approaches zero for $\epsilon, \delta \rightarrow 0$ and $m \rightarrow \infty$.

Hence, if we let $\epsilon,\delta \rightarrow 0$ and $m \rightarrow \infty$, the decoding error probability goes to zero, and we achieve a symmetric sum-rate of 
\begin{align}
R_1 = R_2 = \lim_{\substack{\epsilon,\delta \rightarrow 0 \\ m \rightarrow \infty}}{\frac{m}{t_{\mathrm{total}}}} = \frac{(1-q^2)}{1+(1-q^2)^{-1}p}  \raisebox{2pt}{.}
\end{align}



\section{Converse Proof of Theorem~\ref{THM:IC-DCSITFB}~[Delayed-CSIT and OFB]}
\label{sec:conversedelayedhalfFB}


In this section, we prove the converse for Theorem~\ref{THM:IC-DCSITFB}. Suppose rate tuple $\left( R_1, R_2 \right)$ is achievable, then by letting $\beta = 1+q$, we have
\begin{align}
n & \left( R_1 + \beta R_2 \right) = H(W_1) + \beta H(W_2) \nonumber \\
& \overset{(a)}= H(W_1|W_2,G^n) + \beta H(W_2|G^n) \nonumber \\
& \overset{\mathrm{Fano}}\leq I(W_1;Y_1^n|W_2,G^n) + \beta I(W_2;Y_2^n|G^n) + n \epsilon_n \nonumber \\
& \leq I(W_1;Y_1^n,Y_2^n|W_2,G^n) + \beta I(W_2;Y_2^n|G^n) + n \epsilon_n \nonumber \\
& = H(Y_1^n,Y_2^n|W_2,G^n) - \underbrace{H(Y_1^n,Y_2^n|W_1,W_2,G^n)}_{=~0} \nonumber \\ 
&~+ \beta H(Y_2^n|G^n) - \beta H(Y_2^n|W_2,G^n) + n \epsilon_n \nonumber 
\end{align}
\begin{align}
\label{eq:converseDelayedOFB}
& = \beta H(Y_2^n|G^n) + H(Y_1^n,Y_2^n|W_2,G^n) \nonumber \\
&~- \beta H(Y_2^n|W_2,G^n) + n \epsilon_n \nonumber \\
& = \beta H(Y_2^n|G^n) + \sum_{t=1}^n{H(Y_1[t],Y_2[t]|W_2,Y_1^{t-1},Y_2^{t-1},G^n)} \nonumber \\
&~- \beta \sum_{t=1}^n{H(Y_2[t]|W_2,Y_2^{t-1},G^n)} + n \epsilon_n \nonumber \\
& \overset{(b)}\leq \beta H(Y_2^n|G^n) \nonumber \\
&~+ \sum_{t=1}^n{H(Y_1[t],Y_2[t]|W_2,Y_1^{t-1},Y_2^{t-1},X_2^t,G^n)} \nonumber \\
&~- \beta \sum_{t=1}^n{H(Y_2[t]|W_2,Y_2^{t-1},X_2^t,G^n)} + n \epsilon_n \nonumber \\
& = \beta H(Y_2^n|G^n) + \sum_{t=1}^n H\left(G_{11}[t]X_1[t],G_{12}[t]X_1[t]|W_2, \right. \nonumber \\
& \qquad \qquad \left. G_{11}^{t-1}X_1^{t-1},G_{12}^{t-1}X_1^{t-1},X_2^t,G^n \right) \nonumber \\
& \quad - \beta \sum_{t=1}^n{H(G_{12}[t]X_1[t]|W_2,G_{12}^{t-1}X_1^{t-1},X_2^t,G^n)} + n \epsilon_n \nonumber \\
& \overset{(c)}= \beta H(Y_2^n|G^n) + \sum_{t=1}^n H\left(G_{11}[t]X_1[t],G_{12}[t]X_1[t]|W_2, \right. \nonumber\\
& \qquad \qquad \left. G_{11}^{t-1}X_1^{t-1},G_{12}^{t-1}X_1^{t-1},X_2^t,G^t \right) \nonumber \\
& \quad - \beta \sum_{t=1}^n{H(G_{12}[t]X_1[t]|W_2,G_{12}^{t-1}X_1^{t-1},X_2^t,G^t)} + n \epsilon_n \nonumber \\
& \overset{(d)}\leq \beta H(Y_2^n|G^n) + n \epsilon_n \nonumber \\
& \leq p(1+q)^2 n + n \epsilon_n,
\end{align}
where $(a)$ holds since the channel gains and the messages are mutually independent; $(b)$ follows from the fact that $X_2^t$ is a deterministic function of $\left( W_2,Y_2^{t-1} \right)$\footnote{We have also added $Y_1^{t-1}$ in the condition for the scenario in which output feedback links are available from each receiver to both transmitters.} and the fact that conditioning reduces entropy; $(c)$ follows from the fact that condition on $\hbox{W}_2$, $X_1^{t-1}$, $X_2^t$, $X_1[t]$ is independent of the channel realization at future time instants, hence, we can replace $G^n$ by $G^t$; and $(d)$ follows from Lemma~\ref{lemma:OFBHalf} below. Dividing both sides by $n$ and let $n \rightarrow \infty$, we get 
\begin{align}
R_1 + (1+q) R_2 \leq p(1+q)^2.
\end{align}

Similarly, we can get $(1+q) R_1 + R_2 \leq p(1+q)^2$.

\begin{lemma}
\label{lemma:OFBHalf}
\begin{align}
& \sum_{t=1}^n{H(G_{12}[t]X_1[t]|W_2,G_{12}^{t-1}X_1^{t-1},X_2^t,G^t)} \nonumber \\
& \quad \geq \frac{1}{2-p} \sum_{t=1}^n H\left(G_{11}[t]X_1[t],G_{12}[t]X_1[t]|W_2, \right. \nonumber \\
& \qquad \qquad \left. G_{11}^{t-1}X_1^{t-1},G_{12}^{t-1}X_1^{t-1},X_2^t,G^t \right).
\end{align}
\end{lemma}

\begin{remark}
Lemma~\ref{lemma:OFBHalf} is the counterpart of Lemma~\ref{lemma:portion} when Output Feedback is available. Note that in the condition we have $X_2^t$, and due to the presence of output feedback the proof is different than that of Lemma~\ref{lemma:portion}.
\end{remark}

\begin{proof}
%
%
We have
%
\begin{align}
& H(G_{12}[t]X_1[t]|W_2,G_{12}^{t-1}X_1^{t-1},X_2^t,G^t) \nonumber \\
& \quad = p H(X_1[t]|G_{12}[t] = 1,W_2,G_{12}^{t-1}X_1^{t-1},X_2^t,G^{t-1}) \nonumber \\
& \quad \overset{(a)}= p H(X_1[t]|W_2,G_{12}^{t-1}X_1^{t-1},X_2^t,G^{t-1}) \nonumber \\
& \quad \overset{(b)}= p H(X_1[t]|W_2,G_{12}^{t-1}X_1^{t-1},X_2^t,G^{t}) \nonumber \\
& \quad = \frac{p}{1-q^2} H\left(G_{11}[t]X_1[t], G_{12}[t]X_1[t]|W_2, \right. \nonumber \\
& \qquad \qquad \left. G_{12}^{t-1}X_1^{t-1},X_2^t,G^t \right) \nonumber \\
& \quad \overset{(c)}\geq \frac{1}{2-p} H\left(G_{11}[t]X_1[t], G_{12}[t]X_1[t]|W_2, \right. \nonumber \\
& \qquad \qquad \left. G_{11}^{t-1}X_1^{t-1},G_{12}^{t-1}X_1^{t-1},X_2^t,G^t \right),
\end{align}
where $(a)$ and $(b)$ follow from the fact that condition on $\hbox{W}_2$, $G_{12}^{t-1}X_1^{t-1}$, $X_2^t$ and $G^n$, $X_1[t]$ is independent of the channel realization at time $t$; and $(c)$ holds since conditioning reduces entropy.

Therefore, we have 
\begin{align*}
& \sum_{t=1}^n{H(G_{12}[t]X_1[t]|W_2,G_{12}^{t-1}X_1^{t-1},X_2^t,G^t)} \nonumber \\
& \quad \geq \frac{1}{2-p} \sum_{t=1}^n H\left(G_{11}[t]X_1[t],G_{12}[t]X_1[t]|W_2, \right. \nonumber \\
& \qquad \qquad \left. G_{11}^{t-1}X_1^{t-1},G_{12}^{t-1}X_1^{t-1},X_2^t,G^t \right).
\end{align*}
\end{proof}


In the following two sections, we consider the last scenario we are interested in, \emph{i.e.} Instantaneous-CSIT and OFB, and we provide the proof of Theorem~\ref{THM:IC-ICSITFB}. First, we present the achievability strategy, and we demonstrate how OFB can enhance our achievable rate region. We then present the converse proof.


\section{Achievability Proof of Theorem~\ref{THM:IC-ICSITFB}~[Instantaneous-CSIT and OFB]}
\label{sec:AchInstFB}


In this section, we describe our achievability strategy for the case of Instantaneous-CSIT and output feedback. Note that in this scenario, although transmitters have instantaneous knowledge of the channel state information, the output signals are available at the transmitters with unit delay. We first provide a brief overview of our scheme.

\subsection{Overview} 

By symmetry, it suffices to describe the achievability scheme for corner point $$\left( R_1, R_2 \right) = \left( 1 - q^2, pq \right),$$
as depicted in Fig.~\ref{fig:region-InstFB}. Similarly, we can achieve corner point $\left( R_1, R_2 \right) = \left( pq, 1 - q^2 \right)$, and therefore by time sharing, we can achieve the region.

\begin{figure}[ht]
\centering
\includegraphics[height = 4 cm]{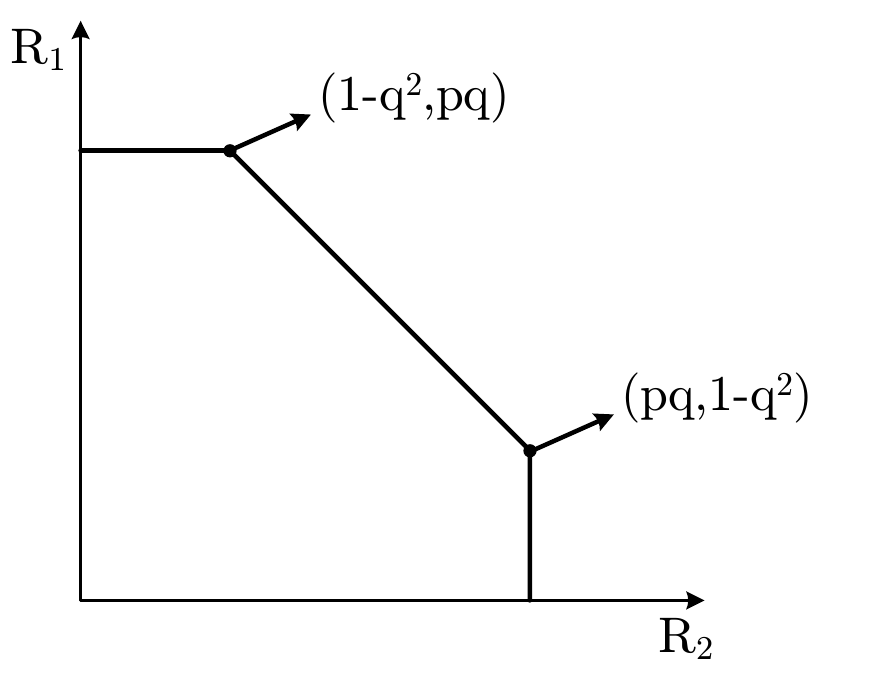}
\caption{Two-user Binary Fading IC: capacity region with Instantaneous-CSIT and output feedback. By symmetry, it suffices to describe the achievability scheme for corner point $\left( R_1, R_2 \right) = \left( 1 - q^2, pq \right)$. \label{fig:region-InstFB}}
\end{figure}

Our achievability strategy is carried on over $b+1$ communication blocks, each block with $n$ time instants. Transmitters communicate fresh data bits in the first $b$ blocks and the final block is to help the receivers decode their corresponding bits. At the end, using our scheme, we achieve a rate tuple arbitrary close to $\frac{b}{b+1} \left( 1 - q^2, pq \right)$ as $n \rightarrow \infty$. Finally letting $b \rightarrow \infty$, we achieve the desired rate tuple.

\subsection{Achievability Strategy} Let $\hbox{W}^j_i$ be the message of transmitter $i$ in block $j$, $i=1,2,$ and $j=1,2,\ldots,b$. Moreover, let $\hbox{W}^j_1 = a^j_1, a^j_2, \ldots, a^j_{m}$, and $\hbox{W}^j_2 = b^j_1, b^j_2, \ldots, b^j_{m_2}$, for $j=1,2,\ldots,b$, where data bits $a^j_i$'s and $b^j_i$'s are picked uniformly and independently from $\{ 0,1 \}$, $i=1,2,\ldots,m$, and 
\begin{align}
m_2 = \frac{q}{1+q} m.
\end{align}

We also set $n = m/(1-q^2) + m^{2/3}$, where $n$ is the length of each block. 

\vspace{1mm}

{\bf Achievability strategy for block $1$}: In the first communication block, at each time instant $t$, if at least one of the outgoing links from ${\sf Tx}_1$ is on, then it sends one of its initial $m$ bits that has not been transmitted before (note that this happens with probability $( 1 - q^2 )$). On the other hand, ${\sf Tx}_2$ communicates a new bit (a bit that has not been transmitted before) if the link to its receiver is on and it does not interfere with receiver one (\emph{i.e.} $G_{22}[t] = 1$ and $G_{21}[t] = 0$). In other words, ${\sf Tx}_2$ communicates a new bit if either one of Cases $2,4,9,$ or $11$ in Table~\ref{table:16cases} occurs (note that this happens with probability $pq$). 

The first block goes on for $n$ time instants. If at the end of the first block, there exists a bit at either of the transmitters that has not yet been transmitted, we consider it as error type-I and halt the transmission.


Assuming that the transmission in not halted, using output feedback links, transmitter two has access to the bits of transmitter one communicated in the first block. In particular, ${\sf Tx}_2$ has access to the bits of ${\sf Tx}_1$ transmitted in Cases $2,11,12,14,$ and $15$ during block $1$. Note that the bits communicated in Cases $11,12,14,$ and $15$ from ${\sf Tx}_1$ have to be provided to ${\sf Rx}_1$. However, the bits communicated in Case $2$ from ${\sf Tx}_1$ are already available at ${\sf Rx}_1$ but needed at ${\sf Rx}_2$, see Fig.~\ref{fig:case2FB}. Transmitter two will provide such bits to ${\sf Rx}_2$ in the following communication block.

\begin{figure}[ht]
\centering
\includegraphics[height = 3 cm]{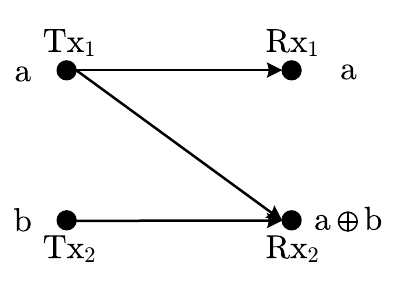}
\caption{The bit communicated in Case $2$ from ${\sf Tx}_1$ is already available at ${\sf Rx}_1$ but it is needed at ${\sf Rx}_2$. Transmitter two learns this bit through the feedback channel and will provide it to ${\sf Rx}_2$ in the following communication block.\label{fig:case2FB}}
\end{figure}

Now, ${\sf Tx}_2$ transfers the bits of ${\sf Tx}_1$ communicated in Cases $2,11,12,14,$ and $15$, during the first communication block to queues $Q^1_{1,{\sf C}_2}$, $Q^1_{1,{\sf C}_{11}}$, $Q^1_{1,{\sf C}_{12}}$, $Q^1_{1,{\sf C}_{14}}$, and $Q^1_{1,{\sf C}_{15}}$ respectively.

\vspace{1mm}

Let random variable $N^1_{1,{\sf C}_{\ell}}$ denote the number of bits in $Q^1_{1,{\sf C}_{\ell}}$, $\ell = 2,11,12,14,15$. Since transition of a bit to this state is distributed as independent Bernoulli RV, upon completion of block $1$, we have 
\begin{align}
& \mathbb{E}[N^1_{1,{\sf C}_{\ell}}] = \frac{\Pr\left( \mathrm{Case~\ell} \right)}{1 - \sum_{i=9,10,13,16}{\Pr\left( \mathrm{Case~i} \right)}} m  \\
&~ = (1-q^2)^{-1} \Pr\left( \mathrm{Case~\ell} \right) m, \qquad \ell = 2,11,12,14,15. \nonumber
\end{align}

If the event $\left[ N^1_{1,{\sf C}_{\ell}} \geq \mathbb{E}[N^1_{1,{\sf C}_{\ell}}] + m^{\frac{2}{3}} \right]$ occurs, we consider it as error type-II and we halt the transmission. At the end of block $1$, we add $0$'s (if necessary) to $Q^1_{1,{\sf C}_{\ell}}$ so that the total number of bits is equal to $\mathbb{E}[N^1_{1,{\sf C}_{\ell}}] + m^{\frac{2}{3}}$. Furthermore, using Chernoff-Hoeffding bound, we can show that the probability of errors of types I and II decreases exponentially with $m$.

\vspace{1mm}

{\bf Achievability strategy for block $j,$ $j=2,3,\ldots,b$}: The transmission strategy for ${\sf Tx}_1$ is the same as block $1$ for the first $b$ blocks (all but the last block). In other words, at time instant $t$, ${\sf Tx}_1$ transmits one of its initial $m$ bits (that has not been transmitted before) if at least one of its outgoing links is on. On the other hand, ${\sf Tx}_2$ communicates $\hbox{W}^2_2$ using similar strategy as the first block, \emph{i.e.} ${\sf Tx}_2$ communicates a new bit if either one of Cases $2,4,9,$ or $11$ occurs. 

Transmitter two transfers the bits communicated in Cases $2,11,12,14,$ and $15$, during communication block $j$ to queues $Q^j_{1,{\sf C}_2}$, $Q^j_{1,{\sf C}_{11}}$, $Q^j_{1,{\sf C}_{12}}$, $Q^j_{1,{\sf C}_{14}}$, and $Q^j_{1,{\sf C}_{15}}$ respectively.

\vspace{1mm}

Moreover, at time instant $t$, 
\begin{itemize}
\item if Case $3$ occurs, ${\sf Tx}_2$ sends one of the bits from $Q^{j-1}_{1,{\sf C}_2}$ and removes it from this queue since it has been delivered successfully to ${\sf Rx}_2$, see Fig.~\ref{fig:case2-3FB}. If Case $3$ occurs and $Q^{j-1}_{1,{\sf C}_2}$ is empty, ${\sf Tx}_2$ remains silent; 
\begin{figure}[ht]
\centering
\includegraphics[height = 3 cm]{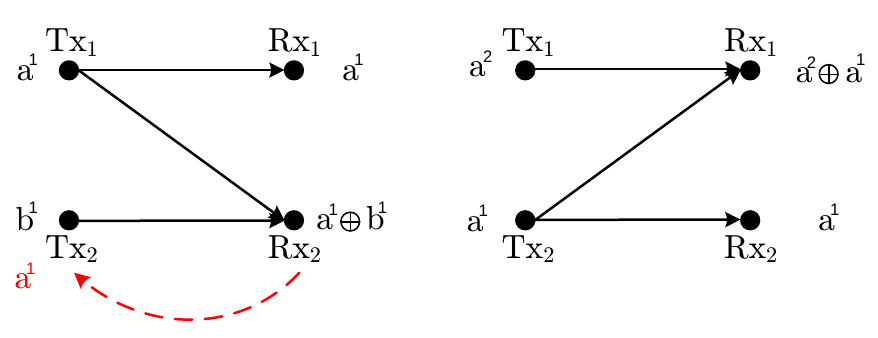}
\caption{In block $j$ when Case $3$ occurs, ${\sf Tx}_2$ retransmits the bit of ${\sf Tx}_1$ communicated in Case $2$ during block $j-1$. Note that this bit does not cause interference at ${\sf Rx}_1$ and it is needed at ${\sf Rx}_2$.\label{fig:case2-3FB}}
\end{figure}

\item if Case $10$ occurs, ${\sf Tx}_2$ sends one of the bits from $Q^{j-1}_{1,{\sf C}_{11}}$ and removes it from this queue, see Fig.~\ref{fig:case10-11FB}. If Case $10$ occurs and $Q^{j-1}_{1,{\sf C}_{11}}$ is empty, ${\sf Tx}_2$ remains silent;
\begin{figure}[ht]
\centering
\includegraphics[height = 3 cm]{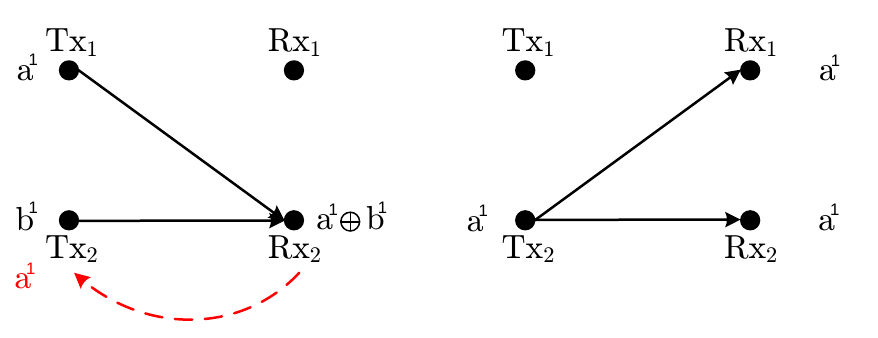}
\caption{In block $j$ when Case $10$ occurs, ${\sf Tx}_2$ retransmits the bit of ${\sf Tx}_1$ communicated in Case $11$ during block $j-1$. Note that this bit is needed at both receivers.\label{fig:case10-11FB}}
\end{figure}

\item if Case $12$ occurs, it sends one of the bits from $Q^{j-1}_{1,{\sf C}_{12}}$ and removes it from this queue. If Case $12$ occurs and $Q^{j-1}_{1,{\sf C}_{12}}$ is empty, ${\sf Tx}_2$ remains silent;

\item if Case $13$ occurs, it sends one of the bits from $Q^{j-1}_{1,{\sf C}_{14}}$ and removes it from this queue. If Case $13$ occurs and $Q^{j-1}_{1,{\sf C}_{14}}$ is empty, ${\sf Tx}_2$ remains silent;

\item if Case $15$ occurs, it sends one of the bits from $Q^{j-1}_{1,{\sf C}_{15}}$ and removes it from this queue. If Case $15$ occurs and $Q^{j-1}_{1,{\sf C}_{15}}$ is empty, ${\sf Tx}_2$ remains silent.
\end{itemize}

If at the end of block $j$, there exists a bit at either of the transmitters that has not yet been transmitted, or any of the queues $Q^{j-1}_{1,{\sf C}_2}$, $Q^{j-1}_{1,{\sf C}_{11}}$, $Q^{j-1}_{1,{\sf C}_{12}}$, $Q^{j-1}_{1,{\sf C}_{14}}$, or $Q^{j-1}_{1,{\sf C}_{15}}$ is not empty, we consider this event as error type-I and halt the transmission.

Assuming that the transmission is not halted, let random variable $N^j_{1,{\sf C}_{\ell}}$ denote the number of bits in $Q^j_{1,{\sf C}_{\ell}}$, $\ell = 2,11,12,14,15$. From basic probability, we have
\begin{align}
& \mathbb{E}[N^j_{1,{\sf C}_{\ell}}] = \frac{\Pr\left( \mathrm{Case~\ell} \right)}{1 - \sum_{i=9,10,13,16}{\Pr\left( \mathrm{Case~i} \right)}} m  \\
&~= (1-q^2)^{-1} \Pr\left( \mathrm{Case~\ell} \right) m, \qquad \ell = 2,11,12,14,15. \nonumber
\end{align}

If the event $\left[ N^j_{1,{\sf C}_{\ell}} \geq \mathbb{E}[N^j_{1,{\sf C}_{\ell}}] + m^{\frac{2}{3}} \right]$ occurs, we consider it as error type-II and we halt the transmission. At the end of block $1$, we add $0$'s (if necessary) to $Q^j_{1,{\sf C}_{\ell}}$ so that the total number of bits is equal to $\mathbb{E}[N^j_{1,{\sf C}_{\ell}}] + m^{\frac{2}{3}}$. Using Chernoff-Hoeffding bound, we can show that the probability of errors of types I and II and decreases exponentially with $m$.

\vspace{1mm}

{\bf Achievability strategy for block $b+1$}: Finally in block $b+1$, no new data bit is transmitted (\emph{i.e.} $\hbox{W}^{b+1}_1, \hbox{W}^{b+1}_2 = 0$), and ${\sf Tx}_2$ only communicates the bits of ${\sf Tx}_1$ communicated in the previous block in Cases $2,11,12,14,$ and $15$ as described above. If at the end of block $b+1$, any of the queues $Q^{b}_{1,{\sf C}_2}$, $Q^{b}_{1,{\sf C}_{11}}$, $Q^{b}_{1,{\sf C}_{12}}$, $Q^{b}_{1,{\sf C}_{14}}$, or $Q^{b}_{1,{\sf C}_{15}}$ is not empty, we consider this event as error type-I and halt the transmission.

\vspace{1mm}

The probability that the transmission strategy halts at the end of each block can be bounded by the summation of error probabilities of types I and II. Using Chernoff-Hoeffding bound, we can show that the probability that the transmission strategy halts at any point approaches zero as $m \rightarrow \infty$.

\subsection{Decoding} At the end of block $j+1$, ${\sf Rx}_1$ has acces to $\hbox{W}^j_1$ with no interference, $j=1,2,\ldots,b$. At the end of block $b+1$, ${\sf Rx}_2$ uses the bits communicated in Cases $3$ and $10$ from ${\sf Tx}_2$ to cancel out the interference it has received from ${\sf Tx}_1$ during the previous block in Cases $2$ and $11$. Therefore, at the end of block $b+1$, ${\sf Rx}_1$ has access to $\hbox{W}^b_2$ with no interference. Then, ${\sf Rx}_2$ follows the same strategy for blocks $b$ and $b-1$. Therefore, using similar idea, ${\sf Rx}_2$ uses backward decoding to cancel out interference in the previous blocks to decode all messages.

Now, since each block has $n = m/(1-q^2) + m^{2/3}$ time instants and the probability that the transmission strategy halts at any point approaches zero for $m \rightarrow \infty$, we achieve a rate tuple 
\begin{align}
\frac{b}{b+1} \left( 1 - q^2, pq \right),
\end{align}
as $m \rightarrow \infty$. Finally letting $b \rightarrow \infty$, we achieve the desired rate tuple.


\section{Converse Proof of Theorem~\ref{THM:IC-ICSITFB}~[Instantaneous-CSIT and OFB]}
\label{sec:ConvInstFB}


To derive the outer-bound on individual rates, we have
\begin{align}
n R_1 & = H(W_1) \overset{(a)}= H(W_1|G^n) \nonumber \\
& \overset{(\mathrm{Fano})}\leq I(W_1;Y_1^n|G^n) + n \epsilon_n \nonumber \\
& = H(Y_1^n|G^n) - H(Y_1^n|W_1,G^n) + n \epsilon_n \nonumber \\
& \leq H(Y_1^n|G^n)+ n \epsilon_n \nonumber \\
& \leq  (1-q^2) n + n \epsilon_n,
\end{align}
where $\epsilon_n \rightarrow 0$ as $n \rightarrow \infty$; $(a)$ holds since message $\hbox{W}_1$ is independent of $G^n$. Similarly, we have
\begin{align}
n R_2 \leq (1-q^2) n + n \epsilon_n,
\end{align}
dividing both sides by $n$ and let $n \rightarrow \infty$, we have 
\begin{equation}
\label{}
\left\{ \begin{array}{ll}
\vspace{1mm} R_1 \leq 1 - q^2, & \\
R_2 \leq 1 - q^2. & 
\end{array} \right.
\end{equation}

The outer-bound on $R_1+R_2$, \emph{i.e.}
\begin{align}
R_1 + R_2 \leq 1 - q^2 + p q,
\end{align}
can be obtained as follows.
\begin{align}
n & (R_1 + R_2 - 2 \epsilon_n ) \nonumber \\
& \overset{(a)}\leq H(W_1|W_2,G^n) + H(W_2|G^n) \nonumber \\
& \overset{\mathrm{Fano}}\leq I(W_1;Y_1^n|W_2,G^n) + I(W_2;Y_2^n|G^n) \nonumber \\
& = H(Y_1^n|W_2,G^n) - \underbrace{H(Y_1^n|W_1,W_2,G^n)}_{=~0} + I(W_2;Y_2^n|G^n) \nonumber \\
& = H(Y_1^n|W_2,G^n) + H(Y_2^n|G^n) - H(Y_2^n|W_2,G^n) \nonumber \\
& = H(Y_1^n|W_2,G^n) + H(Y_2^n|G^n) - \left[ H(Y_1^n,Y_2^n|W_2,G^n) \right. \nonumber \\
& \qquad \left. - H(Y_1^n|Y_2^n,W_2,G^n) \right] \nonumber \\
& = H(Y_1^n|Y_2^n,W_2,G^n) + H(Y_2^n|G^n) \nonumber \\
& \overset{(b)}= H(Y_2^n|G^n) \nonumber \\
&~+ \sum_{t=1}^n{H(Y_1[t]|W_2,Y_2^n,Y_1^{t-1},X_2^t,G_{12}^t X_1^t,G^n)} \nonumber
\end{align}
\begin{align}
\label{eq:sumrateInstFB}
& \overset{(c)}\leq H(Y_2^n|G^n) + H(Y_1^n|G_{12}^n X_1^n,G_{21}^n X_2^n,G^n) \nonumber \\
& \overset{(d)}\leq \sum_{t=1}^n{H(Y_2[t]|G^n)} \nonumber \\
&~+ \sum_{t=1}^n{H(Y_1[t]|G_{12}[t] X_1[t],G_{21}[t] X_2[t],G^n)} \nonumber \\
& \overset{(e)}\leq \left( 1 - q^2 \right) n + p q n,
\end{align}
where $\epsilon_n \rightarrow 0$ as $n \rightarrow \infty$; and $(a)$ follows from the fact that the messages and $G^n$ are mutually independent; $(b)$ holds since $X_2^t$ is a function of $W_2,Y_1^{t-1},Y_2^{t-1},$ and $G^t$; $(c)$ and $(d)$ follow from the fact that conditioning reduces entropy; and $(e)$ holds since 
\begin{align}
& H(Y_2[t]|G^n) \leq 1 - q^2, \nonumber\\
& H(Y_1[t]|G_{12}[t] X_1[t],G_{21}[t] X_2[t],G^n) \leq pq.
\end{align}

Dividing both sides by $n$ and let $n \rightarrow \infty$, we get 
\begin{align}
R_1 + R_2 \leq 1 - q^2 + p q.
\end{align}

We also note this outer-bound on $R_1+R_2$ can be also applied to the case of Instantaneous-CSIT and no output feedback (\emph{i.e.} Theorem~\ref{THM:IC-InstCSIT}).



\section{Extension to the Non-Homogeneous Setting}
\label{sec:extension}


In this section, we discuss the extension of our results to the non-homogeneous case. More precisely, we consider the two-user Binary Fading Interference Channel of Section~\ref{sec:problem} where 
\begin{align}
G_{ii}[t] \overset{d}\sim \mathcal{B}(p_d), \qquad G_{i\bar{i}}[t] \overset{d}\sim \mathcal{B}(p_c),
\end{align}
for $0 \leq p_d,p_c \leq 1$, $\bar{i} = 3 - i$, and $i = 1,2$. We define $q_d = 1 - p_d$ and $q_c = 1 - p_c$. We study the non-homogeneous BFIC in two settings: $(1)$ Delayed-CSIT and output feedback; and $(2)$ Delayed-CSIT (and no output feedback). For the case of Delayed-CSIT and output feedback, we fully characterize the capacity region as follows.

 
\begin{theorem}
\label{THM:IC-DelayedCSIT-OFB-Symmetric}
{\bf [Capacity Region with Delayed-CSIT, OFB, and Non-Homogeneous Channel Gains]} The capacity region of the two-user Binary Fading IC with Delayed-CSIT and output feedback, $\mathcal{C}^{\mathrm{DCSIT,OFB}}\left( p_d,p_c \right)$ is given by
\begin{align}
\label{eq:capacity-FB-Sym}
& \mathcal{C}^{\mathrm{DCSIT,OFB}}(p_d,p_c) = \left\{ R_1, R_2 \in \mathbb{R}^+~s.t. \right. \nonumber \\
& \left.~p_c R_i + \left( 1 - q_dq_c \right) R_{\bar{i}} \leq \left( 1 - q_dq_c \right)^2,~i=1,2 \right\}.
\end{align}
\end{theorem}


\begin{proof}
We first prove the converse. The converse proof follows similar steps as the case of the homogeneous setting described in Scetion~\ref{sec:conversedelayedhalfFB} for Theorem~\ref{THM:IC-DCSITFB}. Set\footnote{For $p_c = 0$ the result is trivial, so we assume that $\beta$ is well defined.}
\begin{align}
\beta = \frac{\left( 1 - q_dq_c \right)}{p_c}.
\end{align}

We have
\begin{align}
\label{eq:converseDelayedOFBnonhomo}
n & \left( R_1 + \beta R_2 \right) = H(W_1) + \beta H(W_2) \nonumber \\
& \overset{(a)}= H(W_1|W_2,G^n) + \beta H(W_2|G^n) \nonumber \\
& \overset{\mathrm{Fano}}\leq I(W_1;Y_1^n|W_2,G^n) + \beta I(W_2;Y_2^n|G^n) + n \epsilon_n \nonumber \\
& \leq I(W_1;Y_1^n,Y_2^n|W_2,G^n) + \beta I(W_2;Y_2^n|G^n) + n \epsilon_n \nonumber \\
& = H(Y_1^n,Y_2^n|W_2,G^n) - \underbrace{H(Y_1^n,Y_2^n|W_1,W_2,G^n)}_{=~0} \nonumber \\
&~+ \beta H(Y_2^n|G^n) - \beta H(Y_2^n|W_2,G^n) + n \epsilon_n \nonumber \\
& = \beta H(Y_2^n|G^n) + H(Y_1^n,Y_2^n|W_2,G^n) \nonumber \\
&~- \beta H(Y_2^n|W_2,G^n) + n \epsilon_n \nonumber \\
& = \beta H(Y_2^n|G^n) + \sum_{t=1}^n{H(Y_1[t],Y_2[t]|W_2,Y_1^{t-1},Y_2^{t-1},G^n)} \nonumber \\
&~- \beta \sum_{t=1}^n{H(Y_2[t]|W_2,Y_2^{t-1},G^n)} + n \epsilon_n \nonumber \\
& \overset{(b)}\leq \beta H(Y_2^n|G^n) \nonumber \\
&~+ \sum_{t=1}^n{H(Y_1[t],Y_2[t]|W_2,Y_1^{t-1},Y_2^{t-1},X_2^t,G^n)} \nonumber \\
&~- \beta \sum_{t=1}^n{H(Y_2[t]|W_2,Y_2^{t-1},X_2^t,G^n)} + n \epsilon_n \nonumber \\
& = \beta H(Y_2^n|G^n) + \sum_{t=1}^n H\left(G_{11}[t]X_1[t],G_{12}[t]X_1[t]|W_2, \right. \nonumber\\
& \qquad \qquad \left. G_{11}^{t-1}X_1^{t-1},G_{12}^{t-1}X_1^{t-1},X_2^t,G^n \right) \nonumber \\
& \quad - \beta \sum_{t=1}^n{H(G_{12}[t]X_1[t]|W_2,G_{12}^{t-1}X_1^{t-1},X_2^t,G^n)} + n \epsilon_n \nonumber \\
& \overset{(c)}= \beta H(Y_2^n|G^n) + \sum_{t=1}^n H\left(G_{11}[t]X_1[t],G_{12}[t]X_1[t]|W_2, \right. \nonumber \\
& \qquad \qquad \left. G_{11}^{t-1}X_1^{t-1},G_{12}^{t-1}X_1^{t-1},X_2^t,G^t \right) \nonumber \\
& \quad - \beta \sum_{t=1}^n{H(G_{12}[t]X_1[t]|W_2,G_{12}^{t-1}X_1^{t-1},X_2^t,G^t)} + n \epsilon_n \nonumber \\
& \overset{(d)}\leq \beta H(Y_2^n|G^n) + n \epsilon_n \nonumber \\
& \leq \frac{\left( 1 - q_dq_c \right)^2}{p_c} n + n \epsilon_n,
\end{align}
where $(a)$ holds since the channel gains and the messages are mutually independent; $(b)$ follows from the fact that $X_2^t$ is a deterministic function of $\left( W_2,Y_2^{t-1} \right)$\footnote{We have also added $Y_1^{t-1}$ in the condition for the scenario in which output feedback links are available from each receiver to both transmitters.} and the fact that conditioning reduces entropy; $(c)$ follows from the fact that condition on $\hbox{W}_2$, $X_1^{t-1}$, $X_2^t$, $X_1[t]$ is independent of the channel realization at future time instants, hence, we can replace $G^n$ by $G^t$; and $(d)$ follows from Lemma~\ref{lemma:portionnonhomo} below. Dividing both sides by $n$ and let $n \rightarrow \infty$, we get 
\begin{align}
p_c R_1 + \left( 1 - q_dq_c \right) R_2 \leq \left( 1 - q_dq_c \right)^2,
\end{align}
and the derivation of the other bound would be similar.

\begin{lemma}
\label{lemma:portionnonhomo}
{\bf [Non-Homogeneous Entropy Leakage with Output Feedback]} For the broadcast channel described in Fig.~\ref{fig:portionHalf} with parameters $p_d$ and $p_c$, and with Delayed-CSIT and output feedback, for \emph{any} input distribution, we have
\begin{align}
& \left( 1 - q_dq_c \right) \sum_{t=1}^n{H(G_{12}[t]X_1[t]|W_2,G_{12}^{t-1}X_1^{t-1},X_2^t,G^t)} \nonumber \\ 
& \quad \geq p_c \sum_{t=1}^n H\left(G_{11}[t]X_1[t],G_{12}[t]X_1[t]|W_2, \right. \nonumber \\
& \qquad \qquad \left. G_{11}^{t-1}X_1^{t-1},G_{12}^{t-1}X_1^{t-1},X_2^t,G^t \right).
\end{align}
\end{lemma}

The proof of Lemma~\ref{lemma:portionnonhomo} follows the same steps as the proof of Lemma~\ref{lemma:OFBHalf}. We note that this lemma, is the generalization of the Entropy Leakage Lemma to the case where $G_1[t] \overset{d}\sim \mathcal{B}(p_d)$ and $G_2[t] \overset{d}\sim \mathcal{B}(p_c)$ and where the output feedback is present.

We now describe the achievability proof. The achievability proof is also similar to that of the homogeneous setting described in Section~\ref{sec:AchDelayedFB}. Hence, we provide an outline of the achievability strategy here. The achievability strategy of corner points $\left( 1- q_dq_c, 0 \right)$ and $\left( 0, 1- q_dq_c \right)$, is based on utilizing the additional communication paths created by the means of the output feedback links, \emph{e.g.}, $${\sf Tx}_1 \rightarrow {\sf Rx}_2 \rightarrow {\sf Tx}_2 \rightarrow {\sf Rx}_1.$$

In the rest of the proof, we provide the outline for the achievability of corner point
\begin{align}
\label{eq:cornerFBnonhomo}
R_1 = R_2 = \frac{(1-q_dq_c)}{1+(1-q_dq_c)^{-1}p_c} \raisebox{2pt}{.}
\end{align}

The strategy is carried on over two phases similar to Phase~$1$ and Phase~$2$ of Section~\ref{sec:AchDelayedFB}. We assume that at the beginning of the communication block, there are $m$ bits in $Q_{i \rightarrow i}$, $i=1,2$. Phase~$1$ is the uncategorized transmission, and it goes on for 
\begin{align}
\left( 1 - q_dq_c \right)^{-1} m + m^{\frac{2}{3}}
\end{align}
time instants and if at the end of this phase, either of the queues $Q_{1 \rightarrow 1}$ or $Q_{2 \rightarrow 2}$ is not empty, we declare an error and halt the transmission. Upon completion of Phase~$1$, using the ideas described in Section~\ref{sec:opportunities} for output feedback, we further create bits of common interest. More precisely, we use the following ideas: we update the status of bits in $Q_{i,{\sf C}_1}$ to bits of common interest; as described in Example~6 of Section~\ref{sec:opportunities}, using output feedback, we combine bits in $Q_{1 \rightarrow 1|2}$ and $Q_{2 \rightarrow 2|1}$ to create bits of common interest; as described in Example~7 of Section~\ref{sec:opportunities}, using output feedback, we combine bits in $Q_{1 \rightarrow 2|1}$ and $Q_{2 \rightarrow 1|2}$ to create bits of common interest.  

In the second phase, we deliver the bits in $Q_{1 \rightarrow \{ 1,2 \}}$ and $Q_{2 \rightarrow \{ 1,2 \}}$ using the transmission strategy for the two-multicast problem. For 
\begin{align}
0 \leq p_c \leq \frac{p_d}{1+p_d},
\end{align}
the cross links become the bottleneck for the two-multicast network depicted in Fig.~\ref{fig:two-multicast}, and as a result, using the two-multicast problem as discussed in Lemma~\ref{lemma:multicast} is sub-optimal. However, using the output feedback link, ${\sf Tx}_i$ learns the interfering bit of ${\sf Tx}_{\bar{i}}$, $i=1,2$; and considering this side information available at the transmitters, we can show that a sum-rate of $\left( 1 - q_d q_c \right)$ for the two-multicast problem is in fact achievable. At the end of Phase~2, all bits are delivered. It takes  
\begin{align}
\left( 1 - q_dq_c \right)^{-2} p_c m + O\left( m^{\frac{2}{3}} \right)
\end{align}
time instants to complete Phase~2. Therefore, we achieve a symmetric sum-rate of   
\begin{align}
R_1 = R_2 = \frac{(1-q_dq_c)}{1+(1-q_dq_c)^{-1}p_c},
\end{align}
as $m \rightarrow \infty$.
\end{proof}

For the case of Delayed-CSIT (and no output feedback), we partially solve the problem as described below.

\begin{theorem}
\label{THM:IC-DelayedCSIT-Symmetric}
{\bf [Capacity Region with Delayed-CSIT and Non-Homogeneous Channel Gains]} The capacity region of the two-user Binary Fading IC with Delayed-CSIT (and no output feedback), $\mathcal{C}^{\mathrm{DCSIT}}\left( p_d,p_c \right)$ for 
\begin{align}
\frac{p_d}{1+p_d} \leq p_c \leq 1,
\end{align}
is the set of all rate tuples $\left( R_1, R_2 \right)$ satisfying
\begin{equation}
\label{eq:DelayedNSIregionSym}
\mathcal{C}^{\mathrm{DCSIT}}\left( p_d,p_c \right) =
\left\{ \begin{array}{ll}
\vspace{1mm} 0 \leq R_i \leq p_d, & \\
p_c R_i + \left( 1 - q_dq_c \right) R_{\bar{i}} \leq \left( 1 - q_dq_c \right)^2, & 
\end{array} \right.
\end{equation}
for $i = 1,2$.
\end{theorem}

\begin{proof}
The proof of converse follows from the previous theorem since the outer-bound of Theorem~\ref{THM:IC-DelayedCSIT-OFB-Symmetric} also serves as an outer-bound for Theorem~\ref{THM:IC-DelayedCSIT-Symmetric}. Here, we discuss the achievability strategy. The achievability proof is similar to that of the homogeneous setting as described in Section~\ref{sec:achallhalf} for Theorem~\ref{THM:IC-DelayedCSIT}. The corner points are as follows.
\begin{align}
& \left( R_1, R_2 \right) = \left( \min\left\{ p_d, \frac{(1-q_dq_c)}{1+(1-q_dq_c)^{-1}p_c} \right\}, \right. \nonumber \\
& \qquad \qquad \left. \min\left\{ p_d, \frac{(1-q_dq_c)}{1+(1-q_dq_c)^{-1}p_c} \right\} \right), \nonumber\\
& \left( R_1, R_2 \right) = \left( p_d, \min\left\{ p_d, \left( 1-q_dq_c \right) q_d \right\} \right), \nonumber\\
& \left( R_1, R_2 \right) = \left( \min\left\{ p_d, \left( 1-q_dq_c \right) q_d \right\}, p_d \right).
\end{align}

Here, we provide the outline for the achievability of the first corner point, \emph{i.e.}
\begin{align}
R_1 = R_2 = \min\left\{ p_d, \frac{(1-q_dq_c)}{1+(1-q_dq_c)^{-1}p_c} \right\} \raisebox{2pt}{.}
\end{align}

We assume that at the beginning of the communication block, there are $m$ bits in $Q_{i \rightarrow i}$, $i=1,2$. Phase~$1$ is the uncategorized transmission.
Upon completion of Phase~$1$, using the ideas described in Section~\ref{sec:oppideas1}, we upgrade the status of the bits to bits of common interest. We use the following coding opportunities.
\begin{itemize}

\item {\bf Type I} Combining bits in $Q_{i \rightarrow \bar{i}|i}$ and $Q_{i \rightarrow i|\bar{i}}$ to create bits of common interest, $i=1,2$;

\item {\bf Type II} Combining the bits in $Q_{i,{\sf C}_1}$ and $Q_{i \rightarrow i|\bar{i}}$ to create bits of common interest, $i=1,2$;

\item {\bf Type III} Combining the bits in $Q_{i,{\sf C}_1}$ and $Q_{i \rightarrow \bar{i}|i}$ to create bits of common interest, $i=1,2$.

\end{itemize}
Then, in the final phase, we deliver the bits in $Q_{1 \rightarrow \{ 1,2 \}}$ and $Q_{2 \rightarrow \{ 1,2 \}}$ using the transmission strategy for the two-multicast problem.
\end{proof}

\begin{remark}
We note that in Theorem~\ref{THM:IC-DelayedCSIT-Symmetric}, we partially characterized the capacity region. In fact, for 
\begin{align}
0 \leq p_c \leq \frac{p_d}{1+p_d},
\end{align}
our achievability region does not match the outer-bounds. Closing the gap in this regime could be an interesting future direction. The reason one might think the achievability could be improved is that in this regime the cross links become the bottleneck for the two-multicast network depicted in Fig.~\ref{fig:two-multicast}, and as a result, using the two-multicast problem might be sub-optimal. On the other hand, as we demonstrated in~\cite{AlirezaNoCSIT,vahid2014binary}, even under No-CSIT assumption, this regime requires a different outer-bound compared to other regimes. 
\end{remark}




\section{Conclusion and Future Directions}
\label{sec:conclusion}


We studied the effect of delayed knowledge of the channel state information at the transmitters, on the capacity region of the two-user binary fading interference channels. We introduced various coding opportunities, created by Delayed-CSIT, and presented an achievability strategy that systematically exploits the coding opportunities. We derived an achievable rate region that matches the outer-bounds for this problem, hence, characterizing the capacity region. We have also derived the capacity region of this problem with Delayed-CSIT and output feedback.

A future direction would be to extend our results to the case of two-user Gaussian fading interference channel with Delayed-CSIT. As discussed in the introduction, one can view our binary fading model as a fading interpretation of the linear deterministic model where the non-negative integer associated to each link is at most $1$. Therefore, one approach is to extend the current results to the case of fading linear deterministic interference channel and then, further extend that result to the case of Gaussian fading interference channel, in order to obtain approximate capacity characterization. This approach has been taken for the No-CSIT assumption in~\cite{tse2012fading}, in the context of fading broadcast channels, and in~\cite{Guo} in the context of one-sided fading interference channels. In fact, one can view our Binary Fading model as the model introduced in~\cite{tse2012fading,Guo} with only one layer.

Another future direction is to consider the $k$-user setting of the problem. In~\cite{Jafar_ergodic}, authors have shown that for the $k$-user fading interference channel with instantaneous knowledge of the channel state information, sum degrees of freedom (DoF) of $k/2$ is achievable. However, in the absence of the CSIT, the achievable sum DoF collapses to $1$. As a result a large degradation in network capacity, due to lack of the CSIT, is observed. It has been recently shown that, with Delayed-CSIT, it is possible to achieve more than one sum DoF \cite{Jafar_Retrospective,Abdoli2011IC-X-Arxiv}, however, the achievable sum DoFs are less than $1.5$ for any number of users. This together with lack of nontrivial DoF upper bounds leaves the problem of sum DoF characterization of interference channels with Delayed-CSIT still open and challenging, to the extent that it is even unknown whether the sum DoF of such networks \emph{scales} with the number of users or not. A promising direction may be to study this problem in the context of our simpler binary fading model, to understand whether the sum capacity of such network with Delayed-CSIT scales with the number of users or it will saturate.

Finally, motivated by recent results that demonstrate that, with Instantaneous-CSIT, multi-hopping can significantly increase the capacity of interference networks (\emph{e.g.},~\cite{gou2011aligned,shomorony2011two-unicast} for two-unicast networks and~\cite{shomorony2012KKK} for multi-unicast networks), an interesting future direction would be explore the effect of Delayed-CSIT on the capacity of muti-hop binary interference networks.

\appendices


\section{Achievability Proof of Theorem~\ref{THM:IC-InstCSIT}~[Instantaneous-CSIT]}
\label{sec:AchInst}


In this appendix, we provide the achievability proof of Theorem~\ref{THM:IC-InstCSIT}. Below, we have stated the capacity region of the two-user BFIC with Instantaneous-CSIT (and no OFB).
\begin{equation}
\label{}
\mathcal{C}^{\mathrm{ICSIT}} =
\left\{ \begin{array}{ll}
\vspace{1mm} 0 \leq R_i \leq p, & i = 1,2, \\
R_1 + R_2 \leq 1 - q^2 + p q.  &  \\
\end{array} \right.
\end{equation}

\begin{figure}
\centering
\subfigure[]{\includegraphics[height = 4cm]{FiguresPDF/region-plessthanhalf.pdf}}
\hspace{1in}
\subfigure[]{\includegraphics[height = 4cm]{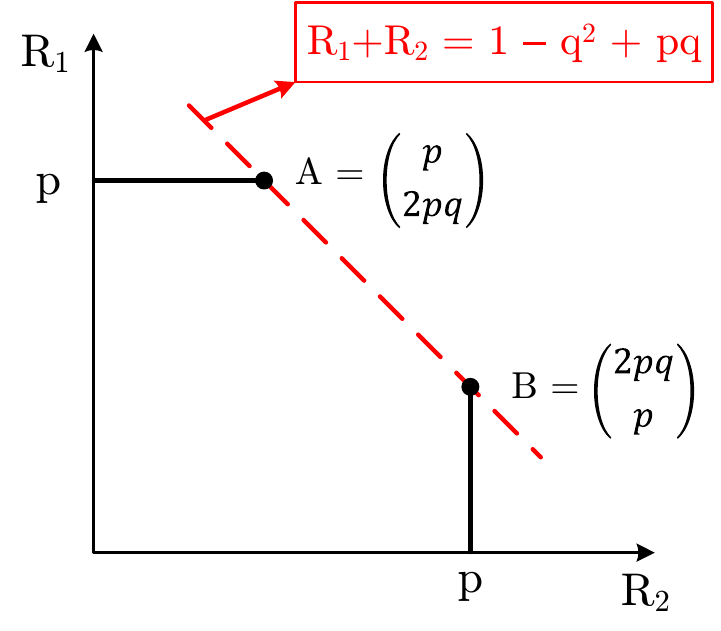}}
\caption{\it Capacity region of the two-user BFIC with Instantaneous-CSIT and for (a) $0 \leq p \leq 0.5$, and (b) $0.5 < p \leq 1$.\label{fig:plessmore}}
\end{figure}

\begin{remark}
For $0 \leq p \leq 0.5$, the capacity region is given by
\begin{equation}
\label{}
\mathcal{C}^{\mathrm{ICSIT}} = \left\{ R_1, R_2 \in \mathbb{R}^+~s.t.~R_i \leq p,~i=1,2 \right\}.
\end{equation}
while for $0.5 < p \leq 1$, the outer-bound on $R_1+R_2$ is also active, see Fig.~\ref{fig:plessmore}.
\end{remark}

With Instantaneous-CSIT, each transmitter knows what channel realization occurs at the time of transmission. Transmitters can take advantage of such knowledge and by pairing different realizations, the optimal rate region as given in Theorem~\ref{THM:IC-InstCSIT} can be achieved. We will first describe the achievability strategy for $0 \leq p \leq 0.5$, since it is easier to follow. We then complete the proof by describing the achievability strategy for $0.5 <  p \leq 1$.

\subsection{Achievabiliy Strategy for $0 \leq p \leq 0.5$}  

Note that the result for $p = 0$ is trivial, so we assume $0 < p \leq 0.5$. Below, we describe the possible pairing opportunities that are useful in this regime and then, we describe the achievability scheme. The possible pairing opportunities are as follows.
\begin{itemize}
\item {Type A} [Cases $1$ and $15$]: In Case $15$, only the corss links are equal to $1$, therefore, by pairing bits in Case $1$ with bits in Case $15$, we can cancel out interference in Case $1$, see Fig.~\ref{fig:TypeA}. In other words by pairing the two cases, we can communicate $2$ bits interference free.

\begin{figure}[h]
\centering
\includegraphics[height = 3cm]{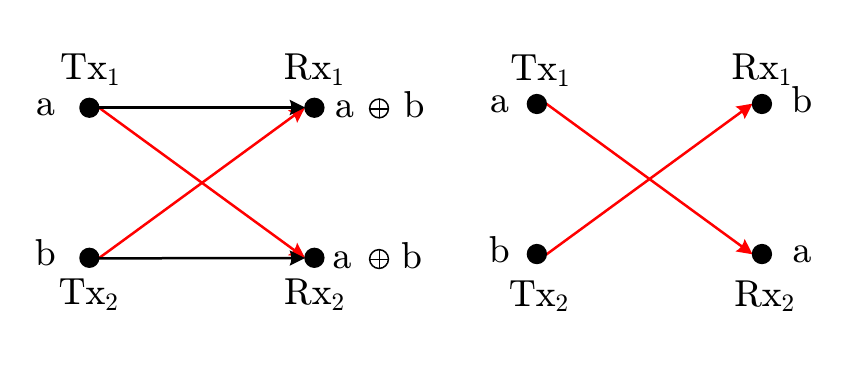}
\caption{Pairing opportunity Type A: By pairing Cases $1$ and $15$, we can communicate two bits interference-free. For instance, receiver one has access to bits $a \oplus b$ and $b$ and as a result, it can decode its desired bit.}\label{fig:TypeA}
\end{figure}

\item {Type B} [Cases $2$ and $14$]: We can pair up Cases $2$ and $14$ to cancel out interference in Case $2$ as depicted in Fig.~\ref{fig:TypeC}.

\begin{figure}[h]
\centering
\includegraphics[height = 3cm]{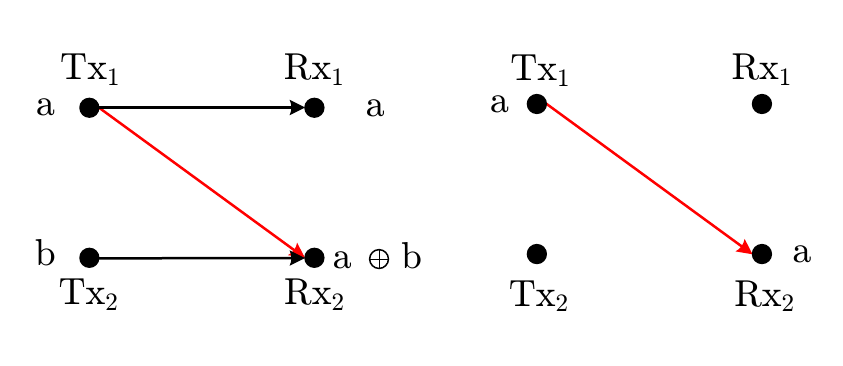}
\caption{Pairing opportunity Type B: By pairing Cases $2$ and $14$, we can communicate two bits interference-free. For instance, receiver two has access to bits $a \oplus b$ and $a$ and as a result, it can decode its desired bit.}\label{fig:TypeC}
\end{figure}

\item {Type C} [Cases $3$ and $13$]: Similar to {Type B} with swapping user IDs.
\end{itemize}

We are now ready to provide the achievability scheme for the Instantaneous-CSIT model and for $0 \leq p \leq 0.5$. We first provide an overview of our scheme.

\subsubsection{Overview} Our achievability strategy is carried on over $b+1$ communication blocks, each block with $n$ time instants. We describe the achievability strategy for rate tuple 
\begin{align}
\left( R_1, R_2 \right) = \left( p, p \right).
\end{align}

Transmitters communicate fresh data bits in the first $b$ blocks and the final block is to help the receivers decode their corresponding bits. At the end, using our scheme, we achieve rate tuple $\frac{b}{b+1} \left( p, p \right)$ as $n \rightarrow \infty$. Finally, letting $b \rightarrow \infty$, we achieve the desired rate tuple. In our scheme the messages transmitted in block $j$, $j=1,2,\ldots,b$, will be decoded at the end of block $j+1$.

\vspace{1mm}

\subsubsection{Achievability strategy} Let $\hbox{W}^j_i$ be the message of transmitter $i$ in block $j$, $i=1,2$, $j=1,2,\ldots,b$. Moreover, let $\hbox{W}^j_1 = a^j_1, a^j_2, \ldots, a^j_{m}$, and $\hbox{W}^j_2 = b^j_1, b^j_2, \ldots, b^j_{m}$, where $a^j_i$'s and $b^j_i$'s are picked uniformly and independently from $\{ 0,1 \}$, $i=1,2,\ldots,m$, $j=1,2,\ldots,b$, and some positive integer $m$. We set
\begin{align}
n = m/p + \left( 2/p^4 \right) m^{2/3}.
\end{align}


\vspace{1mm}

{\bf Achievability strategy for block $1$}: In the first communication block, at each time instant $t$, ${\sf Tx}_i$ sends a new data bit (from its initial $m$ bits) if $G_{ii}[t] = 1$, $i=1,2$. In other words, ${\sf Tx}_1$ sends a new data bit either of the following channel realizations occurs (see Table~\ref{table:16cases}): Cases $1,2,3,4,5,6,7,$ and $8$; while ${\sf Tx}_2$ sends a new data bit if either of the following channel realizations occurs: Cases $1,2,3,4,9,10,11,$ and $12$. 

If not specified, the transmitters remain silent. ${\sf Tx}_1$ transfers its transmitted bits in Cases $1$ and $2$ to queues $Q^1_{1,{\sf C1}}$ and $Q^1_{1,{\sf C2}}$ respectively; and ${\sf Tx}_2$ transfers its transmitted bits in Cases $1$ and $3$ to queues $Q^1_{2,{\sf C1}}$ and $Q^1_{2,{\sf C3}}$ respectively.

If at the end of block $1$, there exists a bit at either of the transmitters that has not yet been transmitted, we consider it as error type-I and halt the transmission. 

\begin{remark}
Note that the transmitted bits in Cases $4,5,6,7,8,9,10,11,$ and $12$ are available at their corresponding receivers without any interference. In other words, they are communicated successfully and no retransmission is required.
\end{remark}

Assuming that the transmission is not halted, let random variable $N^1_{i,{\sf C}_{\ell}}$ denote the number of bits in $Q^1_{i,{\sf C}_{\ell}}$, $\left( i, \ell \right) = (1,1), (1,2), (2,1), (2,3)$. Since transition of a bit to this queue is distributed as independent Bernoulli RV, upon completion of block $1$, we have 
\begin{align}
&\mathbb{E}[N^1_{i,{\sf C}_{\ell}}] = \frac{\Pr\left( \mathrm{Case~\ell} \right)}{1 - \sum_{i=9,10,\ldots,16}{\Pr\left( \mathrm{Case~i} \right)}} m  \nonumber \\
&~= \frac{1}{p} \Pr\left( \mathrm{Case~\ell} \right) m.
\end{align}

If the event $\left[ N^1_{i,{\sf C}_{\ell}} \geq \mathbb{E}[N^1_{i,{\sf C}_{\ell}}] + m^{\frac{2}{3}} \right]$ occurs, we consider it as error type-II and we halt the transmission. At the end of block $1$, we add $0$'s (if necessary) to $Q^1_{i,{\sf C}_{\ell}}$ so that the total number of bits is equal to $\mathbb{E}[N^1_{i,{\sf C}_{\ell}}] + m^{\frac{2}{3}}$. Furthermore, using Chernoff-Hoeffding bound, we can show that the probability of errors of types I and II decreases exponentially with $m$.

\vspace{1mm}

{\bf Achievability strategy for block $j,$ $j=2,3,\ldots,b$}: In communication block $j$, $j=2,3,\ldots,b$, at each time instant $t$, ${\sf Tx}_i$ sends a new data bit (from its initial $m$ bits) if $G_{ii}[t] = 1$, $i=1,2$. Transmitter one transfers its transmitted bit in Cases $1$ and $2$ to queues $Q^j_{1,{\sf C1}}$ and $Q^j_{1,{\sf C2}}$ respectively; and ${\sf Tx}_2$ transfers its transmitted bit in Cases $1$ and $3$ to queues $Q^j_{2,{\sf C1}}$ and $Q^j_{2,{\sf C3}}$ respectively. Note that so far the transmission scheme is similar to the first communication block. 

\begin{figure*}[t]
\centering
\includegraphics[height = 3.5cm]{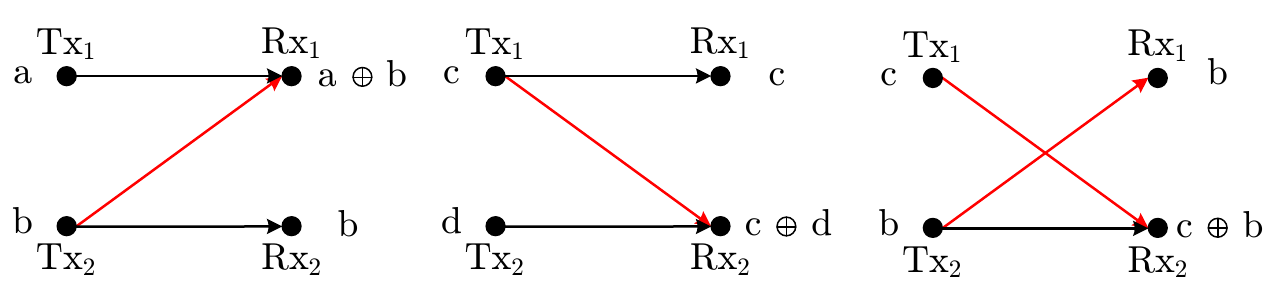}
\caption{Pairing opportunity Type D: Cases $2, 3,$ and $12$. ${\sf Tx}_1$ uses $c$ to recover $b$ and then it decodes $a$, similar argument holds for ${\sf Tx}_2$. All three cases have capacity $1$, and by pairing them, we can communicate $4$ bits.}\label{fig:TypeG}
\end{figure*}

Now if at a given time instant Case $15$ occurs, ${\sf Tx}_i$ sends a bit from queue $Q^{j-1}_{i,{\sf C1}}$ and removes it from the this queue. If at time instant $t$ Case $15$ occurs and $Q^{j-1}_{i,{\sf C1}}$ is empty, then ${\sf Tx}_i$ remains silent. This way, similar to pairing Type A described previously, the transmitted bits in Case $1$ of the previous block can be decoded at the corresponding receiver. 

Furthermore, if at a given time instant Case $14$ ($13$) occurs, ${\sf Tx}_1$ (${\sf Tx}_2$) sends a bit from queue $Q^{j-1}_{1,{\sf C2}}$ ($Q^{j-1}_{2,{\sf C3}}$) and removes it from the this queue. This is motivated by pairing Type B (C) described previously.

If at the end of block $j$, there exists a bit at either of the transmitters that has not yet been transmitted, or any of the queues $Q^{j-1}_{1,{\sf C}_1}$, $Q^{j-1}_{1,{\sf C}_2}$, $Q^{j-1}_{2,{\sf C}_1}$, or $Q^{j-1}_{2,{\sf C}_3}$ is not empty, we consider this event as error type-I and we halt the transmission.

Assuming that the transmission is not halted, let random variable $N^j_{i,{\sf C}_{\ell}}$ denote the number of bits in $Q^j_{i,{\sf C}_{\ell}}$, $\left( i, \ell \right) = (1,1), (1,2), (2,1), (2,3)$. Since transition of a bit to this state is distributed as independent Bernoulli RV, upon completion of block $j$, we have 
\begin{align}
&\mathbb{E}[N^j_{i,{\sf C}_{\ell}}] = \frac{\Pr\left( \mathrm{Case~\ell} \right)}{1 - \sum_{i=9,10,\ldots,16}{\Pr\left( \mathrm{Case~i} \right)}} m \nonumber \\
&~= \frac{1}{p} \Pr\left( \mathrm{Case~\ell} \right) m.
\end{align}

If the event $\left[ N^j_{i,{\sf C}_{\ell}} \geq \mathbb{E}[N^j_{i,{\sf C}_{\ell}}] + m^{\frac{2}{3}} \right]$ occurs, we consider it as error type-II and we halt the transmission. At the end of block $1$, we add $0$'s (if necessary) to $Q^j_{i,{\sf C}_{\ell}}$ so that the total number of bits is equal to $\mathbb{E}[N^j_{i,{\sf C}_{\ell}}] + m^{\frac{2}{3}}$. Using Chernoff-Hoeffding bound, we can show that the probability of errors of types I and II  decreases exponentially with $m$.

\vspace{1mm}

{\bf Achievability strategy for block $b+1$}: In the final communication block, transmitters do not communicate any new data bit. 

If at time instant $t$ Case $15$ occurs, ${\sf Tx}_i$ sends a bit from queue $Q^{b}_{i,{\sf C1}}$ and removes it from the this queue. If at time instant $t$ Case $15$ occurs and $Q^{b}_{i,{\sf C1}}$ is empty, then ${\sf Tx}_i$ remains silent. If at time instant $t$ Case $14$ ($13$) occurs, ${\sf Tx}_1$ (${\sf Tx}_2$) sends a bit from queue $Q^{b}_{1,{\sf C2}}$ ($Q^{b}_{2,{\sf C3}}$) and removes it from the this queue.

If at the end of block $b+1$, any of the states $Q^{b}_{1,{\sf C}_1}$, $Q^{b}_{1,{\sf C}_2}$, $Q^{b}_{2,{\sf C}_1}$, or $Q^{b}_{2,{\sf C}_3}$ is not empty, we consider this event as error type-I and we halt the transmission.

Note that if the transmission is not halted, any bit is either available at its intended receiver interference-free, or the interfering bit is provided to the receiver in the following block. The probability that the transmission strategy halts at the end of each block can be bounded by the summation of error probabilities of types I and II. Using Chernoff-Hoeffding bound, we can show that the probability that the transmission strategy halts at any point approaches zero as $m \rightarrow \infty$.

Now, since each block has $n = m/p + \left( 2/p^4 \right) m^{2/3}$ time instants and the probability that the transmission strategy halts at any point approaches zero as $m \rightarrow \infty$, we achieve a rate tuple 
\begin{align}
\frac{b}{b+1} \left( p, p \right),
\end{align}
as $m \rightarrow \infty$. Finally letting $b \rightarrow \infty$, we achieve the desired rate tuple.


\subsection{Achievabiliy Strategy for $0.5 < p \leq 1$} 

For $p = 1$, the capacity region is the same with no, delayed, or instantaneous CSIT. So in this section, we assume $0.5 < p < 1$. By symmetry, it suffices to describe the strategy for point $A = \left( p, 2pq \right)$. In this regime, we will take advantage of another pairing opportunity as described below.

$\bullet$ {Type D} [Cases $2, 3,$ and $12$]: This type of pairing is different from what we have described so far. In all previous types, we paired up cases that had zero capacity to cancel out interference in other cases. However, here all three cases have capacity $1$. By pairing all three cases together, we can communicate $4$ bits as depicted in Fig.~\ref{fig:TypeG}.

\begin{remark}
This coding opportunity can be applicable to DoF analysis of wireless networks with linear schemes in the context of $2 \times 2 \times 2$ layered networks (Section~III.A of~\cite{issa2013two}).
\end{remark}

\subsubsection{Overview} The achievability is again carried on over $b+1$ communication blocks, each block with $n$ time instants. We describe the achievability strategy for rate tuple 
\begin{align}
\left( R_1, R_2 \right) = \left( p, 2pq \right),
\end{align}
see Fig.~\ref{fig:plessmore}(b). 

Transmitters communicate fresh data bits in the first $b$ blocks and the final block is to help receivers decode their corresponding bits. At the end, using our scheme, we achieve rate tuple $\frac{b}{b+1} \left( p, 2pq \right)$ as $n \rightarrow \infty$. Finally, letting $b \rightarrow \infty$, we achieve the desired corner point. In our scheme, the transmitted bits in block $j$, $j=1,2,\ldots,b$, will be decoded by the end of block $j+1$.

\vspace{1mm}

\subsubsection{Achievability strategy} Let $\hbox{W}^j_i$ be the message of transmitter $i$ in block $j$. We assume $\hbox{W}^j_1 = a^j_1, a^j_2, \ldots, a^j_{m}$, and $\hbox{W}^j_2 = b^j_1, b^j_2, \ldots, b^j_{m_2}$ for $j=1,2,\ldots,b$, where $a^j_i$'s and $b^j_i$'s are picked uniformly and independently from $\{ 0,1 \}$, for some positive value of $m$ and $m_2 = 2qm$ (note that $2q < 1$). We set
\begin{align}
n = m/p + \left( 2/q^4 \right) m^{2/3}.
\end{align}


\vspace{1mm}

{\bf Achievability strategy for block $1$}: In the first communication block, at each time instant $t$, transmitter one sends a new data bit if $G_{11}[t] = 1$ except Case $1$. In other words, ${\sf Tx}_1$ sends a new data bit if either of the following channel realizations occurs (see Table~\ref{table:16cases}): Cases $2,3,4,5,6,7,$ and $8$. Transmitter two sends a new data bit if $G_{22}[t] = 1$ except Cases $1$ and $12$. In other words, ${\sf Tx}_2$ sends a new data bit if either of the following channel realizations occurs: Cases $2,3,4,9,10,$ and $11$. 

If at time instant $t$ where $t \leq \frac{q^2}{p^2} n$, Case $1$ occurs, then each transmitter sends out a new data bit. Then, ${\sf Tx}_i$ transfers its transmitted bit in Case $1$ to queue $Q^1_{i,{\sf C1}}$ for $t \leq \frac{q^2}{p^2} n$. If $t > \frac{q^2}{p^2} n$ and Case $1$ occurs, then ${\sf Tx}_1$ sends out a new data bit while ${\sf Tx}_2$ remains silent, see Fig.~\ref{fig:commblock1}. Note that these bits are delivered to ${\sf Rx}_1$ interference-free.

\begin{figure}[h]
\centering
\includegraphics[height = 3.75cm]{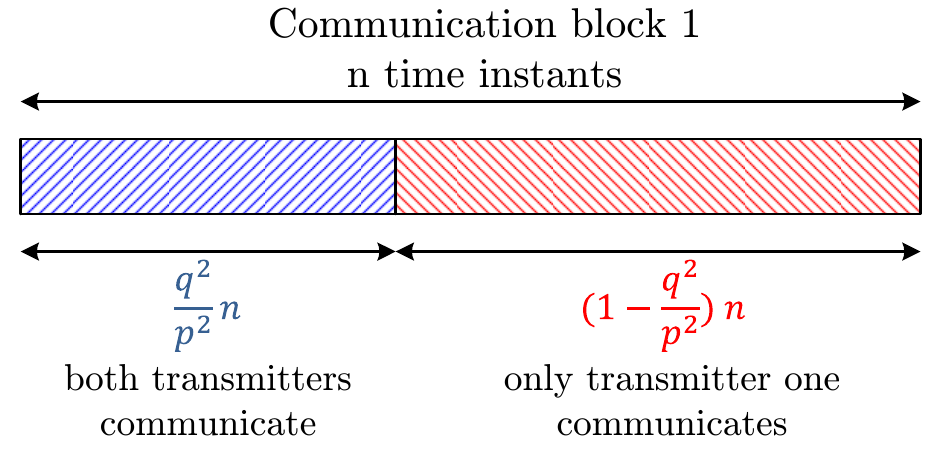}
\caption{If Case 1 occurs during communication block $1$, then if $t \leq \frac{q^2}{p^2} n$, each transmitter sends out a new data bit. However, if $t > \frac{q^2}{p^2} n$, then ${\sf Tx}_1$ sends out a new data bit while ${\sf Tx}_2$ remains silent.}\label{fig:commblock1}
\end{figure}

If $t \leq \frac{q^2}{p^2} n$, and Case $12$ occurs, then ${\sf Tx}_2$ sends out a new data bit while ${\sf Tx}_1$ remains silent. Note that these bits are delivered to ${\sf Rx}_2$ interference-free. 

If not specified, the transmitters remain silent. Note that ${\sf Tx}_1$ sends a bit if $G_{11}[t] = 1$ (\emph{i.e.} with probability $p$). On the other hand, ${\sf Tx}_2$ sends a bit with probability
\begin{align}
\sum_{j = 2,3,4,9,10,11}{\Pr\left( \mathrm{Case~j} \right)} + \frac{q^2}{p^2} \sum_{j = 1,12}{\Pr\left( \mathrm{Case~j} \right)} = 2pq.
\end{align}

Transmitter one transfers its transmitted bit in Case $2$ to queue $Q^1_{1,{\sf C2}}$; and ${\sf Tx}_2$ transfers its transmitted bit in Case $3$ to queue $Q^1_{2,{\sf C3}}$. If at the end of block $1$, there exists a bit at either of the transmitters that has not yet been transmitted, we consider it as error type-I and halt the transmission. 

\begin{remark}
Note that the transmitted bits in Cases $4,5,6,7,8,9,10,$ and $11$ are available at their corresponding receivers without any interference.
\end{remark}

Assuming that the transmission is not halted, let random variable $N^1_{i,{\sf C}_{\ell}}$ denote the number of bits in $Q^1_{i,{\sf C}_{\ell}}$, $\left( i, \ell \right) = (1,1), (1,2), (2,1), (2,3)$. Since transition of a bit to this state is distributed as independent Bernoulli RV, upon completion of block $1$, we have 
\begin{align}
& \mathbb{E}[N^1_{1,{\sf C}_1}] = \frac{\left( q^2/p^2 \right) \Pr\left( \mathrm{Case~1} \right)}{1 - \sum_{j = 9,10,\ldots,16}{\Pr\left( \mathrm{Case~j} \right)}} m  = pq^2 m, \nonumber \\
& \mathbb{E}[N^1_{1,{\sf C}_2}] = \frac{\Pr\left( \mathrm{Case~2} \right)}{1 - \sum_{j = 9,10,\ldots,16}{\Pr\left( \mathrm{Case~j} \right)}} m  = p^2q m, \nonumber \\ 
& \mathbb{E}[N^1_{2,{\sf C}_1}] \nonumber \\
&~= \frac{\left( q^2/p^2 \right) \Pr\left( \mathrm{Case~1} \right) \times 2qm}{\sum_{j = 2,3,4,9,10,11}{\Pr\left( \mathrm{Case~j} \right)} + \frac{q^2}{p^2} \sum_{j = 1,12}{\Pr\left( \mathrm{Case~j} \right)}}  \nonumber \\
&~= pq^2 m, \nonumber \\
& \mathbb{E}[N^1_{2,{\sf C}_3}] \nonumber \\
&~= \frac{\Pr\left( \mathrm{Case~3} \right) \times  2qm}{\sum_{j = 2,3,4,9,10,11}{\Pr\left( \mathrm{Case~j} \right)} + \frac{q^2}{p^2} \sum_{j = 1,12}{\Pr\left( \mathrm{Case~j} \right)}}  \nonumber \\
&~= p^2q m.
\end{align}

If the event $\left[ N^1_{i,{\sf C}_{\ell}} \geq \mathbb{E}[N^1_{i,{\sf C}_{\ell}}] + m^{\frac{2}{3}} \right]$ occurs, we consider it as error type-II and we halt the transmission. At the end of block $1$, we add $0$'s (if necessary) to $Q^1_{i,{\sf C}_{\ell}}$ so that the total number of bits is equal to $\mathbb{E}[N^1_{i,{\sf C}_{\ell}}] + m^{\frac{2}{3}}$. Using Chernoff-Hoeffding bound, we can show that the probability of errors of types I and II decreases exponentially with $m$.

\vspace{1mm}

{\bf Achievability strategy for block $j,$ $j=2,3,\ldots,b$}: In communication block $j$, $j=2,3,\ldots,b$, at each time instant $t$, transmitter one sends a new data bit if $G_{11}[t] = 1$ except Case $1$, while transmitter two sends a new data bit if $G_{22}[t] = 1$ except Cases $1$ and $12$.

If $t \leq \frac{q^2}{p^2} n$ and Case $1$ occurs, then each transmitter sends out a new data bit. Then ${\sf Tx}_i$ transfers its transmitted bit in Case $1$ to queue $Q^j_{i,{\sf C1}}$ for $t \leq \frac{q^2}{p^2} n$. If $t > \frac{q^2}{p^2} n$ and Case $1$ occurs, then ${\sf Tx}_1$ sends out a new data bit while ${\sf Tx}_2$ remains silent. Note that these bits are delivered to ${\sf Rx}_1$ interference-free.

If $t \leq \frac{q^2}{p^2} n$ and Case $12$ occurs, then ${\sf Tx}_2$ sends out a new data bit while ${\sf Tx}_1$ remains silent. We will exploit channel realization $12$ for $t > \frac{q^2}{p^2} n$, to perform pairing Type D.

Transmitter one transfers its transmitted bit in Case $2$ to queue $Q^j_{1,{\sf C2}}$; and transmitter two transfers its transmitted bit in Case $3$ to queue $Q^j_{2,{\sf C3}}$. Note that so far the transmission scheme is similar to the first communication block. 

Now, if at time instant $t$ Case $15$ occurs, ${\sf Tx}_i$ sends a bit from queue $Q^{j-1}_{i,{\sf C1}}$ and removes it from the this queue. If at time instant $t$ Case $15$ occurs and $Q^{j-1}_{i,{\sf C1}}$ is empty, then ${\sf Tx}_i$ remains silent. This way, similar to pairing Type A described previously, the transmitted bits in Case $1$ of the previous block can be decoded at the corresponding receiver. 

Furthermore, if at time instant $t$ Case $14$ ($13$) occurs, ${\sf Tx}_1$ (${\sf Tx}_2$) sends a bit from queue $Q^{j-1}_{1,{\sf C2}}$ ($Q^{j-1}_{2,{\sf C3}}$) and removes it from the this queue. This is motivated by pairing Type B (C) described previously.

Finally, if $t > \frac{q^2}{p^2} n$ and Case $12$ occurs, ${\sf Tx}_1$ sends a bit from queue $Q^{j-1}_{1,{\sf C2}}$ and ${\sf Tx}_2$ sends a bit from queue $Q^{j-1}_{2,{\sf C3}}$. Each transmitter removes the transmitted bit from the corresponding queue. This is motivated by pairing Type D described above.

If at the end of block $j$, there exists a bit at either of the transmitters that has not yet been transmitted, or any of the states $Q^{j-1}_{1,{\sf C1}}$, $Q^{j-1}_{1,{\sf C2}}$, $Q^{j-1}_{2,{\sf C1}}$, or $Q^{j-1}_{2,{\sf C3}}$ is not empty, we consider it as error type-I and halt the transmission.

Assuming that the transmission is not halted, let random variable $N^j_{i,{\sf C}_{\ell}}$ denote the number of bits in $Q^j_{i,{\sf C}_{\ell}}$, $\left( i, \ell \right) = (1,1), (1,2), (2,1), (2,3)$. Since transition of a bit to this state is distributed as independent Bernoulli RV, upon completion of block $j$, we have 
\begin{align}
& \mathbb{E}[N^j_{1,{\sf C}_1}] = \frac{\left( q^2/p^2 \right) \Pr\left( \mathrm{Case~1} \right)}{1 - \sum_{j = 9,10,\ldots,16}{\Pr\left( \mathrm{Case~j} \right)}} m  = pq^2 m, \nonumber \\
& \mathbb{E}[N^j_{1,{\sf C}_2}] = \frac{\Pr\left( \mathrm{Case~2} \right)}{1 - \sum_{j = 9,10,\ldots,16}{\Pr\left( \mathrm{Case~j} \right)}} m  = p^2q m, \nonumber \\
& \mathbb{E}[N^j_{2,{\sf C}_1}] \nonumber \\
&~= \frac{\left( q^2/p^2 \right) \Pr\left( \mathrm{Case~1} \right) \times  2qm}{\sum_{j = 2,3,4,9,10,11}{\Pr\left( \mathrm{Case~j} \right)} + \frac{q^2}{p^2} \sum_{j = 1,12}{\Pr\left( \mathrm{Case~j} \right)}} \nonumber \\
&~= pq^2 m, \nonumber \\
& \mathbb{E}[N^j_{2,{\sf C}_3}] \nonumber \\
&~= \frac{\Pr\left( \mathrm{Case~3} \right) \times 2qm}{\sum_{j = 2,3,4,9,10,11}{\Pr\left( \mathrm{Case~j} \right)} + \frac{q^2}{p^2} \sum_{j = 1,12}{\Pr\left( \mathrm{Case~j} \right)}} \nonumber \\
&~= p^2q m.
\end{align}

If the event $\left[ N^j_{i,{\sf C}_{\ell}} \geq \mathbb{E}[N^j_{i,{\sf C}_{\ell}}] + m^{\frac{2}{3}} \right]$ occurs, we consider it as error type-II and we halt the transmission. At the end of block $1$, we add $0$'s (if necessary) to $Q^j_{i,{\sf C}_{\ell}}$ so that the total number of bits is equal to $\mathbb{E}[N^j_{i,{\sf C}_{\ell}}] + m^{\frac{2}{3}}$. Using Chernoff-Hoeffding bound, we can show that the probability of errors of types I and II decreases exponentially with $m$.

\vspace{1mm}

{\bf Achievability strategy for block $b+1$}: In the final communication block, transmitters do not communicate any new data bit. 

If at time instant $t$ Case $15$ occurs, ${\sf Tx}_i$ sends a bit from queue $Q^{b}_{i,{\sf C1}}$ and removes it from the this queue. If at time instant $t$ Case $15$ occurs and $Q^{b}_{i,{\sf C1}}$ is empty, then ${\sf Tx}_i$ remains silent. If at time instant $t$ Case $14$ ($13$) or $12$ occurs, ${\sf Tx}_1$ (${\sf Tx}_2$) sends a bit from queue $Q^{b}_{1,{\sf C2}}$ ($Q^{b}_{2,{\sf C3}}$) and removes it from the this queue.

If at the end of block $j$ any of the states $Q^{b}_{1,{\sf C1}}$, $Q^{b}_{1,{\sf C2}}$, $Q^{b}_{2,{\sf C1}}$, or $Q^{b}_{2,{\sf C3}}$ is not empty, we consider it as error type-I and halt the transmission.

Note that if the transmission is not halted, any bit is either available at its intended receiver interference-free, or the interfering bits is provided to the receiver in the following block. The probability that the transmission strategy halts at the end of each block can be bounded by the summation of error probabilities of types I and II. Using Chernoff-Hoeffding bound, we can show that the probability that the transmission strategy halts at any point approaches zero for $m \rightarrow \infty$.

Now, since each block has $n = m/p + \left( 2/q^4 \right) m^{2/3}$ time instants and the probability that the transmission strategy halts at any point approaches zero for $m \rightarrow \infty$, we achieve a rate tuple 
\begin{align}
\frac{b}{b+1} \left( p, 2pq \right),
\end{align}
as $m \rightarrow \infty$. Finally letting $b \rightarrow \infty$, we achieve the desired rate tuple.



\section{Converse Proof of Theorem~\ref{THM:IC-InstCSIT}~[Instantaneous-CSIT]}
\label{sec:ConvInst}


The derivation of the outer-bound on individual rates is simple, however for the completeness of the results, we include the proof here. This outer-bound can be used for other theorems as needed. To derive the outer-bound on $R_1$, we have
\begin{align}
n R_1 & = H(W_1) \overset{(a)}= H(W_1|G^n) \nonumber \\
& \overset{(b)}= H(W_1|X_2^n,G^n) \nonumber \\
& \overset{(\mathrm{Fano})}\leq I(W_1;Y_1^n|X_2^n,G^n) + n \epsilon_n \nonumber \\
& \overset{(\mathrm{data~proc.})}\leq I(X_1^n;Y_1^n|X_2^n,G^n) + n \epsilon_n \nonumber \\
& = H(Y_1^n|X_2^n,G^n) - H(Y_1^n|X_1^n,X_2^n,G^n) + n \epsilon_n \nonumber \\
& = H(G_{11}^n X_1^n|X_2^n,G^n) + n \epsilon_n \nonumber \\
& \leq p n + n \epsilon_n,
\end{align}
where $\epsilon_n \rightarrow 0$ as $n \rightarrow \infty$; $(a)$ holds since message $\hbox{W}_1$ is independent of $G^n$; and $(b)$ holds since given $G^n$, $\hbox{W}_1$ is independent of $X_2^n$, see (\ref{eq:addX2}). Similarly, we have
\begin{align}
n R_2 \leq p n + n \epsilon_n.
\end{align}
dividing both sides by $n$ and let $n \rightarrow \infty$, we have 
\begin{equation}
\label{}
\left\{ \begin{array}{ll}
\vspace{1mm} R_1 \leq p & \\
R_2 \leq p & 
\end{array} \right.
\end{equation}

The outer-bound on $R_1 + R_2$ follows from the proof of Theorem~\ref{THM:IC-ICSITFB} in Section~\ref{sec:ConvInstFB}.


\section{Achievability Proof of Theorem~\ref{THM:IC-DelayedCSIT}: Sum-rate for $0 \leq p < 0.5$}
\label{Appendix:lessthanHalf}


In this appendix, we provide the achievability proof of Theorem~\ref{THM:IC-DelayedCSIT} with Delayed-CSIT and for $0 \leq p < 0.5$. We provide the an achievability strategy for rate tuple
\begin{align}
\label{eq:cornerlessthanhalf}
R_1 = R_2 = \min\left\{ p, \frac{(1-q^2)}{1+(1-q^2)^{-1}p} \right\} \raisebox{2pt}{.}
\end{align}

Let the messages of transmitters one and two be denoted by $\hbox{W}_1 = a_1,a_2,\ldots,a_m$, and $\hbox{W}_2 = b_1,b_2,\ldots,b_m$, respectively, where $a_i$'s and $b_i$'s are picked uniformly and independently from $\{ 0,1 \}$, $i=1,\ldots,m$. We show that it is possible to communicate these bits in 
\begin{align}
n & = \\
&\max \left\{ m/p, \left( 1 - q^2 \right)^{-1}m + \left( 1 - q^2 \right)^{-2} pm \right\} + O\left( m^{2/3}\right) \nonumber 
\end{align}
time instants with vanishing error probability (as $m \rightarrow \infty$). Therefore achieving the rates given in (\ref{eq:cornerlessthanhalf}) as $m \rightarrow \infty$.

\vspace{1mm}

\noindent {\bf Phase 1} [uncategorized transmission]: At the beginning of the communication block, we assume that the bits at ${\sf Tx}_i$ are in queue $Q_{i \rightarrow i}$, $i=1,2$. At each time instant, ${\sf Tx}_i$ sends out a bit from $Q_{i \rightarrow i}$ and this bit will either stay in the initial queue or a transition to a new queue will take place. Table~\ref{table:lessthanhalf} summarizes the transitions for each channel realization. The arguments are very similar to our discussion in Section~\ref{sec:achallhalf}, and the only difference is the way we handle Cases $7,8,11,$ and $12$. We provide some details about these cases. 

For Cases~7~$\left( \caseseven \right)$ and~8~$\left( \caseeight \right)$, in Section~\ref{sec:achallhalf}, we updated the status of the transmitted bit of ${\sf Tx}_2$ to $Q_{2 \rightarrow \{ 1,2 \}}$. However, this scheme is suboptimal for $0 \leq p < 0.5$, and instead we update the status of  the transmitted bit of ${\sf Tx}_2$ to an intermediate queue $Q_{2,INT}$. Then in Phase~2, we retransmit these bits and upgrade their status once more. Similar story holds for Cases~$11$ and $12$. The main reason for doing this is as follows. As we discussed in Section~\ref{sec:opportunities}, there are many opportunities to combine bits in order to improve the achievable rates. However, we could never combine the bits that were transmitted in Cases $7,8,11,$ or $12$ with other bits. This was not an issue for $0.5 \leq p \leq 1$, however for $0 \leq p < 0.5$, we need to find a way to combine these bits with other bits in future time instants. To do so, the only way is to keep them in an intermediate queue and retransmit them again in Phase~2.


\begin{table*}[t]
\caption{Summary of Phase $1$ for the Achievability Scheme of Corner Point $B$. Bit ``$a$'' represents a bit in $Q_{1 \rightarrow 1}$ while bit ``$b$'' represents a bit in $Q_{2 \rightarrow 2}$.}
\centering
\begin{tabular}{| c | c | c | c | c | c |}
\hline
case ID		 & channel realization    & state transition  & case ID		 & channel realization    & state transition \\
					 & at time instant $n$    &                   & 					 & at time instant $n$    &                  \\ [0.5ex]

\hline

\raisebox{18pt}{$1$}    &    \includegraphics[height = 1.4cm]{FiguresPDF/IC-c1.pdf}     &  \raisebox{18pt}{$ \left\{ \begin{array}{ll}  \vspace{1mm} a \rightarrow Q_{1,{\sf C}_1}  & \\ b \rightarrow Q_{2,{\sf C}_1}  &  \end{array} \right. $}  &  \raisebox{18pt}{$9$}    &    \includegraphics[height = 1.4cm]{FiguresPDF/IC-c9.pdf}     &   \raisebox{18pt}{$ \left\{ \begin{array}{ll}  \vspace{1mm} a \rightarrow Q_{1 \rightarrow 1}  & \\ b \rightarrow Q_{2,F}  &  \end{array} \right. $} \\

\hline

\raisebox{18pt}{$2$}    &    \includegraphics[height = 1.4cm]{FiguresPDF/IC-c2.pdf}     &  \raisebox{18pt}{$ \left\{ \begin{array}{ll}  \vspace{1mm} a \rightarrow Q_{1 \rightarrow 2|1}  & \\ b \rightarrow Q_{2,F}  &  \end{array} \right. $} &  \raisebox{18pt}{$10$}    &    \includegraphics[height = 1.4cm]{FiguresPDF/IC-c10.pdf}     &   \raisebox{18pt}{$ \left\{ \begin{array}{ll} \vspace{1mm} a \rightarrow Q_{1 \rightarrow 1}  & \\ b \rightarrow Q_{2,F}  &  \end{array} \right. $} \\

\hline

\raisebox{18pt}{$3$}    &    \includegraphics[height = 1.4cm]{FiguresPDF/IC-c3.pdf}     &  \raisebox{18pt}{$ \left\{ \begin{array}{ll}  \vspace{1mm} a \rightarrow Q_{1,F}  & \\ b \rightarrow Q_{2 \rightarrow 1|2}   &  \end{array} \right. $} &  \raisebox{18pt}{$11$}    &    \includegraphics[height = 1.4cm]{FiguresPDF/IC-c11.pdf}     &   \raisebox{18pt}{$ \left\{ \begin{array}{ll}  \vspace{1mm} a \rightarrow Q_{1,INT}  & \\ b \rightarrow Q_{2,F} &  \end{array} \right. $} \\

\hline

\raisebox{18pt}{$4$}    &    \includegraphics[height = 1.4cm]{FiguresPDF/IC-c4.pdf}     &   \raisebox{18pt}{$ \left\{ \begin{array}{ll} \vspace{1mm} a \rightarrow Q_{1,F}  & \\ b \rightarrow Q_{2,F}  &  \end{array} \right. $} &  \raisebox{18pt}{$12$}    &    \includegraphics[height = 1.4cm]{FiguresPDF/IC-c12.pdf}     &   \raisebox{18pt}{$ \left\{ \begin{array}{ll}  \vspace{1mm} a \rightarrow Q_{1,INT}  & \\ b \rightarrow Q_{2,F} &  \end{array} \right. $} \\

\hline

\raisebox{18pt}{$5$}    &    \includegraphics[height = 1.4cm]{FiguresPDF/IC-c5.pdf}     &   \raisebox{18pt}{$ \left\{ \begin{array}{ll} \vspace{1mm} a \rightarrow Q_{1,F}  & \\ b \rightarrow Q_{2 \rightarrow 2}  &  \end{array} \right. $} &  \raisebox{18pt}{$13$}    &    \includegraphics[height = 1.4cm]{FiguresPDF/IC-c13.pdf}     &   \raisebox{18pt}{$ \left\{ \begin{array}{ll}  \vspace{1mm} a \rightarrow Q_{1 \rightarrow 1}  & \\ b \rightarrow Q_{2 \rightarrow 2|1}  &  \end{array} \right. $} \\

\hline

\raisebox{18pt}{$6$}    &    \includegraphics[height = 1.4cm]{FiguresPDF/IC-c6.pdf}     &   \raisebox{18pt}{$ \left\{ \begin{array}{ll} \vspace{1mm} a \rightarrow Q_{1,F}  & \\ b \rightarrow Q_{2 \rightarrow 2}  &  \end{array} \right. $} &  \raisebox{18pt}{$14$}    &    \includegraphics[height = 1.4cm]{FiguresPDF/IC-c14.pdf}     &   \raisebox{18pt}{$ \left\{ \begin{array}{ll}  \vspace{1mm} a \rightarrow Q_{1 \rightarrow 1|2}  & \\ b \rightarrow Q_{2 \rightarrow 2}  &  \end{array} \right. $} \\

\hline

\raisebox{18pt}{$7$}    &    \includegraphics[height = 1.4cm]{FiguresPDF/IC-c7.pdf}     &   \raisebox{18pt}{$ \left\{ \begin{array}{ll}  \vspace{1mm} a \rightarrow Q_{1,F}  & \\ b \rightarrow Q_{2,INT} &  \end{array} \right. $} &  \raisebox{18pt}{$15$}    &    \includegraphics[height = 1.4cm]{FiguresPDF/IC-c15.pdf}     &   \raisebox{18pt}{$ \left\{ \begin{array}{ll} \vspace{1mm} a \rightarrow Q_{1 \rightarrow 1|2}  & \\ b \rightarrow Q_{2 \rightarrow 2|1}  &  \end{array} \right. $} \\

\hline

\raisebox{18pt}{$8$}    &    \includegraphics[height = 1.4cm]{FiguresPDF/IC-c8.pdf}     &   \raisebox{18pt}{$ \left\{ \begin{array}{ll}  \vspace{1mm} a \rightarrow Q_{1,F}  & \\ b \rightarrow Q_{2,INT} &  \end{array} \right. $} &  \raisebox{18pt}{$16$}    &    \includegraphics[height = 1.4cm]{FiguresPDF/IC-c16.pdf}     &   \raisebox{18pt}{$ \left\{ \begin{array}{ll} \vspace{1mm} a \rightarrow Q_{1 \rightarrow 1}  & \\ b \rightarrow Q_{2 \rightarrow 2}  &  \end{array} \right. $} \\

\hline

\end{tabular}
\label{table:lessthanhalf}
\end{table*}


Phase $1$ goes on for 
\begin{align}
\left( 1 - q^2 \right)^{-1} m + m^{\frac{2}{3}}
\end{align}
time instants and if at the end of this phase either of the queues $Q_{i \rightarrow i}$ is not empty, we declare error type-I and halt the transmission.

Assuming that the transmission is not halted, upon completion of Phase $1$, we have 
{\small \begin{align}
& \mathbb{E}[N_{1,{\sf C}_1}] = \frac{\Pr\left( \mathrm{Case~1} \right) m}{1 - \sum_{i=9,10,13,16}{\Pr\left( \mathrm{Case~i} \right)}} = (1-q^2)^{-1} p^4 m, \nonumber \\
& \mathbb{E}[N_{1 \rightarrow 2|1}] = \frac{\Pr\left( \mathrm{Case~2} \right) m}{1 - \sum_{i=9,10,13,16}{\Pr\left( \mathrm{Case~i} \right)}}  = (1-q^2)^{-1} p^3q m, \nonumber \\
& \mathbb{E}[N_{1 \rightarrow 1|2}] = \frac{\sum_{j=14,15}{\Pr\left( \mathrm{Case~j} \right)}m}{1 - \sum_{i=9,10,13,16}{\Pr\left( \mathrm{Case~i} \right)}}  = (1-q^2)^{-1} pq^2 m, \nonumber \\
& \mathbb{E}[N_{1,INT}] = \frac{\sum_{j=11,12}{\Pr\left( \mathrm{Case~j} \right)}m}{1 - \sum_{i=9,10,13,16}{\Pr\left( \mathrm{Case~i} \right)}}  = (1-q^2)^{-1} p^2q m,
\end{align}}
similarly, we have
{\small \begin{align}
& \mathbb{E}[N_{2,{\sf C}_1}] = \frac{\Pr\left( \mathrm{Case~1} \right)m}{1 - \sum_{i=9,10,13,16}{\Pr\left( \mathrm{Case~i} \right)}} = (1-q^2)^{-1} p^4 m, \nonumber \\
& \mathbb{E}[N_{2 \rightarrow 1|2}] = \frac{\Pr\left( \mathrm{Case~3} \right)m}{1 - \sum_{i=9,10,13,16}{\Pr\left( \mathrm{Case~i} \right)}}  = (1-q^2)^{-1} p^3q m, \nonumber \\
& \mathbb{E}[N_{2 \rightarrow 2|1}] = \frac{\sum_{j=13,15}{\Pr\left( \mathrm{Case~j} \right)}m}{1 - \sum_{i=9,10,13,16}{\Pr\left( \mathrm{Case~i} \right)}}  = (1-q^2)^{-1} pq^2 m, \nonumber \\
& \mathbb{E}[N_{2,INT}] = \frac{\sum_{j=7,8}{\Pr\left( \mathrm{Case~j} \right)}m}{1 - \sum_{i=9,10,13,16}{\Pr\left( \mathrm{Case~i} \right)}}  = (1-q^2)^{-1} p^2q m,
\end{align}}

If the event $\left[ N \geq \mathbb{E}[N] + m^{\frac{2}{3}} \right]$ for $N = N_{i,{\sf C}_1}, N_{i \rightarrow i|\bar{i}}, N_{i \rightarrow \bar{i}|i}, N_{i,INT}$, $i=1,2$, occurs, we consider it as error type-II and we halt the transmission strategy. At the end of Phase $1$, we add $0$'s (if necessary) in order to make queues $Q_{i,{\sf C}_1}$, $Q_{i \rightarrow j|\bar{j}}$, and $Q_{i,INT}$ of size equal to $\mathbb{E}[N_{i,{\sf C}_1}] + m^{\frac{2}{3}}$, $\mathbb{E}[N_{i \rightarrow j|\bar{j}}] + m^{\frac{2}{3}}$, and $\mathbb{E}[N_{i,INT}] + m^{\frac{2}{3}}$ respectively, $i=1,2$, and $j = i,\bar{i}$. Furthermore, using Chernoff-Hoeffding bound, we can show that the probability of errors of types I and II decreases exponentially with $m$.

\vspace{1mm}

\noindent{\bf Phase 2} [upgrading status of interfering bits in $Q_{1,{\sf C}_j}$]: In this phase, we focus on the bits in $Q_{1,INT}$ and $Q_{2,INT}$. At each time instant, ${\sf Tx}_i$ picks a bit from $Q_{i,INT}$ and sends it. This bit will either stay in $Q_{i,INT}$ or a transition to a new queue will take place. Table~\ref{table:upgradesum} describes what happens to the status of the bits if either of the $16$ possible cases occurs. Due to symmetry, we only describe the transitions for bits in $Q_{1,INT}$. Consider bit ``$a$'' in $Q_{1,INT}$.

\begin{itemize}

\item Cases $1,2,3,4,$ and $5$: The transitions for these cases are consistent with our previous discussions.

\item Cases $9,10,11,12,13,$ and $16$: In these cases, it is easy to see that no change occurs in the status of bit $a$.

\item Case $6$: In this case, bit $a$ is delivered to both receivers and hence, no further transmission is required. Therefore, it joins $Q_{1,F}$.

\item Case $7$: Here with slight abuse of, $Q_{1,{\sf C}_1}$ represents the bits of ${\sf Tx}_1$ that are received at both receivers with interference but not necessarily in Case~$1$, $i=1,2$. For instance if at a given time, ${\sf Tx}_1$ sends a bit from $Q_{1,INT}$ and Case~$7$ occurs, then this bit joins $Q_{1,{\sf C}_1}$ since now both receivers have received this bit with interference. 

\item Case $8$: In this case, bit $a$ is available at ${\sf Rx}_2$ but it is interfered at ${\sf Rx}_1$ by bit $b$. However, in Case $8$ no change occurs for the bits in $Q_{2,INT}$. Therefore, since bit $b$ will be retransmitted until it is provided to ${\sf Rx}_1$, no retransmission is required for bit $a$ and it joins $Q_{1,F}$.

\item Cases $14$ and $15$: If either of these cases occur, bit $a$ becomes available at ${\sf Rx}_2$ and is needed at ${\sf Rx}_1$. Thus, we update the status of such bits to $Q_{1 \rightarrow 1|2}$.

\end{itemize}

\begin{table*}[t]
\caption{Summary of Phase $2$ for the Achievability Scheme of Corner Point $B$. Bit ``$a$'' represents a bit in $Q_{1,INT}$ while bit ``$b$'' represents a bit in $Q_{2,INT}$.}
\centering
\begin{tabular}{| c | c | c | c | c | c |}
\hline
case ID		 & channel realization    & state transition  & case ID		 & channel realization    & state transition \\
					 & at time instant $n$    &                   & 					 & at time instant $n$    &                  \\ [0.5ex]

\hline

\raisebox{18pt}{$1$}    &    \includegraphics[height = 1.4cm]{FiguresPDF/IC-c1.pdf}     &  \raisebox{18pt}{$ \left\{ \begin{array}{ll}  \vspace{1mm} a \rightarrow Q_{1,{\sf C}_1}  & \\ b \rightarrow Q_{2,{\sf C}_1}  &  \end{array} \right. $}  &  \raisebox{18pt}{$9$}    &    \includegraphics[height = 1.4cm]{FiguresPDF/IC-c9.pdf}     &   \raisebox{18pt}{$ \left\{ \begin{array}{ll}  \vspace{1mm} a \rightarrow Q_{1,INT}  & \\ b \rightarrow Q_{2 \rightarrow 1|2}  &  \end{array} \right. $} \\

\hline

\raisebox{18pt}{$2$}    &    \includegraphics[height = 1.4cm]{FiguresPDF/IC-c2.pdf}     &  \raisebox{18pt}{$ \left\{ \begin{array}{ll}  \vspace{1mm} a \rightarrow Q_{1 \rightarrow 2|1}  & \\ b \rightarrow Q_{2,F}  &  \end{array} \right. $} &  \raisebox{18pt}{$10$}    &    \includegraphics[height = 1.4cm]{FiguresPDF/IC-c10.pdf}     &   \raisebox{18pt}{$ \left\{ \begin{array}{ll} \vspace{1mm} a \rightarrow Q_{1,INT}  & \\ b \rightarrow Q_{2,F}  &  \end{array} \right. $} \\

\hline

\raisebox{18pt}{$3$}    &    \includegraphics[height = 1.4cm]{FiguresPDF/IC-c3.pdf}     &  \raisebox{18pt}{$ \left\{ \begin{array}{ll}  \vspace{1mm} a \rightarrow Q_{1,F}  & \\ b \rightarrow Q_{2 \rightarrow 1|2}   &  \end{array} \right. $} &  \raisebox{18pt}{$11$}    &    \includegraphics[height = 1.4cm]{FiguresPDF/IC-c11.pdf}     &   \raisebox{18pt}{$ \left\{ \begin{array}{ll}  \vspace{1mm} a \rightarrow Q_{1,INT}  & \\ b \rightarrow Q_{2,{\sf C}_1} &  \end{array} \right. $} \\

\hline

\raisebox{18pt}{$4$}    &    \includegraphics[height = 1.4cm]{FiguresPDF/IC-c4.pdf}     &   \raisebox{18pt}{$ \left\{ \begin{array}{ll} \vspace{1mm} a \rightarrow Q_{1 \rightarrow 2|1}  & \\ b \rightarrow Q_{2 \rightarrow 1|2}  &  \end{array} \right. $} &  \raisebox{18pt}{$12$}    &    \includegraphics[height = 1.4cm]{FiguresPDF/IC-c12.pdf}     &   \raisebox{18pt}{$ \left\{ \begin{array}{ll}  \vspace{1mm} a \rightarrow Q_{1,INT}  & \\ b \rightarrow Q_{2,F} &  \end{array} \right. $} \\

\hline

\raisebox{18pt}{$5$}    &    \includegraphics[height = 1.4cm]{FiguresPDF/IC-c5.pdf}     &   \raisebox{18pt}{$ \left\{ \begin{array}{ll} \vspace{1mm} a \rightarrow Q_{1 \rightarrow 2|1}  & \\ b \rightarrow Q_{2,INT}  &  \end{array} \right. $} &  \raisebox{18pt}{$13$}    &    \includegraphics[height = 1.4cm]{FiguresPDF/IC-c13.pdf}     &   \raisebox{18pt}{$ \left\{ \begin{array}{ll}  \vspace{1mm} a \rightarrow Q_{1,INT}  & \\ b \rightarrow Q_{2 \rightarrow 2|1}  &  \end{array} \right. $} \\

\hline

\raisebox{18pt}{$6$}    &    \includegraphics[height = 1.4cm]{FiguresPDF/IC-c6.pdf}     &   \raisebox{18pt}{$ \left\{ \begin{array}{ll} \vspace{1mm} a \rightarrow Q_{1,F}  & \\ b \rightarrow Q_{2,INT}  &  \end{array} \right. $} &  \raisebox{18pt}{$14$}    &    \includegraphics[height = 1.4cm]{FiguresPDF/IC-c14.pdf}     &   \raisebox{18pt}{$ \left\{ \begin{array}{ll}  \vspace{1mm} a \rightarrow Q_{1 \rightarrow 1|2}  & \\ b \rightarrow Q_{2,INT}  &  \end{array} \right. $} \\

\hline

\raisebox{18pt}{$7$}    &    \includegraphics[height = 1.4cm]{FiguresPDF/IC-c7.pdf}     &   \raisebox{18pt}{$ \left\{ \begin{array}{ll}  \vspace{1mm} a \rightarrow Q_{1,{\sf C}_1}  & \\ b \rightarrow Q_{2,INT} &  \end{array} \right. $} &  \raisebox{18pt}{$15$}    &    \includegraphics[height = 1.4cm]{FiguresPDF/IC-c15.pdf}     &   \raisebox{18pt}{$ \left\{ \begin{array}{ll} \vspace{1mm} a \rightarrow Q_{1 \rightarrow 1|2}  & \\ b \rightarrow Q_{2 \rightarrow 2|1}  &  \end{array} \right. $} \\

\hline

\raisebox{18pt}{$8$}    &    \includegraphics[height = 1.4cm]{FiguresPDF/IC-c8.pdf}     &   \raisebox{18pt}{$ \left\{ \begin{array}{ll}  \vspace{1mm} a \rightarrow Q_{1,F}  & \\ b \rightarrow Q_{2,INT} &  \end{array} \right. $} &  \raisebox{18pt}{$16$}    &    \includegraphics[height = 1.4cm]{FiguresPDF/IC-c16.pdf}     &   \raisebox{18pt}{$ \left\{ \begin{array}{ll} \vspace{1mm} a \rightarrow Q_{1,INT}  & \\ b \rightarrow Q_{2,INT}  &  \end{array} \right. $} \\

\hline

\end{tabular}
\label{table:upgradesum}
\end{table*}

Phase $2$ goes on for 
{\small \begin{align}
& \left( 1 - \sum_{i = 9,10,11,12,13,16}{\Pr\left( \mathrm{Case}~i \right)} \right)^{-1} \left( 1 - q^2 \right)^{-1} p^2q m + 2 m^{\frac{2}{3}} \nonumber \\
& \quad = \left( 1 - \left[ p^2q + q^2 \right] \right)^{-1} \left( 1 - q^2 \right)^{-1} p^2q m + 2 m^{\frac{2}{3}}
\end{align}}
time instants and if at the end of this phase either of the states $Q_{i,INT}$ is not empty, we declare error type-I and halt the transmission. 

Assuming that the transmission is not halted, upon completion of Phase $2$, the states $Q_{1,INT}$ and $Q_{2,INT}$ are empty and we have
\begin{align}
\mathbb{E}& [N_{i,{\sf C}_1}] = (1-q^2)^{-1} p^4 m \\
&~+ \frac{\sum_{j=1,7}{\Pr\left( \mathrm{Case~j} \right)}}{1 - \sum_{i=9,10,11,12,13,16}{\Pr\left( \mathrm{Case~i} \right)}} ( p^2q m + m^{2/3} ) \nonumber \\ 
& = (1-q^2)^{-1} p^4 m \nonumber \\
&~+ \left( 1 - \left[ p^2q + q^2 \right] \right)^{-1} \left( p^4 + p^2q^2 \right) \left( p^2q m + m^{2/3} \right), \nonumber 
\end{align}
similarly, we have
\begin{align}
\mathbb{E}& [N_{i \rightarrow \bar{i}|i}] = (1-q^2)^{-1} p^3q m 
\end{align}
\begin{align}
&~+ \frac{\sum_{j=2,4,5}{\Pr\left( \mathrm{Case~j} \right)}}{1 - \sum_{i=9,10,11,12,13,16}{\Pr\left( \mathrm{Case~i} \right)}} ( p^2q m + m^{2/3} ) \nonumber \\ 
& = (1-q^2)^{-1} p^3q m \nonumber \\
&~+ \left( 1 - \left[ p^2q + q^2 \right] \right)^{-1} \left( p^2q + p^2q^2 \right) \left( p^2q m + m^{2/3} \right), \nonumber
\end{align}
and
\begin{align}
\mathbb{E}& [N_{i \rightarrow i|\bar{i}}] = (1-q^2)^{-1} pq^2 m \\
&~+ \frac{\sum_{j=14,15}{\Pr\left( \mathrm{Case~j} \right)}}{1 - \sum_{i=9,10,11,12,13,16}{\Pr\left( \mathrm{Case~i} \right)}} ( p^2q m + m^{2/3} ) \nonumber \\ 
& = (1-q^2)^{-1} pq^2 m \nonumber \\
&~+ \left( 1 - \left[ p^2q + q^2 \right] \right)^{-1} pq^2 \left( p^2q m + m^{2/3} \right). \nonumber
\end{align}

If the event $\left[ N \geq \mathbb{E}[N] + m^{\frac{2}{3}} \right]$ for $N = N_{i,{\sf C}_1}, N_{i \rightarrow i|\bar{i}}, N_{i \rightarrow \bar{i}|i}$, $i=1,2$, occurs, we consider it as error type-II and we halt the transmission strategy. At the end of Phase $2$, we add $0$'s (if necessary) in order to make queues $Q_{i,{\sf C}_1}$ and $Q_{i \rightarrow j|\bar{j}}$ of size equal to $\mathbb{E}[N_{i,{\sf C}_1}] + m^{\frac{2}{3}}$, and $\mathbb{E}[N_{i \rightarrow j|\bar{j}}] + m^{\frac{2}{3}}$ respectively, $i=1,2$, and $j = i,\bar{i}$. Using Chernoff-Hoeffding bound, we can show that the probability of errors of types I and II decreases exponentially with $m$.

%

\vspace{1mm}

\noindent{\bf Phase 3} [encoding and retransmission]: In this phase, ${\sf Tx}_i$ communicates the bits in $Q_{i \rightarrow i|\bar{i}}$ to ${\sf Rx}_i$, $i=1,2$. However, it is possible to create XOR of these bits with the bits in $Q_{i \rightarrow \bar{i}|i}$ and the bits in $Q_{i,{\sf C}_1}$ to create bits of common interest. To do so, we first encode the bits in these states using the results of~\cite{Elias}, and then we create the XOR of the encoded bits. 

In other words, given $\epsilon, \delta > 0$, ${\sf Tx}_i$ encodes all the bits in $Q_{i \rightarrow i|\bar{i}}$ at rate $p - \delta$ using random coding scheme of~\cite{Elias}. Similarly, ${\sf Tx}_i$ encodes $q \left( \mathbb{E}[N_{i \rightarrow i|\bar{i}}]  + m^{\frac{2}{3}} \right)$ bits from $Q_{i \rightarrow \bar{i}|i}$ and $Q_{i,{\sf C}_1}$ at rate $pq - \delta$ (if there are less bits in $Q_{i \rightarrow \bar{i}|i}$ and $Q_{i,{\sf C}_1}$, then encode all of the bits in these queues). More precisely, first ${\sf Tx}_i$ encodes bits from $Q_{i \rightarrow \bar{i}|i}$, and if the number of bits in $Q_{i \rightarrow \bar{i}|i}$ is less than $q \left( \mathbb{E}[N_{i \rightarrow i|\bar{i}}]  + m^{\frac{2}{3}} \right)$, then ${\sf Tx}_i$ uses bits in $Q_{i,{\sf C}_1}$. ${\sf Tx}_i$ will then communicate the XOR of these encoded bits. 

Note that since ${\sf Rx}_i$ already knows the bits in $Q_{i \rightarrow \bar{i}|i}$, it can remove the corresponding part of the received signal. Then since the channel from ${\sf Tx}_{\bar{i}}$ to ${\sf Rx}_i$ can be viewed as a binary erasure channel with success probability of $pq$, from~\cite{Elias}, we know that ${\sf Rx}_i$ can decode $Q_{\bar{i},{\sf C}_1}$ with decoding error probability less than or equal to $\epsilon$. Thus, ${\sf Rx}_i$ can decode the transmitted bits from $Q_{\bar{i},{\sf C}_1}$ and use them to decode the bits in $Q_{i,{\sf C}_1}$. Then, ${\sf Rx}_i$ removes the contribution of the bits in $Q_{i,{\sf C}_1}$ the received signal. Finally, since the channel from ${\sf Tx}_i$ to ${\sf Rx}_i$ can be viewed as a binary erasure channel with success probability of $p$, from~\cite{Elias}, we know that ${\sf Rx}_i$ can decode $Q_{i \rightarrow i|\bar{i}}$ with decoding error probability less than or equal to $\epsilon$.

If an error occurs in decoding of the encoded bits, we halt the transmission. Assuming that the transmission is not halted, at the end of Phase~3, $Q_{i \rightarrow i|\bar{i}}$ becomes empty and there are 
\begin{align}
\left( \mathbb{E}[N_{i \rightarrow \bar{i}|i}] + \mathbb{E}[N_{i,{\sf C}_1}] + 2 m^{\frac{2}{3}} - q \left( \mathbb{E}[N_{i \rightarrow i|\bar{i}}]  + m^{\frac{2}{3}} \right) \right)^+
\end{align}
bits left in $Q_{i \rightarrow \bar{i}|i}$ and $Q_{i,{\sf C}_1}$.

If $Q_{i \rightarrow \bar{i}|i}$ and $Q_{i,{\sf C}_1}$ are also empty, the transmission strategy ends here. Otherwise, we merge the remaining bits in $Q_{i \rightarrow \bar{i}|i}$ (if any) with the bits in $Q_{i,{\sf C}_1}$ as {\bf Type-III} (see Section~\ref{sec:opportunities}) and put the XOR of them in $Q_{i \rightarrow \{ 1,2\}}$, $i=1,2$. Finally, we need to describe what happens to the remaining bits in $Q_{i,{\sf C}_1}$. As mentioned before, a bit in $Q_{i,{\sf C}_1}$ can be viewed as a bit of common interest by itself. For the remaining bits in $Q_{1,{\sf C}_1}$, we put the first half in $Q_{1 \rightarrow \{ 1,2 \}}$ (suppose $m$ is picked such that the remaining number of bits is even). Note that if these bits are delivered to ${\sf Rx}_2$, then ${\sf Rx}_2$ can decode the first half of the bits in $Q_{2,{\sf C}_1}$. Therefore, the first half of the bits in $Q_{2,{\sf C}_1}$ join $Q_{2,F}$. 

\vspace{1mm}

\noindent{\bf Phase 4} [communicating bits of common interest]:  During Phase~4, we deliver the bits in $Q_{1 \rightarrow \{ 1,2 \}}$ and $Q_{2 \rightarrow \{ 1,2 \}}$ using the transmission strategy for the two-source multicast problem. More precisely, the bits in $Q_{i \rightarrow \{ 1,2 \}}$ will be considered as the message of transmitter ${\sf Tx}_i$ and they will be encoded as in the achievability scheme of Lemma~\ref{lemma:multicast}, $i=1,2$. Fix $\epsilon, \delta > 0$, from Lemma~\ref{lemma:multicast} we know that the rate tuple $$\left( R_1, R_2 \right) = \frac{1}{2}\left( (1-q^2) - \delta, (1-q^2) - \delta \right)$$ is achievable with decoding error probability less than or equal to $\epsilon$. Thus, using Lemma~\ref{lemma:multicast}, we can communicate the remaining bits at rate $(1-q^2) - \delta$ with decoding error probability less than or equal to $\epsilon$. If an error occurs in decoding of the encoded bits, we halt the transmission.

Using Chernoff-Hoeffding bound and the results of~\cite{Elias}, we can show that the probability that the transmission strategy halts at any point approaches zero for $\epsilon, \delta \rightarrow 0$ and $m \rightarrow \infty$. Moreover, it is easy to verify that for $0 \leq p \leq \left( 3 - \sqrt{5} \right)/2$, at the end of Phase~3, $Q_{i \rightarrow \bar{i}|i}$ and $Q_{i,{\sf C}_1}$ are empty and the transmission strategy ends there. However, for $\left( 3 - \sqrt{5} \right)/2 < p < 0.5$, the transmission strategy continues to Phase~4. Therefore, we can show that if no error occurs, the transmission strategy end in 
\begin{align}
n & = \\
& \max \left\{ m/p, \left( 1 - q^2 \right)^{-1}m + \left( 1 - q^2 \right)^{-2} pm \right\} + O\left( m^{2/3}\right) \nonumber 
\end{align}
time instants. Therefore achieving the rates given in (\ref{eq:cornerlessthanhalf}).


\section{Achievability Proof of Theorem~\ref{THM:IC-DelayedCSIT}: Corner Point $C$}
\label{Appendix:cornerdelayed}


In this appendix, we describe the achievability strategy for corner point $C$ depicted in Fig.~\ref{fig:regionHalf}(b), \emph{i.e.}
\begin{align}
\left( R_1, R_2 \right) = \left( pq(1+q), p \right).
\end{align}
  
%

Let the messages of transmitters one and two be denoted by $\hbox{W}_1 = a_1,a_2,\ldots,a_{m_1}$, and $\hbox{W}_2 = b_1,b_2,\ldots,b_m$, respectively, where data bits $a_i$'s and $b_i$'s are picked uniformly and independently from $\{ 0,1 \}$, and $m_1 = q(1+q)m$ (suppose the parameters are such that $m, m_1 \in \mathbb{Z}$). Note that for $\left( 3 - \sqrt{5} \right)/2 < p \leq 1$, we have $q(1+q) < 1$. We show that it is possible to communicate these bits in
\begin{align}
n = \frac{1}{p} m + O\left( m^{2/3}\right)
\end{align}
time instants with vanishing error probability (as $m \rightarrow \infty$). Therefore, achieving corner point $C$ as $m \rightarrow \infty$. Our transmission strategy consists of five phases as described before.


\begin{table*}[t]
\caption{Summary of Phase $1$ for the Achievability Scheme of Corner Point $C$. Bit ``$a$'' represents a bit in $Q_{1 \rightarrow 1}$ while bit ``$b$'' represents a bit in $Q_{2 \rightarrow 2}$.}
\centering
\begin{tabular}{| c | c | c | c | c | c |}
\hline
case ID		 & channel realization    & state transition  & case ID		 & channel realization    & state transition \\
					 & at time instant $n$    &                   & 					 & at time instant $n$    &                  \\ [0.5ex]

\hline

\raisebox{18pt}{$1$}    &    \includegraphics[height = 1.4cm]{FiguresPDF/IC-c1.pdf}     &  \raisebox{18pt}{$ \left\{ \begin{array}{ll}  \vspace{1mm} a \rightarrow Q_{1,{\sf C}_1}  & \\ b \rightarrow Q_{2,F}  &  \end{array} \right. $}  &  \raisebox{18pt}{$9$}    &    \includegraphics[height = 1.4cm]{FiguresPDF/IC-c9.pdf}     &   \raisebox{18pt}{$ \left\{ \begin{array}{ll}  \vspace{1mm} a \rightarrow Q_{1 \rightarrow 1}  & \\ b \rightarrow Q_{2,F}  &  \end{array} \right. $} \\

\hline

\raisebox{18pt}{$2$}    &    \includegraphics[height = 1.4cm]{FiguresPDF/IC-c2.pdf}     &  \raisebox{18pt}{$ \left\{ \begin{array}{ll}  \vspace{1mm} a \rightarrow Q_{1,OP}  & \\ b \rightarrow Q_{2,F}  &  \end{array} \right. $} &  \raisebox{18pt}{$10$}    &    \includegraphics[height = 1.4cm]{FiguresPDF/IC-c10.pdf}     &   \raisebox{18pt}{$ \left\{ \begin{array}{ll} \vspace{1mm} a \rightarrow Q_{1 \rightarrow 1}  & \\ b \rightarrow Q_{2,F}  &  \end{array} \right. $} \\

\hline

\raisebox{18pt}{$3$}    &    \includegraphics[height = 1.4cm]{FiguresPDF/IC-c3.pdf}     &  \raisebox{18pt}{$ \left\{ \begin{array}{ll}  \vspace{1mm} a \rightarrow Q_{1,F}  & \\ b \rightarrow Q_{2,OP}   &  \end{array} \right. $} &  \raisebox{18pt}{$11$}    &    \includegraphics[height = 1.4cm]{FiguresPDF/IC-c11.pdf}     &   \raisebox{18pt}{$ \left\{ \begin{array}{ll}  \vspace{1mm} a \rightarrow Q_{1,INT}  & \\ b \rightarrow Q_{2,F} &  \end{array} \right. $} \\

\hline

\raisebox{18pt}{$4$}    &    \includegraphics[height = 1.4cm]{FiguresPDF/IC-c4.pdf}     &   \raisebox{18pt}{$ \left\{ \begin{array}{ll} \vspace{1mm} a \rightarrow Q_{1,F}  & \\ b \rightarrow Q_{2,F}  &  \end{array} \right. $} &  \raisebox{18pt}{$12$}    &    \includegraphics[height = 1.4cm]{FiguresPDF/IC-c12.pdf}     &   \raisebox{18pt}{$ \left\{ \begin{array}{ll}  \vspace{1mm} a \rightarrow Q_{1,INT}  & \\ b \rightarrow Q_{2,F} &  \end{array} \right. $} \\

\hline

\raisebox{18pt}{$5$}    &    \includegraphics[height = 1.4cm]{FiguresPDF/IC-c5.pdf}     &   \raisebox{18pt}{$ \left\{ \begin{array}{ll} \vspace{1mm} a \rightarrow Q_{1,F}  & \\ b \rightarrow Q_{2 \rightarrow 2}  &  \end{array} \right. $} &  \raisebox{18pt}{$13$}    &    \includegraphics[height = 1.4cm]{FiguresPDF/IC-c13.pdf}     &   \raisebox{18pt}{$ \left\{ \begin{array}{ll}  \vspace{1mm} a \rightarrow Q_{1 \rightarrow 1}  & \\ b \rightarrow Q_{2 \rightarrow 2|1}  &  \end{array} \right. $} \\

\hline

\raisebox{18pt}{$6$}    &    \includegraphics[height = 1.4cm]{FiguresPDF/IC-c6.pdf}     &   \raisebox{18pt}{$ \left\{ \begin{array}{ll} \vspace{1mm} a \rightarrow Q_{1,F}  & \\ b \rightarrow Q_{2 \rightarrow 2}  &  \end{array} \right. $} &  \raisebox{18pt}{$14$}    &    \includegraphics[height = 1.4cm]{FiguresPDF/IC-c14.pdf}     &   \raisebox{18pt}{$ \left\{ \begin{array}{ll}  \vspace{1mm} a \rightarrow Q_{1 \rightarrow 1|2}  & \\ b \rightarrow Q_{2 \rightarrow 2}  &  \end{array} \right. $} \\

\hline

\raisebox{18pt}{$7$}    &    \includegraphics[height = 1.4cm]{FiguresPDF/IC-c7.pdf}     &   \raisebox{18pt}{$ \left\{ \begin{array}{ll}  \vspace{1mm} a \rightarrow Q_{1,F}  & \\ b \rightarrow Q_{2,INT} &  \end{array} \right. $} &  \raisebox{18pt}{$15$}    &    \includegraphics[height = 1.4cm]{FiguresPDF/IC-c15.pdf}     &   \raisebox{18pt}{$ \left\{ \begin{array}{ll} \vspace{1mm} a \rightarrow Q_{1 \rightarrow 1|2}  & \\ b \rightarrow Q_{2 \rightarrow 2|1}  &  \end{array} \right. $} \\

\hline

\raisebox{18pt}{$8$}    &    \includegraphics[height = 1.4cm]{FiguresPDF/IC-c8.pdf}     &   \raisebox{18pt}{$ \left\{ \begin{array}{ll}  \vspace{1mm} a \rightarrow Q_{1,F}  & \\ b \rightarrow Q_{2,INT} &  \end{array} \right. $} &  \raisebox{18pt}{$16$}    &    \includegraphics[height = 1.4cm]{FiguresPDF/IC-c16.pdf}     &   \raisebox{18pt}{$ \left\{ \begin{array}{ll} \vspace{1mm} a \rightarrow Q_{1 \rightarrow 1}  & \\ b \rightarrow Q_{2 \rightarrow 2}  &  \end{array} \right. $} \\

\hline

\end{tabular}
\label{table:corner}
\end{table*}


\noindent {\bf Phase 1} [uncategorized transmission]: This phase is similar to Phase 1 of the achievability strategy for the optimal sum-rate point $A$. The main difference is due to the fact that the transmitters start with unequal number of bits. At the beginning of the communication block, we assume that the bits $a_1, a_2, \ldots, a_{m_1}$ at ${\sf Tx}_1$ and the bits $b_1, b_2, \ldots, b_{m_1}$ at ${\sf Tx}_2$ are in queues $Q_{1 \rightarrow 1}$ and $Q_{2 \rightarrow 2}$ respectively. 

\begin{remark}
Note that ${\sf Tx}_2$ has $m$ initial bits, however, only $m_1$ of them are in $Q_{2 \rightarrow 2}$ at the beginning of the communication block.
\end{remark}

At each time instant $t$, ${\sf Tx}_i$ sends out a bit from $Q_{i \rightarrow i}$, and this bit will either stay in the initial queues or a transition will take place. Based on the channel realizations, a total of $16$ possible configurations may occur at any time instant. Table~\ref{table:corner} summarizes the transition from the initial queue for each channel realization.

In comparison to the achievability strategy of the sum-rate point $A$, we have new queues for the bits: 
\begin{enumerate}
\item $Q_{i,OP}$ denotes the bits that have caused interference at the unintended receiver and this interference has to get resolved.

\item $Q_{1,INT}$ denotes an intermediate queue of the bits at ${\sf Tx}_1$ that were transmitted when channel realizations $11$ or $12$ occurred. 

\item $Q_{i,INT}$ denotes an intermediate queue of the bits at ${\sf Tx}_2$ that were transmitted when channel realizations $7$ or $8$ occurred. 
\end{enumerate}

Phase $1$ goes on for 
\begin{align}
\left( 1/p - 1 \right) m + m^{\frac{2}{3}}
\end{align} 
time instants and if at the end of this phase, either of queues $Q_{i \rightarrow i}$ is not empty, we declare error type-I and halt the transmission.

Assuming that the transmission is not halted, let random variable $N_{1,{\sf C}_1}$, $N_{i,OP}$, $N_{i \rightarrow i|\bar{i}}$, and $N_{i,INT}$ denote the number of bits in $Q_{1,{\sf C}_1}$, $Q_{i,OP}$, $Q_{i \rightarrow i|\bar{i}}$, and $Q_{i,INT}$ respectively $i=1,2$. The transmission strategy will be halted and an error (that we refer to as error type-II) will occur, if any of the following events happens.
\begin{align}
\label{eq:errortypeIcornerC}
& N_{1,{\sf C}_1} > \mathbb{E}[N_{1,{\sf C}_1}] + m^{\frac{2}{3}} \overset{\triangle}= n_{1,{\sf C}_1}; \nonumber \\
& N_{i,OP} > \mathbb{E}[N_{i,OP}] + m^{\frac{2}{3}} \overset{\triangle}= n_{i,OP}, \quad i=1,2; \nonumber \\
& N_{i \rightarrow i|\bar{i}} > \mathbb{E}[N_{i \rightarrow i|\bar{i}}] + m^{\frac{2}{3}} \overset{\triangle}= n_{i \rightarrow i|\bar{i}}, \quad i=1,2; \nonumber \\
& N_{i,INT} > \mathbb{E}[N_{i,INT}] + m^{\frac{2}{3}} \overset{\triangle}= n_{i,INT}, \quad i=1,2.
\end{align}

%
%
%
%
%

From basic probability, and we have 
\begin{align}
\label{eq:expectedvalues1}
\mathbb{E}[N_{1,{\sf C}_1}] & = \frac{\Pr\left( \mathrm{Case~1} \right)}{1 - \sum_{i=9,10,13,16}{\Pr\left( \mathrm{Case~i} \right)}} m_1 \nonumber \\
& =  (1-q^2)^{-1} p^4 m_1 = p^3q m, \nonumber \\
\mathbb{E}[N_{1,OP}] & = \frac{\Pr\left( \mathrm{Case~2} \right)}{1 - \sum_{i=9,10,13,16}{\Pr\left( \mathrm{Case~i} \right)}} m_1 \nonumber \\
& = (1-q^2)^{-1} p^3q m_1 = p^2q^2 m, \nonumber \\
\mathbb{E}[N_{2,OP}] & = \frac{\Pr\left( \mathrm{Case~3} \right)}{1 - \sum_{i=5,6,14,16}{\Pr\left( \mathrm{Case~i} \right)}} m_1 \nonumber \\ 
& = (1-q^2)^{-1} p^3q m_1 = p^2q^2 m, \nonumber \\
\mathbb{E}[N_{i \rightarrow i|\bar{i}}] & = \frac{\sum_{j=14,15}{\Pr\left( \mathrm{Case~j} \right)}}{1 - \sum_{i=9,10,13,16}{\Pr\left( \mathrm{Case~i} \right)}} m_1 \nonumber \\
& = (1-q^2)^{-1} \left( pq^3 + p^2q^2 \right) m_1 = q^3 m,~ i = 1,2, \nonumber \\
\mathbb{E}[N_{i,INT}] & = \frac{\sum_{j=11,12}{\Pr\left( \mathrm{Case~j} \right)}}{1 - \sum_{i=9,10,13,16}{\Pr\left( \mathrm{Case~i} \right)}} m_1 \nonumber \\
& = (1-q^2)^{-1} \left( p^3q + p^2q^2 \right) m_1 = pq^2 m,~ i = 1,2.
\end{align}

Furthermore, using Chernoff-Hoeffding bound, we can show that the probability of errors of types I and II decreases exponentially with $m$.

At the end of Phase $1$, we add $0$'s (if necessary) in order to make queues $Q_{1,{\sf C}_1}$, $Q_{i,OP}$, $Q_{i \rightarrow i|\bar{i}}$, and $Q_{i,INT}$ of size equal to $n_{1,{\sf C}_1}$, $n_{i,OP}$, $n_{i \rightarrow i|\bar{i}}$, and $n_{i,INT}$ respectively as defined in (\ref{eq:errortypeIcornerC}), $i=1,2$. For the rest of this appendix, we assume that Phase 1 is completed and no error has occurred.


\noindent{\bf Phase 2} [updating status of the bits in $Q_{i,INT}$]: In this phase, we focus on the bits in $Q_{i,INT}$, $i=1,2$. The ultimate goal is to deliver the bits in $Q_{i,INT}$ to both receivers. At each time instant, ${\sf Tx}_i$ picks a bit from $Q_{i,INT}$ and sends it. This bit will either stay in $Q_{i,INT}$ or a transition to a new queue will take place. Table~\ref{table:upgrade} describes what happens to the status of the bits if either of the $16$ cases occurs. 

\begin{table*}[t]
\caption{Summary of Phase $2$ for the Achievability Scheme of Corner Point $B$. Bit ``$a$'' represents a bit in $Q_{1,INT}$ while bit ``$b$'' represents a bit in $Q_{2,INT}$.}
\centering
\begin{tabular}{| c | c | c | c | c | c |}
\hline
case ID		 & channel realization    & state transition  & case ID		 & channel realization    & state transition \\
					 & at time instant $n$    &                   & 					 & at time instant $n$    &                  \\ [0.5ex]

\hline

\raisebox{18pt}{$1$}    &    \includegraphics[height = 1.4cm]{FiguresPDF/IC-c1.pdf}     &  \raisebox{18pt}{$ \left\{ \begin{array}{ll}  \vspace{1mm} a \rightarrow Q_{1,{\sf C}_1}  & \\ b \rightarrow Q_{2,F}  &  \end{array} \right. $}  &  \raisebox{18pt}{$9$}    &    \includegraphics[height = 1.4cm]{FiguresPDF/IC-c9.pdf}     &   \raisebox{18pt}{$ \left\{ \begin{array}{ll}  \vspace{1mm} a \rightarrow Q_{1,INT}  & \\ b \rightarrow Q_{2,OP}  &  \end{array} \right. $} \\

\hline

\raisebox{18pt}{$2$}    &    \includegraphics[height = 1.4cm]{FiguresPDF/IC-c2.pdf}     &  \raisebox{18pt}{$ \left\{ \begin{array}{ll}  \vspace{1mm} a \rightarrow Q_{1,OP}  & \\ b \rightarrow Q_{2,OP}  &  \end{array} \right. $} &  \raisebox{18pt}{$10$}    &    \includegraphics[height = 1.4cm]{FiguresPDF/IC-c10.pdf}     &   \raisebox{18pt}{$ \left\{ \begin{array}{ll} \vspace{1mm} a \rightarrow Q_{1,INT}  & \\ b \rightarrow Q_{2,F}  &  \end{array} \right. $} \\

\hline

\raisebox{18pt}{$3$}    &    \includegraphics[height = 1.4cm]{FiguresPDF/IC-c3.pdf}     &  \raisebox{18pt}{$ \left\{ \begin{array}{ll}  \vspace{1mm} a \rightarrow Q_{1,{\sf C}_1}  & \\ b \rightarrow Q_{2,OP}   &  \end{array} \right. $} &  \raisebox{18pt}{$11$}    &    \includegraphics[height = 1.4cm]{FiguresPDF/IC-c11.pdf}     &   \raisebox{18pt}{$ \left\{ \begin{array}{ll}  \vspace{1mm} a \rightarrow Q_{1,INT}  & \\ b \rightarrow Q_{2,OP} &  \end{array} \right. $} \\

\hline

\raisebox{18pt}{$4$}    &    \includegraphics[height = 1.4cm]{FiguresPDF/IC-c4.pdf}     &   \raisebox{18pt}{$ \left\{ \begin{array}{ll} \vspace{1mm} a \rightarrow Q_{1,OP}  & \\ b \rightarrow Q_{2,OP}  &  \end{array} \right. $} &  \raisebox{18pt}{$12$}    &    \includegraphics[height = 1.4cm]{FiguresPDF/IC-c12.pdf}     &   \raisebox{18pt}{$ \left\{ \begin{array}{ll}  \vspace{1mm} a \rightarrow Q_{1,INT}  & \\ b \rightarrow Q_{2,F} &  \end{array} \right. $} \\

\hline

\raisebox{18pt}{$5$}    &    \includegraphics[height = 1.4cm]{FiguresPDF/IC-c5.pdf}     &   \raisebox{18pt}{$ \left\{ \begin{array}{ll} \vspace{1mm} a \rightarrow Q_{1,OP}  & \\ b \rightarrow Q_{2,INT}  &  \end{array} \right. $} &  \raisebox{18pt}{$13$}    &    \includegraphics[height = 1.4cm]{FiguresPDF/IC-c13.pdf}     &   \raisebox{18pt}{$ \left\{ \begin{array}{ll}  \vspace{1mm} a \rightarrow Q_{1,INT}  & \\ b \rightarrow Q_{2 \rightarrow 2|1}  &  \end{array} \right. $} \\

\hline

\raisebox{18pt}{$6$}    &    \includegraphics[height = 1.4cm]{FiguresPDF/IC-c6.pdf}     &   \raisebox{18pt}{$ \left\{ \begin{array}{ll} \vspace{1mm} a \rightarrow Q_{1,F}  & \\ b \rightarrow Q_{2,INT}  &  \end{array} \right. $} &  \raisebox{18pt}{$14$}    &    \includegraphics[height = 1.4cm]{FiguresPDF/IC-c14.pdf}     &   \raisebox{18pt}{$ \left\{ \begin{array}{ll}  \vspace{1mm} a \rightarrow Q_{1 \rightarrow 1|2}  & \\ b \rightarrow Q_{2,INT}  &  \end{array} \right. $} \\

\hline

\raisebox{18pt}{$7$}    &    \includegraphics[height = 1.4cm]{FiguresPDF/IC-c7.pdf}     &   \raisebox{18pt}{$ \left\{ \begin{array}{ll}  \vspace{1mm} a \rightarrow Q_{1,OP}  & \\ b \rightarrow Q_{2,INT} &  \end{array} \right. $} &  \raisebox{18pt}{$15$}    &    \includegraphics[height = 1.4cm]{FiguresPDF/IC-c15.pdf}     &   \raisebox{18pt}{$ \left\{ \begin{array}{ll} \vspace{1mm} a \rightarrow Q_{1 \rightarrow 1|2}  & \\ b \rightarrow Q_{2 \rightarrow 2|1}  &  \end{array} \right. $} \\

\hline

\raisebox{18pt}{$8$}    &    \includegraphics[height = 1.4cm]{FiguresPDF/IC-c8.pdf}     &   \raisebox{18pt}{$ \left\{ \begin{array}{ll}  \vspace{1mm} a \rightarrow Q_{1,F}  & \\ b \rightarrow Q_{2,INT} &  \end{array} \right. $} &  \raisebox{18pt}{$16$}    &    \includegraphics[height = 1.4cm]{FiguresPDF/IC-c16.pdf}     &   \raisebox{18pt}{$ \left\{ \begin{array}{ll} \vspace{1mm} a \rightarrow Q_{1,INT}  & \\ b \rightarrow Q_{2,INT}  &  \end{array} \right. $} \\

\hline

\end{tabular}
\label{table:upgrade}
\end{table*}

Here, we describe what happens to the status of a bit in $Q_{1,INT}$ if either of the $16$ channel realizations occur. The description for a bit in $Q_{2,INT}$ is very similar and is summarized in Table~\ref{table:upgrade}. Consider a bit ``$a$'' in $Q_{1,INT}$. At each time instant, $16$ possible cases may occur:
\begin{itemize}

\item Cases 9,10,11,12,13, and 16: In these cases, it is easy to see that no change occurs in the status of bit $a$.

\item Case 6: In this case, bit $a$ is delivered to both receivers and hence, no further transmission is required. Therefore, it joins $Q_{1,F}$.

\item Case 8: In this case, bit $a$ is available at ${\sf Rx}_2$ but it is interfered at ${\sf Rx}_1$ by bit $b$. However, in Case $8$ no change occurs for the bits in $Q_{2,INT}$. Therefore, since bit $b$ will be retransmitted until it is provided to ${\sf Rx}_1$, no retransmission is required for bit $a$ and it joins $Q_{1,F}$.

\item Cases 14 and 15: If either of these cases occur, bit $a$ becomes available at ${\sf Rx}_2$ and is needed at ${\sf Rx}_1$. Thus, we update the status of such bits to $Q_{1 \rightarrow 1|2}$.

\item Cases 1,2,3,4,5, and 7: If either of these cases occur, we upgrade the status of bit $a$ to the opportunistic state $Q_{1,OP}$, meaning that from now on bit $a$ has to be provided to either ${\sf Rx}_2$ or both receivers such that it causes no further interference. For instance, if Case $2$ occurs, providing bit $a$ to both receivers is suffiecient to decode the simultaneously transmitted bits.

\end{itemize}

%
%
%
%
%
%
%

Phase $2$ goes on for 
\begin{align}
\left( 1 - \left[ p^3q + 2 pq^2 + q^4 \right] \right)^{-1} p q^2 m + 2 m^{\frac{2}{3}}
\end{align}
time instants, and if at the end of this phase either of the states $Q_{i,INT}$ is not empty, we declare error type-I and halt the transmission. 

Assuming that the transmission is not halted, since transition of a bit to this state is distributed as independent Bernoulli RV, upon completion of Phase $2$, we have 
\begin{align}
& \mathbb{E}[N_{1,{\sf C}_1}] = p^3q m \nonumber \\
&~+ \frac{\sum_{j=1,3}{\Pr\left( \mathrm{Case~j} \right)}}{1 - \sum_{i=9,10,11,12,13,16}{\Pr\left( \mathrm{Case~i} \right)}} ( pq^2 m + m^{2/3} ) \nonumber \\ 
&~= p^3q m + \left( 1 - \left[ p^3q + 2 pq^2 + q^4 \right] \right)^{-1} p^3 ( pq^2 m + m^{2/3} ), \nonumber \\ 
& \mathbb{E}[N_{1,OP}] = p^2q^2 m \nonumber \\
&~+ \left( 1 - \left[ p^3q + 2 pq^2 + q^4 \right] \right)^{-1} pq ( pq^2 m + m^{2/3} ), \nonumber \\
&\mathbb{E}[N_{2,OP}] = p^2q^2 m \nonumber \\
&~+ \left( 1 - \left[ p^3q + 2 pq^2 + q^4 \right] \right)^{-1} pq (1+p^2) ( pq^2 m + m^{2/3} ), \nonumber 
\end{align}
\begin{align}
\label{eq:expectedvalues2}
&\mathbb{E}[N_{i \rightarrow i|\bar{i}}] = q^3 m \\
&~+ \left( 1 - \left[ p^3q + 2 pq^2 + q^4 \right] \right)^{-1} pq^2 \left( pq^2 + m^{2/3} \right),~ i = 1,2. \nonumber
\end{align}

The transmission strategy will halt and an error (that we refer to as error type-II) will occur, if any of the following events happens.
\begin{align}
\label{eq:errortypeIcornerC2}
& N_{1,{\sf C}_1} > \mathbb{E}[N_{1,{\sf C}_1}] + m^{\frac{2}{3}} \overset{\triangle}= n_{1,{\sf C}_1}; \nonumber \\
& N_{i,OP} > \mathbb{E}[N_{i,OP}] + m^{\frac{2}{3}} \overset{\triangle}= n_{i,OP}, \quad i=1,2; \nonumber \\
& N_{i \rightarrow i|\bar{i}} > \mathbb{E}[N_{i \rightarrow i|\bar{i}}] + m^{\frac{2}{3}} \overset{\triangle}= n_{i \rightarrow i|\bar{i}}, \quad i=1,2.
\end{align}

%
%
%
%

Using Chernoff-Hoeffding bound, we can show that the probability of errors of types I and II decreases exponentially with $m$.

Again, at the end of Phase $2$, we add $0$'s (if necessary) in order to make queues $Q_{1,{\sf C}_1}$, $Q_{i,OP}$, and $Q_{i \rightarrow i|\bar{i}}$ of size equal to $n_{1,{\sf C}_1}$, $n_{i,OP}$, and $n_{i \rightarrow i|\bar{i}}$ respectively as defined in (\ref{eq:errortypeIcornerC2}), $i=1,2$. For the rest of this appendix, we assume that Phase 2 is completed and no error has occurred.

Note that ${\sf Tx}_2$ initially had $m$ fresh data bits but during Phase $1$ it only communicated $m_1$ of them. The rest of those bits will be transmitted during Phase $3$ as described below.


\noindent{\bf Phase 3} [uncategorized transmission vs interference management]: During Phase $3$, ${\sf Tx}_1$ (the secondary user) communicates $\frac{q}{1+q} (p-q^2)m$ bits from states $Q_{1,{\sf C}_1}$ and $Q_{1,OP}$ at a rate such that both receivers can decode them at the end of Phase $3$, regardless of the transmitted signal of ${\sf Tx}_2$. In fact, at ${\sf Rx}_1$, we have 
\begin{align}
\Pr \left[ G_{11}[t] = 1, G_{21}[t] = 0 \right] = pq,
\end{align}
and at ${\sf Rx}_2$, we have
\begin{align}
\Pr \left[ G_{22}[t] = 0, G_{12}[t] = 1 \right] = pq.
\end{align}

Hence, using the results of~\cite{Elias}, we know that given any $\epsilon, \delta > 0$, ${\sf Tx}_1$ can use a random code of rate $pq - \delta$ to encode $\frac{q}{1+q} (p-q^2)m$ bits from states $Q_{1,{\sf C}_1}$ and $Q_{1,OP}$, and transmits them such that both receivers can decode the transmitted message with error probability less than or equal to  $\epsilon$ for sufficiently large block length (${\sf Tx}_1$ picks bits from $Q_{1,{\sf C}_1}$ and if this state becomes empty it starts picking from the bits in $Q_{1,OP}$). Since ${\sf Rx}_2$ can decode the transmitted signal of ${\sf Tx}_1$ in this phase, we can assume that the encoded bits of ${\sf Tx}_1$, do not create any new interference during Phase 3.

We now describe what ${\sf Tx}_2$ does during Phase 3. At the beginning of Phase $3$, we assume that the bits $b_{m_1+1}, b_{m_1+2}, \ldots, b_m$ at ${\sf Tx}_2$ are in state $Q_{2 \rightarrow 2}$. At each time instant, ${\sf Tx}_2$ picks a bit from $Q_{2 \rightarrow 2}$ and sends it. This bit will either stay in $Q_{2 \rightarrow 2}$ or a transition occurs as described below.
\begin{itemize}

\item Cases $1,2,3,4,9,10,11,$ and $12$: In these cases the direct link from ${\sf Tx}_2$ to ${\sf Rx}_2$ is on. Therefore, since at the end of block (assuming large enough block length), we can decode and remove the transmitted signal of ${\sf Tx}_1$, the transmitted bit of ${\sf Tx}_2$ leaves $Q_{2 \rightarrow 2}$ and joins $Q_{2,F}$.

\item Cases $7,8,13,$ and $15$: In these cases (assuming the transmitted signal of ${\sf Tx}_1$ can be removed), the transmitted bit of ${\sf Tx}_2$ becomes available at ${\sf Rx}_1$ while it is required at ${\sf Rx}_2$. Thus, the transmitted bit of ${\sf Tx}_2$ leaves $Q_{2 \rightarrow 2}$ and joins $Q_{2 \rightarrow 2|1}$.

\item Cases $5,6,14,$ and $16$: In these cases, no change hanppens in the status of the transmitted bit from ${\sf Tx}_2$.

\end{itemize}

Phase $3$ goes on for 
\begin{align}
\frac{(p-q^2)}{(1-q^2)} m + m^{\frac{2}{3}}
\end{align}
time instants and if at the end of this phase there is a bit left in $Q_{2 \rightarrow 2}$ or an error occurs in decoding the transmitted signal of ${\sf Tx}_1$, we declare error type-I and halt the transmission. Note that during Phase $3$, the number of bits in $Q_{1 \rightarrow 1|2}$ and $Q_{2,OP}$ remain unchanged.

Assuming that the transmission is not halted, since transition of a bit to this state is distributed as independent Bernoulli RV, upon completion of Phase $2$, we have 
\begin{align}
\mathbb{E}[N_{1,{\sf C}_1}] & = \left[ p^3q m + \left( 1 - \left[ p^3q + 2 pq^2 + q^4 \right] \right)^{-1} \right. \nonumber \\
& \left. \times p^3 ( pq^2 m + m^{2/3} ) - \frac{q}{1+q} (p-q^2)m \right]^+, \nonumber \\
\mathbb{E}[N_{1,OP}] & = p^2q^2 m + \left( 1 - \left[ p^3q + 2 pq^2 + q^4 \right] \right)^{-1} \\
& \times pq ( pq^2 m + m^{2/3} ) - \left[ \frac{q}{1+q} (p-q^2)m - p^3q m \right. \nonumber \\
& \left. - \left( 1 - \left[ p^3q + 2 pq^2 + q^4 \right] \right)^{-1} p^3 ( pq^2 m + m^{2/3} ) \right]^+\hspace{-1mm}. \nonumber 
\end{align}


For $\left( 3 - \sqrt{5} \right)/2 \leq p \leq 1$, $\mathbb{E}[N_{1,OP}]$ is non-negative. The transmission strategy will halt and an error (that we refer to as error type-II) will occur, if any of the following events happens.
\begin{enumerate}

\item $N_{1,{\sf C}_1} > \mathbb{E}[N_{1,{\sf C}_1}] + m^{\frac{2}{3}}$;

\item $N_{1,OP} > \mathbb{E}[N_{1,OP}] + m^{\frac{2}{3}}$;

\item $N_{2 \rightarrow 2|1} > \mathbb{E}[N_{2 \rightarrow 2|1}] + m^{\frac{2}{3}}$.

\end{enumerate}

Using Chernoff-Hoeffding bound, we can show that the probability of errors of types I and II decreases exponentially with $m$.


\noindent{\bf Phase 4} [delivering interference-free bits and interference management]: In Phase $4$, ${\sf Tx}_1$ will communicate all the bits in $Q_{1 \rightarrow 1|2}$. However, it is possible to create XOR of these bits with bits in $Q_{1,OP}$ in order to create bits of common interest. To do so, we first encode the bits in these states using the results of~\cite{Elias}, and then we create the XOR of the encoded bits. On the other hand, ${\sf Tx}_2$ will do the same to part of the bits in $Q_{2 \rightarrow 2|1}$ and $Q_{2,OP}$.

More precisely, for any $\epsilon, \delta > 0$, ${\sf Tx}_1$ encodes all the bits in $Q_{1 \rightarrow 1|2}$ at rate $p - \delta$ using random coding scheme of~\cite{Elias}. Similarly, ${\sf Tx}_1$ encodes 
\begin{align}
\label{eq:betadefepdel}
q^4 +  \left( 1 - \left[ p^3q + 2 pq^2 + q^4 \right] \right)^{-1} p^2q^5
\end{align}
bits from $Q_{1,OP}$ at rate $pq - \delta$. Then ${\sf Tx}_1$ will communicate the XOR of these encoded bits. 

During Phase 3, ${\sf Tx}_2$ encodes same number of the bits as in $Q_{1 \rightarrow 1|2}$ from $Q_{2 \rightarrow 2|1}$ at rate $p - \delta$ and $$q^4 +  \left( 1 - \left[ p^3q + 2 pq^2 + q^4 \right] \right)^{-1} p^2q^5$$ bits from $Q_{2,OP}$ at rate $pq - \delta$ using random coding scheme of~\cite{Elias}. Then ${\sf Tx}_2$ will communicate the XOR of these encoded bits.

Since ${\sf Rx}_2$ already has access to the bits in $Q_{1 \rightarrow 1|2}$ and $Q_{2,OP}$, it can remove their contribution from the received signals. Then for sufficiently large block length, ${\sf Rx}_2$ can decode the transmitted bits from $Q_{1,OP}$ with decoding error probability less than or equal to $\epsilon$. After decoding and removing this part, ${\sf Rx}_2$ can decode the encoded bits from $Q_{2 \rightarrow 2|1}$ with decoding error probability less than or equal to $\epsilon$.

On the other hand, since ${\sf Rx}_1$ already has access to the bits in $Q_{2 \rightarrow 2|1}$ and $Q_{1,OP}$, it can remove their contribution from the received signals. Then for sufficiently large block length, ${\sf Rx}_1$ can decode the transmitted bits from $Q_{2,OP}$ with decoding error probability less than or equal to $\epsilon$. Finally, after decoding and removing this part, ${\sf Rx}_1$ can decode the encoded bits from $Q_{1 \rightarrow 1|2}$ with decoding error probability less than or equal to $\epsilon$.

Phase $4$ goes on for 
\begin{align}
\frac{\left[ q^3 m + \left( 1 - \left[ p^3q + 2 pq^2 + q^4 \right] \right)^{-1} pq^2 \left( pq^2 + m^{2/3} \right) \right]}{\left( p - \delta \right)}
\end{align}
time instants. If an error occurs in decoding any of the encoded signals in Phase $4$, we consider it as error and we halt the transmission strategy.

At the end of Phase $4$, $Q_{1 \rightarrow 1|2}$ becomes empty. Define
{\small \begin{align}  
& \beta = \\
& \frac{(pq - \delta) \left[ q^3 m + \left( 1 - \left[ p^3q + 2 pq^2 + q^4 \right] \right)^{-1} pq^2 \left( pq^2 + m^{2/3} \right) \right]}{\left( p - \delta \right)}, \nonumber 
\end{align}}
note that as $\epsilon, \delta \rightarrow 0$, $\beta$ agrees with the expression given in (\ref{eq:betadefepdel}).

Upon completion of Phase $4$, we have 
\begin{align}
& \mathbb{E}[N_{1,OP}] = p^2q^2 m \nonumber \\
&~+ \left( 1 - \left[ p^3q + 2 pq^2 + q^4 \right] \right)^{-1} pq ( pq^2 m + m^{2/3} ) \nonumber \\
&~- \left[ \frac{q}{1+q} (p-q^2)m - p^3q m - \left( 1 - \left[ p^3q + 2 pq^2 + q^4 \right] \right)^{-1} \right. \nonumber \\
&~\left. \times p^3 ( pq^2 m + m^{2/3} ) \right]^+ - \frac{q}{1+q} (p-q^2)m - \beta, \nonumber \\
& \mathbb{E}[N_{2,OP}] = p^2q^2 m + \left( 1 - \left[ p^3q + 2 pq^2 + q^4 \right] \right)^{-1} \nonumber \\
&~ \times pq (1+p^2) ( pq^2 m + m^{2/3} ) - \beta, \nonumber
\end{align}
\begin{align}
\hspace{-25mm} \mathbb{E}[N_{2 \rightarrow 2|1}] = \frac{pq}{(1-q^2)} \left( p - q^2 \right)m.
\end{align}


The transmission strategy will halt and an error of type-II will occur, if any of the following events happens.
\begin{enumerate}
\item $N_{i,OP} > \mathbb{E}[N_{i,OP}] + m^{\frac{2}{3}}$, $i=1,2$;

\item $N_{2 \rightarrow 2|1} > \mathbb{E}[N_{2 \rightarrow 2|1}] + m^{\frac{2}{3}}$.

\end{enumerate}

Using Chernoff-Hoeffding bound, we can show that the probability of errors of type II decreases exponentially with $m$. Furthermore, the probability that an error occurs in decoding any of the encoded signals in Phase $4$ can be made arbitrary small as $m \rightarrow \infty$.


\noindent{\bf Phase 5} [delivering interference-free bits and interference management]: In Phase 5, the transmitters will communicate the remaining bits in $Q_{1,OP}$, $Q_{1,{\sf C}_1}$, $Q_{2 \rightarrow 2|1}$, and $Q_{2,OP}$. ${\sf Tx}_1$ will communicate all the bits in $Q_{1,OP}$ and $Q_{1,{\sf C}_1}$ such that for sufficiently large block length, both receivers can decode them with arbitrary small error. On the other hand, ${\sf Tx}_2$ will communicate the bits in $Q_{2 \rightarrow 2|1}$ and $Q_{2,OP}$ similar to Phase 4 with one main difference. Since both receivers can completely remove the contribution of ${\sf Tx}_1$ at the end of the block, ${\sf Tx}_2$ can send the bits in $Q_{2,OP}$ at a higher rate of $p$ as opposed to $pq$ during Phase 4.

More precisely, for any $\epsilon, \delta > 0$, ${\sf Tx}_1$ using random coding scheme of~\cite{Elias}, encodes all the bits in $Q_{1,OP}$ and $Q_{1,{\sf C}_1}$ at rate $pq - \delta$ and communicates them. On the other hand, ${\sf Tx}_2$ using random coding, encodes all the bits in $Q_{2 \rightarrow 2|1}$ and all bits in $Q_{2,OP}$ at rate $p - \delta$. Then, ${\sf Tx}_2$ communicates the XOR of its encoded bits. 

Since ${\sf Rx}_1$ already has access to the bits in $Q_{2 \rightarrow 2|1}$ and $Q_{1,OP}$, it can remove the corresponding parts of the transmitted signals. Then for sufficiently large block length, ${\sf Rx}_1$ can decode the transmitted bits from $Q_{1,{\sf C}_1}$ and $Q_{2,OP}$ with decoding error probability less than or equal to $\epsilon$.

Finally, since ${\sf Rx}_2$ already has access to the bits in $Q_{2,OP}$, it can remove the corresponding part of the transmitted signal. Then for sufficiently large block length, ${\sf Rx}_2$ can decode the transmitted bits from $Q_{1,OP}$ and $Q_{1,{\sf C}_1}$ with decoding error probability less than or equal to $\epsilon$. After decoding and removing this part, ${\sf Rx}_2$ can decode the encoded bits from $Q_{2 \rightarrow 2|1}$ with decoding error probability less than or equal to $\epsilon$.

Phase $5$ goes on for 
\begin{align}
\left[  \frac{pq}{(1-q^2)} \left( p - q^2 \right)m + 2 m^{2/3} \right]/\left( p - \delta \right)
\end{align}
time instants. If an error occurs in decoding any of the encoded signals in Phase $5$, we consider it as error and halt the transmission strategy. It is straight forward to verify that at the end of Phase $5$, if the transmission is not halted, all states are empty and all bits are successfully delivered.

The probability that the transmission strategy halts at any point can be bounded by the summation of error probabilities of types I-II and the probability that an error occurs in decoding the encoded bits. Using Chernoff-Hoeffding bound and the results of~\cite{Elias}, we can show that the probability that the transmission strategy halts at any point approaches zero for $\epsilon, \delta \rightarrow 0$ and $m \rightarrow \infty$. Moreover, the total transmission requires
\begin{align}
\frac{1}{p} m + 6 m^{2/3}
\end{align}
time instants. Thus, ${\sf Tx}_1$ achieves a rate of $pq(1+q)$ while ${\sf Tx}_2$ achieves a rate of $p$. 

This completes the achievability proof of Theorem~\ref{THM:IC-DelayedCSIT}.


\section{Proof of Lemma~\ref{lemma:multicast}}
\label{Appendix:multicast}


In this appendix, we provide the proof of Lemma~\ref{lemma:multicast}. We first derive the outer-bound and then we describe the achievability. The outer-bound on $R_i$ is the same as in Section~\ref{sec:ConvInst}. 

Suppose there are encoders and decoders at the transmitters and receivers respectively, such that each receiver can decode both messages with arbitrary small decoding error probability as the block length goes to infinity. We have
\begin{align}
n & (R_1 + R_2 - \epsilon_n ) \nonumber \\
& \overset{(a)}\leq I(W_1,W_2;Y_1^n|G^n) \nonumber \\
& = H(Y_1^n|G^n) - H(Y_1^n|W_1,W_2,G^n) \nonumber \\
& \overset{(b)}= H(Y_1^n|G^n) \nonumber \\
& \overset{(c)}\leq \sum_{t=1}^n{H(Y_1[t]|G^n)} \leq \left( 1 - q^2 \right) n,
\end{align}
where $\epsilon_n \rightarrow 0$ as $n \rightarrow \infty$; and $(a)$ follows from the fact that the messages and $G^n$ are mutually independent, Fano's inequality, and the fact that ${\sf Rx}_1$ should be able to decode both messages; $(b)$ holds since the received signal $Y_1^n$ is a deterministic function of $\hbox{W}_1$, $\hbox{W}_2$, and $G^n$; and $(c)$ follows from the fact that conditioning reduces entropy. Dividing both sides by $n$ and let $n \rightarrow \infty$, we get 
\begin{align}
R_1 + R_2 \leq 1 - q^2.
\end{align}

Below, we provide the achievability proof of Lemma~\ref{lemma:multicast}. Let $\hbox{W}_i \in \{ 1,2,\ldots,2^{nR_i}\}$ denote the message of user $i$. 

In~\cite{Elias}, it has been shown that for a binary erasure channel with success probability $p$, and for any $\epsilon,\delta > 0$, as long as the communication  rate is less than or equal to $p - \delta$, we can have decoding error probability less than or equal to $\epsilon$.

Codebook generation is as follows. Transmitter $i$ creates $2^{nR_i}$ ($R_1 = p - \delta$ and $R_2 = pq - \delta$) independent codewords where each entry of the codewords is an i.i.d. Bernoulli $0.5$ RV. For message index $j$, transmitter $i$ will send the $j^{th}$ codeword. Note that we can view the channel from ${\sf Tx}_2$ to ${\sf Rx}_1$ as a binary erasure channel with success probability $pq$ (whenever $G_{11}[t] = 0$ and $G_{21}[t] = 1$, we get a clean observation of $X_2[t]$). Therefore, since $R_2 = pq - \delta$, ${\sf Rx}_1$ can decode $\hbox{W}_2$ with arbitrary small decoding error probability as $n \rightarrow \infty$ and remove $X_2^n$ from its received signal. After removing $X_2^n$,  we can view the channel from ${\sf Tx}_1$ to ${\sf Rx}_1$ as a binary erasure channel with success probability $p$ (whenever $G_{11}[t] = 1$, we get a clean observation of $X_1[t]$). Therefore, since $R_1 = p - \delta$, ${\sf Rx}_1$ can decode $\hbox{W}_1$ with arbitrary small decoding error probability as $n \rightarrow \infty$. Similar argument holds for ${\sf Rx}_2$. This completes the achievability proof of corner point
\begin{align}
\left( R_1, R_2 \right) = \left( p, pq \right).
\end{align}

Similarly, we can achieve corner point
\begin{align}
\left( R_1, R_2 \right) = \left( pq, p \right).
\end{align}

Therefore with time sharing, we can achieve the entire region as described in Lemma~\ref{lemma:multicast}.


\section{Achievability Proof of Theorem~\ref{THM:IC-DCSITFB}: Corner Point $\left( 1 - q^2, 0 \right)$}
\label{Appendix:cornerDelayedOFB}


By symmetry, it suffices to describe the achievability strategy for corner ponit 
\begin{align}
\left( R_1, R_2 \right) = \left( 1 - q^2, 0 \right),
\end{align}
when transmitters have delayed knowledge of channel state information and noiseless output feedback links are available from the receivers to the transmitters. 

Our achievability strategy is carried on over $b+1$ communication blocks each block with $n$ time instants. Transmitter one communicates fresh data bits in the first $b$ blocks and the final block is to help receiver one decode its corresponding bits. Transmitter and receiver two act as a relay to facilitate the communication between transmitter and receiver one. At the end, using our scheme we achieve rate tuple $\frac{b}{b+1} \left( 1 - q^2, 0 \right)$ as $n \rightarrow \infty$. Finally, letting $b \rightarrow \infty$, we achieve the desired corner point.

Let $\hbox{W}^j_1$ be the message of ${\sf Tx}_1$ in block $j$, $j=1,2,\ldots,b$. We assume $\hbox{W}^j_1 = a^j_1, a^j_2, \ldots, a^j_{m}$. We set
\begin{align}
n = \left( 1 - q^2 \right)^{-1}m + m^{2/3}.
\end{align}


{\bf Achievability strategy for block $1$}: At the beginning of the communication block, we assume that the bits at ${\sf Tx}_1$ are in queue (or state) $Q^1_{1 \rightarrow 1}$. At each time instant $t$, ${\sf Tx}_1$ sends out a bit from $Q^1_{1 \rightarrow 1}$, and this bit will leave this queue if at least one of the outgoing links from ${\sf Tx}_1$ was equal to $1$ at the time of transmission. On the other hand, ${\sf Tx}_2$ remains silent during the first communication block. If at the end of the communication block, queue $Q^1_{1 \rightarrow 1}$ is not empty, we declare error type-I and halt the transmission. 

At the end of first block, using output feedback links, transmitter two has access to the bits of ${\sf Tx}_1$ communicated in the first block. More precisely, ${\sf Tx}_2$ has access to the bits of ${\sf Tx}_1$ communicated in Cases $11,12,14,$ and $15$ during the first communication block. Note that the bits communicated in these cases are available at ${\sf Rx}_2$ and have to be provided to ${\sf Rx}_1$. Transmitter two transfers these bits to $Q^1_{2 \rightarrow 1|2}$. 

Assuming that the transmission is not halted, let $N^1_{2 \rightarrow 1|2}$ denote the number of bits in queue $Q^1_{2 \rightarrow 1|2}$. The transmission strategy will be halted and an error type-II will occur, if $N^1_{2 \rightarrow 1|2} > \mathbb{E}[N^1_{2 \rightarrow 1|2}] + pq m^{\frac{2}{3}}$. From basic probability, we know that
\begin{align}
\mathbb{E}[N^1_{2 \rightarrow 1|2}] & = \frac{\sum_{j=11,12,14,15}{\Pr\left( \mathrm{Case~j} \right)}}{1 - \sum_{i=9,10,13,16}{\Pr\left( \mathrm{Case~i} \right)}} m  \nonumber \\
& = (1-q^2)^{-1} pq m.
\end{align}

Using Chernoff-Hoeffding bound, we can show that the probability of errors of types I and II and decreases exponentially with $m$.

{\bf Achievability strategy for block $j,$ $j=2,3,\ldots,b$}: In the communication block $j$, ${\sf Tx}_2$ treats the bits in $Q^{j-1}_{2 \rightarrow 1|2}$ as its message and it uses a random code of rate $pq - \delta$ to transmit them. Note that the channel from ${\sf Tx}_2$ to ${\sf Rx}_1$ can be modeled as a point-to-point erasure channel (any time $G_{11}[t] = 1$ or $G_{21}[t] = 0$, we consider an erasure has taken place). Hence from~\cite{Elias}, we know that for any $\epsilon, \delta > 0$ and sufficiently large block length, a rate of $pq - \delta$ is achievable from ${\sf Tx}_2$ to ${\sf Rx}_1$ with decoding error probability less than or equal to $\epsilon$.  Note that at rate $pq - \delta$, both receivers will be able to decode and hence remove the transmitted signal of ${\sf Tx}_2$ at the end of communication block. If an error occurs in decoding the transmitted signal of ${\sf Tx}_2$, we consider it as error and halt the transmission strategy.

On the other hand, the transmission strategy for ${\sf Tx}_1$ is the same as block $1$ for the first $b$ blocks (all but the last block). At the end of communication block $j$, using output feedback links, transmitter two has access to the bits of ${\sf Tx}_1$ communicated in Cases $11,12,14,$ and $15$ during the communication block $j$. Transmitter two transfers these bits to $Q^j_{2 \rightarrow 1|2}$. If at the end of the communication block, queue $Q^j_{1 \rightarrow 1}$ is not empty, we declare error type-I and halt the transmission. 

Assuming that the transmission is not halted, let $N^j_{2 \rightarrow 1|2}$ denote the number of bits in queue $Q^j_{2 \rightarrow 1|2}$. The transmission strategy will be halted and an error type-II will occur, if $N^j_{2 \rightarrow 1|2} > \mathbb{E}[N^j_{2 \rightarrow 1|2}] + pq m^{\frac{2}{3}}$. From basic probability, we know that
\begin{align}
\mathbb{E}[N^j_{2 \rightarrow 1|2}] & = \frac{\sum_{j=11,12,14,15}{\Pr\left( \mathrm{Case~j} \right)}}{1 - \sum_{i=9,10,13,16}{\Pr\left( \mathrm{Case~i} \right)}} m  \nonumber \\
& = (1-q^2)^{-1} pq m.
\end{align}

Using Chernoff-Hoeffding bound, we can show that the probability of errors of types I and II and decreases exponentially with $m$.

{\bf Achievability strategy for block $b+1$}: Finally in block $b+1$, no new data bit is transmitted and ${\sf Tx}_2$ only communicates the bits of ${\sf Tx}_1$ communicated in the previous block in Cases $11,12,14,$ and $15$ as described before.

We can show that the probability that the transmission strategy halts at any point approaches zero as $m \rightarrow \infty$.

{\bf Decoding}: At the end of block $j+1$, ${\sf Rx}_1$ decodes the transmitted message of ${\sf Tx}_2$ in block $j+1$ and removes it from the received signal. Together with the bits it has obtained during block $j$, it can decode message $\hbox{W}^{j}_1$. Using similar idea, ${\sf Rx}_2$ uses backward decoding to cancel out interfernce in the previous blocks to decode all messages.

This completes the achievability proof for corner ponit 
\begin{align*}
\left( R_1, R_2 \right) = \left( 1 - q^2, 0 \right).
\end{align*}

\bibliographystyle{ieeetr}
\bibliography{bib_nsi}

\begin{IEEEbiographynophoto} {Alireza Vahid} received the B.Sc. degree in electrical engineering from Sharif University of Technology, Tehran, Iran, in 2009, and the M.Sc. degree and Ph.D. degree in electrical and computer engineering both from Cornell University, Ithaca, NY, in 2012 and 2014 respectively. As of September 2014, he is a postdoctoral scholar at Information Initiative at Duke University, Durham, NC. His research interests include information theory and wireless communications, statistics and machine learning.

He has received the Director's Ph.D. Teaching Assistant Award in 2010 from the school of electrical and computer engineering, Cornell University, and Jacobs Scholar Fellowship in 2009. He has also received Qualcomm Innovation Fellowship in 2013 for his research on ``Collaborative Interference Management''.
\end{IEEEbiographynophoto}

\begin{IEEEbiographynophoto} {Mohammad Ali Maddah-Ali} received the B.Sc. degree from Isfahan University of Technology, Isfahan, Iran,  the M.A.Sc. degree from the University of Tehran, Tehran, Iran,  and PhD degree from University of Waterloo, Waterloo, ON, Canada,  all in electrical engineering. Then he joined the Wireless Technology Laboratories, Nortel Networks, Ottawa, ON, Canada, for one year.  From January 2008 to August 2010, he was a Postdoctoral Fellow at the Department of Electrical Engineering and Computer Sciences in the University of California at Berkeley. Since September 2010, he has been at Bell Laboratories, Alcatel-Lucent, Holmdel, NJ, as a communication network research scientist. 
\end{IEEEbiographynophoto}

\begin{IEEEbiographynophoto} {Amir Salman Avestimehr} received the B.S. degree in electrical engineering from Sharif University of Technology, Tehran, Iran, in 2003 and the M.S. degree and Ph.D. degree in electrical engineering and computer science, both from the University of California, Berkeley, in 2005 and 2008, respectively. He is currently an Associate Professor at the EE department of University of Southern California, Los Angeles, CA. He was also a postdoctoral scholar at the Center for the Mathematics of Information (CMI) at the California Institute of Technology, Pasadena, in 2008. His research interests include information theory, the theory of communications, and their applications.

Dr. Avestimehr has received a number of awards for research and teaching, including the Communications Society and Information Theory Society Joint Paper Award in 2013, the Presidential Early Career Award for Scientists and Engineers (PECASE) in 2011, the Michael Tien 72 Excellence in Teaching Award in 2012, the Young Investigator Program (YIP) award from the U. S. Air Force Office of Scientific Research in 2011, the National Science Foundation CAREER award in 2010, and the David J. Sakrison Memorial Prize in 2008. He is currently an Associate Editor for the IEEE Transactions on Information Theory. He has also been a Guest Associate Editor for the IEEE Transactions on Information Theory Special Issue on Interference Networks and General Co-Chair of the 2012 North America Information Theory Summer School and the 2012 Workshop on Interference Networks.
\end{IEEEbiographynophoto}

\end{document}